\documentclass[12pt,english]{IEEEtran}

\usepackage{hyperref}
\usepackage{setspace}
\usepackage{microtype}
\usepackage{graphicx}
\RequirePackage[numbers]{natbib}
\usepackage{amssymb}
\usepackage{amsmath}
\usepackage{bigints}
\usepackage{dsfont}
\usepackage{algorithm}
\usepackage{algpseudocode}
\usepackage{color}
\usepackage{lipsum}
\usepackage{amsmath}
\usepackage{amssymb}
\usepackage{mathtools}
\usepackage{amsthm}

\newtheorem{thrm}{Theorem}
\newtheorem{lemm}{Lemma}
\newtheorem{deffn}{Definition}
\newtheorem{propo}{Proposition}
\newtheorem{assump}{Assumption}

\def \F {\mathbb{F}}
\def \R {\mathbb{R}}
\def \C {\mathbb{C}}
\def \ssconstant {\beta}
\def \diff {\mathrm{d}}
\def \statedim {d_X}
\def \controldim {d_U}
\def \noisedim {d_W}
\def \totaldim {d}

\def \policy {\boldsymbol{\pi}}
\def \optimalpolicy {\boldsymbol{\pi^\star}}
\def \Qmat {Q}
\def \Rmat {R}
\def \truth {A_\star,B_\star}

\def \stabradii {\rho}
\def \auxA {\widetilde{A}}
\def \auxQ {\widetilde{Q}}
\def \auxG {\widetilde{G}}
\def \auxB {\widetilde{B}}
\def \auxR {\widetilde{R}}
\def \auxP {\widetilde{P}}
\def \auxpolicy {\widetilde{\pi}}
\def \RiccUpperBound {\zeta}
\def \curve {\varphi}

\def \algoconstant {\boldsymbol{\mathrm{\omega}}_{\policy}}
\def \learnconstant {\boldsymbol{\mathrm{\omega}}_{\mathcal{E}}}
\def \regconstant {\boldsymbol{\mathrm{\omega}}_{\mathcal{R}}}
\def \variance {\sigma}

\newcommand{\Mnorm}[2]{{\left\vert\kern-0.35ex\left\vert\kern-0.35ex\left\vert #1 
		\right\vert\kern-0.35ex\right\vert\kern-0.35ex\right\vert}}
\newcommand{\norm}[2]{{\left\vert\kern-0.25ex\left\vert #1 
		\right\vert\kern-0.25ex\right\vert}}
\newcommand{\eigmax}[1]{\boldsymbol{\lambda}_{\max} \left( #1 \right)}
\newcommand{\eigmin}[1]{\boldsymbol{\lambda}_{\min} \left( #1 \right)}
\newcommand{\tr}[1]{\boldsymbol{\mathrm{tr}} \left( #1 \right)}

\newcommand{\PP}[1]{%
	\mathbb{P}{\ifthenelse{ \equal{#1}{} }{}{\left(#1\right)}}
}%
\newcommand{\E}[1]{\mathbb{E}\left[#1\right]}

\newcommand{\regret}[2]{\boldsymbol{\mathcal{R}}_{#2} \left(#1\right)}

\newcommand{\RiccSol}[1]{\mathcal{K}\left(#1\right)}
\newcommand{\Optgain}[1]{\mathcal{L}\left(#1\right)}
\newcommand{\Gainmat}[1]{L_{#1}}
\newcommand{\avecost}[2]{{\mathcal{J}}_{#1}}
\newcommand{\optavecost}[1]{\overline{\mathcal{J}}^\star }

\newcommand{\optdisccost}[1]{\mathcal{J}_\gamma^\star}
\newcommand{\instantcost}[2]{c_{#2}\left(#1\right)}
\newcommand{\order}[1]{ \mathcal{O} \left(#1\right)}
\newcommand{\orderlog}[1]{\widetilde{\mathcal{O}} \left(#1\right)}

\newcommand{\Amat}[1]{A_{#1}}
\newcommand{\Bmat}[1]{B_{#1}}
\newcommand{\CLmat}[1]{D_{#1}}
\newcommand{\estpara}[1]{{A}_{#1},{B}_{#1}}
\newcommand{\estA}[1]{{A}_{#1}}
\newcommand{\estB}[1]{{B}_{#1}}
\newcommand{\estD}[1]{{D}_{#1}}
\newcommand{\MatOpAve}[2]{\Phi_{#1}\left(#2\right)}

\newcommand{\empiricalcovmat}[1]{V_{#1}}
\newcommand{\filter}[1]{\mathcal{F}_{#1}}

\newcommand{\state}[1]{X_{#1}}

\newcommand{\statetwo}[1]{Y_{#1}}
\newcommand{\action}[1]{U_{#1}}
\newcommand{\itointeg}[4]{\int\limits_{#1}^{#2} {#3} \diff {#4}}

\newcommand*{\BM}[1]{
	W_{\ifthenelse{ \equal{#1}{} }{}{#1}}
}%

\newcommand{\BMcoeff}[1]{C}

\newcommand{\parameter}[1]{A_{#1},B_{#1}}
\newcommand{\normaldist}[2]{\mathcal{N} \left( #1, #2 \right)}
\newcommand{\sigfield}[1]{\sigma \left( #1 \right)}
\newcommand{\Learnerror}[2]{\mathcal{E}_{#1}\left( #2 \right)}
\newcommand{\mosteig}[1]{\boldsymbol{\overline{\lambda}} \left( #1 \right)}

\newcommand{\episodetime}[1]{\gamma^{#1}}
\newcommand{\paraspace}[1]{\mathcal{S}_{#1}}
\newcommand{\regterm}[1]{\boldsymbol{\alpha}_{#1}}
\newcommand{\randommatrix}[1]{\Theta_{#1}}
\newcommand{\erterm}[1]{\Delta_{#1}}

\newcommand{\project}[2]{\Pi_{#1}\left(#2\right)}
\newcommand{\trans}[1]{D}
\newcommand{\mult}[1]{\boldsymbol{\mu}_{#1}}
\newcommand{\diag}[1]{\mathrm{diag}\left(#1\right)}
\newcommand{\valuefunc}[2]{\mathcal{V}_{#1}\left(#2\right)}
\newcommand{\auxvaluefunc}[2]{\widetilde{\mathcal{V}}_{#1}\left(#2\right)}

\newcommand{\manifold}[1]{\mathcal{#1}}

\usepackage{times}
\onecolumn

\begin{document}
\title{\bf Reinforcement Learning Policies in Continuous-Time Linear Systems}
\author{Mohamad Kazem Shirani Faradonbeh, Mohamad Sadegh Shirani Faradonbeh} 
\maketitle

\begin{abstract}%
  Linear dynamical systems that obey stochastic differential equations are canonical models. While optimal control of known systems has a rich literature, the problem is technically hard under model uncertainty and there are hardly any results. We initiate study of this problem and aim to learn (and \emph{simultaneously} deploy) optimal actions for minimizing a quadratic cost function. Indeed, this work is the first that comprehensively addresses the crucial challenge of balancing exploration versus exploitation in continuous-time systems. We present online policies that learn optimal actions fast by carefully randomizing the parameter estimates, and establish their performance guarantees: a regret bound that grows with \emph{square-root of time} multiplied by the \emph{number of parameters}. Implementation of the policy for a flight-control task demonstrates its efficacy. Further, we prove sharp stability results for inexact system dynamics and tightly specify the infinitesimal regret caused by sub-optimal actions. To obtain the results, we conduct a novel eigenvalue-sensitivity analysis for matrix perturbation, establish upper-bounds for comparative ratios of stochastic integrals, and introduce the new method of policy differentiation. Our analysis sheds light on fundamental challenges in continuous-time reinforcement learning and suggests a useful cornerstone for similar problems.
\end{abstract}

\begin{IEEEkeywords}%
  Ito Process; {Regret Bounds}; {Stability Analysis}; 
  Exploration vs Exploitation; Randomized Estimates.
\end{IEEEkeywords}

\section{Introduction}

State-space models are widely-used for decision-making in dynamic environments. A popular such model is the one that represents the continuous-time dynamics of the environment by linear stochastic differential equations. In this setting, the multidimensional state of the system is driven by the control action and the Brownian noise, according to an Ito stochastic differential equation. The range of application areas is extensive, including chemistry, biology, finance, insurance, and engineering~\citep{gillespie2007stochastic,schmidli2007stochastic,pham2009continuous,lawrence2010learning}. 

In many applications, uncertainties about the true dynamics necessitate reinforcement learning policies that adaptively learn optimal actions. Unlike the continuous-time setting, reinforcement learning policies are extensively studied in discrete-time systems. The literature is rich and includes efficient algorithms that use optimism in the face of uncertainty, posterior sampling, or bootstrap~\citep{abbasi2011regret,abeille2018improved,ouyang2019posterior,faradonbeh2019applications,faradonbeh2020optimism,dean2020sample,faradonbeh2020adaptive, faradonbeh2020input}, and regret bounds are shown in the presence of domain knowledge and partial observations~\citep{cassel2020logarithmic,ziemann2020phase,ziemann2020regret,asghari2020regret,lale2020logarithmic}.

On the other hand, the existing literature for continuous-time systems is still immature, mainly due to technical difficulties that will be discussed shortly. Early papers focus on estimation after an infinitely long interaction with the environment~\citep{mandl1988consistency,mandl1989consistency,duncan1990adaptive,duncan1992least,duncan1999adaptive}. Recently, sub-optimal policies with linear regret bounds are proposed and consistency is shown for systems with full-rank input matrices~\citep{caines2019stochastic}. Ensuing papers study offline algorithms for computing the control actions according to a batch of data, using methods such as dynamic programming and entropy regularization
~\citep{doya2000reinforcement,rizvi2018output,wang2020reinforcement,basei2021logarithmic}. 

{However, \emph{online} reinforcement learning policies that can learn optimal actions from a \emph{single} state trajectory without imposing undue costs, are currently unavailable. The existing results are merely asymptotic, require restrictive assumptions, and provide linear regret bounds~\citep{duncan1999adaptive,caines2019stochastic}. The only exception is a recent paper that appeared after the first version of this work~\citep{shirani2022thompson}. A fundamental challenge of online policies is that they need to \emph{simultaneously} minimize the cost and estimate the unknown parameters. These two goals contradict and constitute the exploration-exploitation dilemma; accurate estimation is necessary for optimal decision-making, while sub-optimal actions are required for obtaining accurate estimates. This crux remains unsolved in continuous-time systems as conventional frameworks are incapable of relating exploring actions to estimation accuracy and optimal policies. In fact, the discrete-time analysis fails in Ito processes that the evolution is infinitesimal, and the signal is highly dominated by noise. 

{The main contributions of this paper can be summarized as follows.
	In Algorithm~\ref{algo1}, we propose an efficient online reinforcement learning policy based on randomized estimates of the unknown system matrices. Then, we establish the rates for the error in learning the dynamics matrices, and for the regret that the algorithm incurs. Algorithm~\ref{algo1} is easy to implement, yet it learns the optimal actions \emph{fast} so that its regret at time $T$ is $\order{{T}^{1/2} \log T}$ (Theorem~\ref{BoundsThm}). So, the per-unit-time sub-optimality caused by uncertain system parameters decays with time as $\orderlog{t^{-1/2}}$ under Algorithm~\ref{algo1}, which is the first efficiency result for online policies. Furthermore, we study stability of linear systems for inaccurate system matrices and establish the stabilizability margin (Theorem~\ref{StabThm}). Finally, a sharp regret expression is provided that fully captures sub-optimalities due to inaccuracies in approximating the optimal actions (Theorem~\ref{GeneralRegretThm}). The results provide both the average-case and worst-case analyses, the presented bounds are tight, and the technical assumptions are minimal.} 

To study online reinforcement learning policies, one needs to address the following technical challenges. First, sensitivity analysis of (complex) eigenvalues of matrices with perturbed entries is needed. Further, anti-concentration results on singular values of partially-random matrices are required. Finally, we need to accurately characterize the time-varying sub-optimalities in cost function in terms of model uncertainties. Thus, we develop multiple novel techniques for (i)~\emph{matrix-perturbation} analysis, (ii) spectral properties of \emph{random matrices}, (iii) comparative ratios of \emph{stochastic integrals}, and also (iv) introduce \emph{policy differentiation} to precisely capture the infinitesimal cost of sub-optimal actions. Note that (iii) and (iv) above do \emph{not} appear in discrete-time settings. The former two are different in differential and difference equations, such that the existing literature fails to properly address them in continuous-time systems. We also use various tools from Ito calculus and stochastic analysis, including Hamilton-Jacobi-Bellman partial differential  equations, Ito Isometry, as well as (dominated and martingale) convergence theorems for continuous-time stochastic processes~\citep{yong1999stochastic,oksendal2013stochastic,baldi2017stochastic}.

This paper is organized as follows. In Section~\ref{ProblemStatement}, we discuss the problem and preliminary material. Then, in Section~\ref{StabilitySection}, we study system stability when the control action is designed based on dynamics matrices \emph{other than} the true ones and establish stabilizability guarantees. Next, we examine effects of sub-optimal actions and the regret they cause, in Section~\ref{SubOptSection}. Section~\ref{AlgoSection} contains the randomized-estimates policy of Algorithm~\ref{algo1}, as well as its theoretical analysis showing efficiency. Experimental results are presented in Section~\ref{NumSection}, and the paper is concluded in Section~\ref{ConclusionSection}. 
Because of space limitations, all proofs are delegated to the appendices, as outlined on page \pageref{toc}.

The following notation will be used in this work.
The smallest (largest) eigenvalue of $A$, in magnitude, is $\eigmin{A}$ ($\eigmax{A}$). For $v \in \C^d$, its norm is defined as $\norm{v}{}^2 = {\sum_{i=1}^{d} \left| v_i \right|^2}$. Moreover, we write $\Mnorm{A}{}$ for the operator norm of matrices;   $\Mnorm{A}{} = \sup_{\norm{v}{2}=1} \norm{Av}{2}$,
and $A^\dagger$ for Moore-Penrose generalized inverse. 
The sigma-field generated by the stochastic process $\left\{ \statetwo{s} \right\}_{0 \leq s \leq t}$ is denoted by $\sigfield {Y_{0:t}}$. 
A multivariate normal distribution with mean $\mu$ and covariance matrix $\Sigma$ is shown by $\normaldist{\mu}{\Sigma}$. For $\lambda \in \C$, we use $\Re \left( \lambda \right) , \Im \left( \lambda \right)$ to denote the real and imaginary parts of $\lambda$, respectively. The symbol $\vee$ (resp., $\wedge$) is used to show the maximum (resp., minimum). Finally, $\order{\cdot}$ refers to the order of magnitude.

\section{Problem Statement} \label{ProblemStatement}
We study reinforcement learning policies for a multidimensional Ito stochastic differential equation with unknown drift matrices. That is, the state vector at time $t$ is $\state{t} \in \R^{\statedim}$, which follows
\begin{equation} \label{dynamics}
\diff \state{t} = \left( \Amat{\star} \state{t} + \Bmat{\star} \action{t} \right) \diff t + \BMcoeff{t} \diff \BM{t},
\end{equation}
the vector $\action{t} \in \R^{\controldim}$ is the control action at time $t$, and the disturbance $\left\{ \BM{t}\right\}_{t \geq 0}$ is a standard Brownian motion in a $\noisedim$ dimensional space. Technically, by fixing the probability space $(\Omega, \F, \PP{})$ which is completed by adding the null-sets of $\PP{}$, let all stochastic objects belong to this probability space, and let $\E{\cdot}$ be the expectation with respect to $\PP{}$ (unless otherwise explicitly stated). The Brownian motion $\left\{ \BM{t} \right\}_{t \geq 0}$ starts from the origin, and has continuous sample paths as well as independent normally distributed increments. That is, $\BM{0}=0$, for all $0 \leq t_1 \leq t_2 \leq t_3 \leq t_4$, the vectors $\BM{t_2}-\BM{t_1}$ and $\BM{t_4}-\BM{t_3}$ are statistically independent, and for all non-negative reals $s<t$, it holds that $\BM{t}-\BM{s} \sim \normaldist{0}{\left( t-s \right) I_{\noisedim} }$.
Furthermore, $\BMcoeff{}\in \R^{\statedim \times \noisedim}$ reflects the effect of $\BM{t}$ on the state evolution. 

We aim to design computationally tractable and provably efficient reinforcement learning policies for the system in~\eqref{dynamics}. The transition matrix $\Amat{\star}$, the input matrix $\Bmat{\star}$, and the noise-coefficient matrix $\BMcoeff{}$,  \emph{all are unknown}. The goal is to minimize the expected average cost 
\begin{equation*}
\avecost{\policy}{}=\limsup\limits_{T \to \infty} \frac{1}{T} ~ \E{\itointeg{0}{T}{\instantcost{\policy}{t} }{t}},
\end{equation*}
where $\instantcost{\policy}{t}= \state{t}^{\top} \Qmat \state{t}+ \action{t}^{\top} \Rmat \action{t}$
is the cost of policy $\policy$ at $t$, its value being determined by the positive definite matrices $\Qmat,\Rmat$ of proper dimensions, as explained below.  

The policy $\policy$ is non-anticipative closed-loop: At every time $t$, $\policy$ determines $\action{t}$ according to the information available at the time. More precisely, $\policy$ maps the state observations (i.e., $\state{s} $ for $s \in \left[0,t\right]$) and the previously taken actions (i.e., $ \action{s}$ for $s$ in the semi-open interval $\left[0,t\right[$) to the current control action $\action{t}$. This mapping can be stochastic or deterministic. Importantly, $\policy$ faces the fundamental exploration-exploitation dilemma for minimizing the expected average cost, because the dynamics matrices $\truth$ are unknown and need to be learned based on the state and action observations. The details of this dilemma is discussed in Section~\ref{AlgoSection}. We assume that $\Qmat, \Rmat$ are known to the policy, the rationale being that the decision-makers are aware of the objective they aim to achieve, while their uncertainty about the environment impedes them from deciding optimally. 

The benchmark for assessing reinforcement learning policies is the optimal policy $\optimalpolicy$ that designs $\action{t}$ \emph{having access} to $\truth$. To define $\optimalpolicy$, let $\MatOpAve{\truth}{\cdot} : \R^{\statedim \times \statedim} \to \R^{\statedim \times \statedim}$ be
\begin{equation*}
\MatOpAve{\truth}{M} = \Amat{\star}^{\top} M + M \Amat{\star} - M \Bmat{\star} \Rmat^{-1} \Bmat{\star}^{\top} M + \Qmat.
\end{equation*}
This function is vital for finding $\optimalpolicy$. To see the intuition, first note that an action $\action{t}$ directly influences the current cost $\instantcost{\policy}{t}$, and indirectly affects the future cost values according to~\eqref{dynamics}. So, to capture future consequences of decisions, $\MatOpAve{\truth}{\cdot}$ is employed~\citep{yong1999stochastic}. To proceed toward identifying $\optimalpolicy$, let the positive semidefinite matrix $M=\RiccSol{\truth}$ solve $\MatOpAve{\truth}{M}=0$. To investigate existence and uniqueness of $\RiccSol{\truth}$, we need the followings.
\begin{deffn} [Notations $\mosteig{\cdot}, \Learnerror{}{\cdot}$] \label{NotationDef}
	{Let $\mosteig{D}$ be the largest real-part of the eigenvalues of an arbitrary square matrix $D$:
		$\mosteig{D} = \max \left\{ \Re \left( \lambda \right): \det \left( D - \lambda I \right)=0 \right\}$.
		Further, for arbitrary matrices $\estA{} \in \R^{\statedim \times \statedim}$, $\estB{} \in \R^{\statedim \times \controldim}$, define $\Learnerror{}{\estpara{}} = \Mnorm{\estA{}-\Amat{\star}}{2} + \Mnorm{\estB{}-\Bmat{\star}}{2}$.
		So, $\Learnerror{}{\estpara{}}$ measures the deviation of $\estpara{}$ from the true dynamics matrices $\truth$.}
\end{deffn}
Note that unlike $\eigmax{\cdot}$ that considers only \emph{magnitudes} of the eigenvalues, $\mosteig{\cdot}$ reflects the signs of the eigenvalues as well, and so can be either positive, zero, or negative. However, they are related according to $\eigmax{e^{D}}=e^{\mosteig{D}}$.

We assume that the true dynamics matrices $\truth$ are stabilizable, in the following sense:
\begin{assump} [Stabilizability] \label{StabAssump} 
	There exists some $\Gainmat{} \in \R^{\controldim \times \statedim}$, such that $\mosteig{\Amat{\star}+\Bmat{\star}\Gainmat{}} <0$.
\end{assump}
Assumption~\ref{StabAssump} expresses that by applying $\action{t}=\Gainmat{}\state{t}$, the system can operate without any explosion. To see that, solve the differential equation~\eqref{dynamics} under the feedback policy $\action{t}=\Gainmat{}\state{t}$ to obtain
\begin{equation} \label{stateevol}
\state{t}=e^{\left(\Amat{\star}+\Bmat{\star}\Gainmat{}\right)t} \state{0}+ \itointeg{0}{t}{e^{\left(\Amat{\star}+\Bmat{\star}\Gainmat{}\right)(t-s)} \BMcoeff{}}{\BM{s}} .
\end{equation}
So, because of $\mosteig{\Amat{\star}+\Bmat{\star}\Gainmat{}} <0$, $\state{t}$ does not grow unbounded with $t$. Importantly, existence of a stabilizing matrix $\Gainmat{}$ is \emph{necessary} for the problem to be well-defined. Otherwise, state explosion renders the average cost infinite  for \emph{all} decision-making policies~\citep{chen1995linear,yong1999stochastic}. Now, recall that $\MatOpAve{\truth}{\RiccSol{\truth}}=0$, let $\Optgain{\truth}= - \Rmat^{-1} \Bmat{\star}^{\top}\RiccSol{\truth}$, and define the linear feedback policy 
\begin{equation} \label{OptimalPolicy} 
\optimalpolicy : ~~~~~~~~~~~~~~~ \action{t}=\Optgain{\Amat{\star},\Bmat{\star}} \state{t}, ~~~~~~\forall t \geq 0.
\end{equation}
We show that Assumption~\ref{StabAssump} suffices for unique existence of $\RiccSol{\truth}$ and for optimality of $\optimalpolicy$. 
\begin{thrm}[Optimal policy] \label{OptimalityProof}
	Under Assumption~\ref{StabAssump}, the matrix $\RiccSol{\truth}$ uniquely exists, and the policy $\optimalpolicy$ in~\eqref{OptimalPolicy} gives 
	$$\avecost{\optimalpolicy}{}= \inf\limits_{\policy} \avecost{\policy}{} = \tr{  \RiccSol{\truth} \BMcoeff{} \BMcoeff{}^{\top} }, ~~~~~~~~~~~~\text{and}~~~~~~~~~~ \mosteig{\Amat{\star}+\Bmat{\star}\Optgain{\truth}}<0.$$
\end{thrm}
To compute $\RiccSol{\truth}$, it suffices to solve the differential equation $\dot{M}_t=\MatOpAve{\truth}{M_t}$ starting from a positive semidefinite $M_0$, or equivalently calculate the integral $M_t=M_0 + \itointeg{0}{t}{\MatOpAve{\truth}{M_s}}{s}$.
In the proof of Theorem~\ref{OptimalityProof}, we show that $\lim\limits_{t\to \infty}M_t = \RiccSol{\truth}$.

Next, we focus on learning $\truth$ and the additional cost compared to the cost of $\optimalpolicy$ that we are charged for, because of uncertainties about $\truth$. 
To that end, we formulate sub-optimalities in the performance of decision-making policies and the penalty due to lack of knowledge about the optimal actions $\action{t}=\Optgain{\truth}\state{t}$. For a general policy $\policy$, the \emph{regret} of $\policy$ at time $T$ is denoted by $\regret{T}{\policy}$ and is defined as the cumulative increase in cost by time $T$. That is, the difference between the instantaneous costs of $\policy$ and $\optimalpolicy$ in~\eqref{OptimalPolicy} is integrated over the time interval $\left[0 , T\right]$: 
\begin{equation} \label{RegretDefEq}
\regret{T}{\policy} = \itointeg{0}{T}{ \left[ \instantcost{\policy}{t} - \instantcost{\optimalpolicy}{t} \right] }{t}.
\end{equation}
Note that under $\policy$, the state trajectory is generated by \eqref{dynamics} for $\action{t}=\policy \left( \left\{ \action{s} \right\}_{0 \leq s < t} , \left\{\state{s} \right\}_{0 \leq s \leq t} \right)$, at all times $t \geq 0$. So, $\instantcost{\policy}{t} - \instantcost{\optimalpolicy}{t}$ includes the differences between the control actions as well as the state trajectories of $\policy,\optimalpolicy$. Clearly, the random state evolution in \eqref{dynamics} renders $\regret{T}{\policy}$ random. So, regret analyses for reinforcement learning policies include worst-case analyses that establish upper-bounds for $\regret{T}{\policy}$, as well as average-case analyses that provide bounds for $\E{\regret{T}{\policy}}$. Further, for unknown $\truth$, we hope that the increasing observations of state and action over time will be effectively leveraged so that eventually, the policy makes near-optimal decisions. So, as $t$ grows, we desire $\instantcost{\policy}{t}-\instantcost{\optimalpolicy}{t}$ to shrink and so $\regret{T}{\policy}$ to scale sub-linearly with $T$. In the sequel, we study $\regret{T}{\policy}$, $\E{\regret{T}{\policy}}$ and their dependence on $T$ and the problem parameters.

Another quantity of interest is the \emph{accuracy} of estimating the unknown dynamics. So, letting $\estA{t},\estB{t}$ be estimates of $\truth$ based on the state-action observations by time $t$; i.e., $\left\{\state{s}, \action{s}\right\}_{0 \leq s \leq t}$, we study the decay rate of the estimation error $\Learnerror{}{\estpara{t}}$, as defined in Definition~\ref{NotationDef}. Note that similar to regret, $\Learnerror{}{\estpara{t}}$ is stochastic.

\section{Stability Analysis for Perturbed Dynamics Matrices} \label{StabilitySection}
In this section, we study the effects of uncertainties about the dynamical model on system stability. We specify the minimal information one needs to possess in order to ensure stabilization, and show that a coarse-grained approximation of the truth is sufficient for this purpose. Results of this section will be used later in the design of randomized-estimates policy in Section~\ref{AlgoSection}. Importantly, the following stability analysis is general, captures effects of all involved quantities, and provides tight results in the sense that the conditions of Theorem~\ref{StabThm} are required for guaranteeing stabilization. In addition, the results presented here are of independent interests, because stability is required for letting the system operate for a reasonable time period, regardless of optimality of the control actions. 

To proceed, note that if hypothetically the optimal linear feedback in~\eqref{OptimalPolicy} is applied to the system in~\eqref{dynamics}, then stability is guaranteed. More precisely, applying $\Optgain{\truth}$, the resulting closed-loop transition matrix $\CLmat{\star}=\Amat{\star}+\Bmat{\star}\Optgain{\truth}$ has all its eigenvalues on the open left half-plane of the complex plane, as stated in Theorem~\ref{OptimalityProof}. The issue is that the true dynamics matrices $\truth$ are \emph{unknown} and need to be learned. However, if some matrices $\estpara{}$ meet the conditions we shortly discuss, one can stabilize the system by applying the linear feedback $\action{t}=\Optgain{\estpara{}} \state{t}$. 

To proceed, let $\estD{} = \estA{}+\estB{}\Optgain{\estpara{}}$ be the closed-loop transition matrix of a system with dynamics matrices $\estpara{}$, under the feedback $\action{t}=\Optgain{\estpara{}}\state{t}$. Then, let $\stabradii > 0$ and $\RiccUpperBound<\infty$ satisfy
\begin{eqnarray} \label{StabCond1}
\mosteig{\estD{}} \leq - \stabradii, ~~~~~~~~~~~~~~~~~ \Mnorm{\RiccSol{\estpara{}}}{2} \leq \RiccUpperBound.
\end{eqnarray}
The quantities in~\eqref{StabCond1} are required for studying stability of the matrix $\Amat{\star}+\Bmat{\star}\Optgain{\estpara{}}$, as follows. Remember that we have a closed-loop stability result in Theorem~\ref{OptimalityProof}: $\mosteig{\estD{}}<0$. So, the first inequality in \eqref{StabCond1} quantifies the extent to which $\Optgain{\estpara{}}$ is able to stabilize the system of parameters $\estpara{}$. Intuitively, $\mosteig{\estD{}}$ is the best (i.e., most negative) upper-bound one can hope for the eigenvalues of $\Amat{\star}+\Bmat{\star}\Optgain{\estpara{}}$, since the optimal feedback $\Optgain{\estpara{}}$ is \emph{purposefully} designed for the certainly known matrices $\estpara{}$. Later on, we show that $\stabradii$ enjoys a positive uniform lower-bound as long as $\estpara{}$ live in some neighborhoods of $\truth$. The second inequality in \eqref{StabCond1} is somewhat guaranteed by the first one, as we will show in the proof of Theorem~\ref{OptimalityProof} (in \eqref{LyapInteg2}) that 
\begin{equation} \label{LyapInteg}
\RiccSol{\estpara{}} = \itointeg{0}{\infty}{ e^{\estD{}^{\top}t} \left( \Qmat+ \Optgain{\estpara{}}^{\top} \Rmat \Optgain{\estpara{}} \right) e^{\estD{}t} }{t}.
\end{equation}
Therefore, $\mosteig{\estD{}} \leq -\stabradii<0$ implies that for some  $\RiccUpperBound<\infty$, we have $\Mnorm{\RiccSol{\estpara{}}}{} \leq \RiccUpperBound$. So, $\RiccUpperBound$ is merely used for simplifying the expressions.

Towards stability analysis, we need further information of $\estD{} = \estA{}+\estB{}\Optgain{\estpara{}}$ that the Jordan form of this matrix provides. Suppose that eigenvalues of $\estD{}$ are $\lambda_1, \cdots, \lambda_k$, and let the Jordan decomposition be $\estD{}=P^{-1}\Lambda P$; i.e., $\Lambda = \diag{\Lambda_1,\cdots, \Lambda_k}$ is a block-diagonal matrix and all diagonal entries of $\Lambda_i$ are $\lambda_i$, the immediate off-diagonal entries above the diagonal of $\Lambda_i$ are $1$, and all other entries of $\Lambda_i$ are $0$ (as shown in~\eqref{JordanBlocks}). Now, we introduce a very important quantity for determining the stability margin. For the above-mentioned blocks $\Lambda_i$, $i=1, \cdots, k$, let $\mult{i}$ denote the dimension of the square matrix $\Lambda_i$, and refer to the largest value among $\mult{1}, \cdots, \mult{k}$ by $\mult{\estD{}}$.
\begin{deffn} [Largest block-size $\mult{\estD{}}$] \label{MultDef}
	Letting $P$ and $\Lambda_i \in \C^{\mult{i} \times \mult{i}}$ be as in the Jordan decomposition $\estD{}=P^{-1} \Lambda P$ explained above, define $\mult{\estD{}}=\max\limits_{1 \leq i \leq k} \mult{i}$.
\end{deffn}
The quantity $\mult{\estD{}}$ is the largest size of the blocks $\Lambda_1, \cdots, \Lambda_k$ in the Jordan form and crucially determines the \emph{order} of  stability margin, as established in the following theorem.

\begin{thrm}[Stability margin] \label{StabThm}
	Let $P,\mult{\estD{}}$ and $\stabradii,\RiccUpperBound$ be as in Definition~\ref{MultDef} and \eqref{StabCond1}, respectively. Then, following Definition~\ref{NotationDef}, for $\delta>0$, we have $\mosteig{\Amat{\star}+ \Bmat{\star} \Optgain{\estpara{}}} < -\delta$, as long as
	\begin{equation} \label{StabNeighborhood}
	\Learnerror{}{\estpara{}} < \left( 1 \wedge \frac{1}{\Mnorm{\Optgain{\estpara{}}}{2}} \right) \frac{ \left(\stabradii-\delta\right) \wedge \left(\stabradii-\delta\right)^{\mult{\estD{}}}}{\mult{\estD{}}^{1/2} \Mnorm{P}{2} \Mnorm{P^{-1}}{2}}.
	\end{equation}
\end{thrm}
Note that the definition of $\Optgain{\estpara{}}$ before \eqref{OptimalPolicy} shows that $\Mnorm{\Optgain{\estpara{}}}{2}^{-1}$ in \eqref{StabNeighborhood} can be replaced with $\eigmin{\Rmat}\RiccUpperBound^{-1} \Mnorm{\estB{}}{2}^{-1}$. Theorem~\ref{StabThm} states that if $\Learnerror{}{\estpara{}}$ is sufficiently small to satisfy~\eqref{StabNeighborhood}, then ${\Amat{\star}+ \Bmat{\star} \Optgain{\estpara{}}}$ is stable and all of its eigenvalues in the complex plane lie on the left-hand-side of the vertical line $\Re =-\delta$. In addition, \eqref{StabNeighborhood} reflects effects of different factors, as follows. First, the stability margin on the right-hand-side of \eqref{StabNeighborhood} decreases as $\Mnorm{\Optgain{\estpara{}}}{2}$ increases. To see the intuition, note that the difference between the closed-loop matrices is $\Amat{\star}-\estA{}+ \left(\Bmat{\star}-\estB{}\right)\Optgain{\estpara{}}$, which shows the multiplicative effect of $\Optgain{\estpara{}}$. Further, Definition~\ref{MultDef} indicates that $\mult{\estD{}}^{1/2} \Mnorm{P^{}}{2} \Mnorm{P^{-1}}{2}$ quantifies non-diagonality of $\estD{}$, is at least $1$, and becomes $1$ for diagonal $\estD{}$ (where $P$ is the identity matrix and $\mult{\estD{}}=1$). Therefore, more non-diagonal closed-loop matrices $\estD{}$ lead to smaller stability margins and make stabilization harder to be learned.

By \eqref{StabCond1}, the dependence on $\stabradii$ corroborates the intuition that systems whose optimal closed-loop matrices has eigenvalues of larger real-parts (i.e., smaller $\stabradii$), are harder to stabilize. Moreover, the expression $\left(\stabradii-\delta\right) \wedge \left(\stabradii-\delta\right)^{\mult{\estD{}}}$ indicates that $\mult{\estD{}}$ is very important and determines the \emph{rates} of bounding the eigenvalues of ${\Amat{\star}+ \Bmat{\star} \Optgain{\estpara{}}}$. The rates for $\stabradii-\delta<1$ and $\stabradii-\delta>1$ are different, because of a similar phenomena in the sensitivity of eigenvalues of matrices to perturbations in their entries. This result is of independent interest as it is a generalization of Bauer-Fike Theorem~\citep{bauer1960norms} to \emph{asymmetric} matrices. Indeed, we establish in the proof of Theorem~\ref{StabThm} that \emph{larger blocks in Jordan forms can lead to drastically higher eigenvalue-sensitivities against entries}. 

To close this section, we provide uniform lower and upper bounds for $\stabradii>0$ and $\RiccUpperBound < \infty$, respectively. For that purpose, similar to Definition~\ref{MultDef}, define the largest block size $\mult{\star}=\mult{\CLmat{\star}}$ based on the Jordan decomposition $\CLmat{\star}=\Amat{\star}+\Bmat{\star}\Optgain{\truth}=P_\star^{-1} \Lambda_\star P_\star$. Then, we show in the proof of Theorem~\ref{StabThm} that $\Learnerror{}{\estpara{}} \leq \epsilon_0$ is sufficient for stabilization, and it holds that $\stabradii \geq \eigmin{\Qmat} 4^{-1} \Mnorm{\RiccSol{\truth}}{2}^{-1}$, and $\RiccUpperBound \leq 2 \Mnorm{\RiccSol{\truth}}{2}$, as long as
\begin{equation} \label{StabRadiEq0}
\left( 1 \vee \Mnorm{\Optgain{\truth}}{2} \right) \epsilon_0 = \frac{ \left( -\mosteig{\CLmat{\star}} \right) \wedge \left( -\mosteig{\CLmat{\star}}\right)^{\mult{\star}} }{\mult{\star}^{1/2} \Mnorm{P_\star^{-1}}{2} \Mnorm{P_\star}{2} }
\wedge \left[4 \itointeg{0}{\infty}{ \Mnorm{e^{\CLmat{\star}t}}{2}^2 }{t} \right]^{-1} . 
\end{equation}

\section{Tight Regret Expressions and Policy Differentiation} \label{SubOptSection}
In this section, we investigate sub-optimalities and provide a sharp expression for the regret that non-optimal control actions cause. Such an investigation is vital since reinforcement learning policies need to learn the unknown dynamics $\truth$ and so they require to take non-optimal actions. 

To proceed, let $\action{t}$ be the control action of the policy $\policy$ at time $t$. In the following theorem, we quantify $\regret{T}{\policy}$ in terms of deviations $\action{t} - \Optgain{\truth}\state{t}$, and introduce $\regterm{T}$ that fully assesses the regret of $\policy$. In fact, $\regterm{T}$ unifies average-case and worst-case analyses by capturing both $\E{\regret{T}{\policy}}$ and $\regret{T}{\policy}$. Further, Theorem~\ref{GeneralRegretThm} provides scalings with different problem parameters and shows that $\regret{T}{\policy}-\E{\regret{T}{\policy}}$ scales linearly with the dimension of the Brownian motion. 

\begin{thrm} [Regret analysis] \label{GeneralRegretThm}
	Suppose that $\Gainmat{t} $ is a bounded piecewise-continuous function of $t$, and $\policy$ is the policy $\action{t}=\Gainmat{t}\state{t}$. Then, we have $\E{\regret{T}{\policy}}=\E{\regterm{T}}$, and
	\begin{eqnarray*}
		\regret{T}{\policy} &=& \regterm{T} + \order{ \regconstant \regterm{T}^{1/2} \log \regterm{T}},
	\end{eqnarray*}
	where $\regconstant= \frac{ \Mnorm{\BMcoeff{}}{2} \Mnorm{\RiccSol{\truth}}{2}^{3/2} \noisedim }{ \eigmin{\Qmat}^{1/2} \eigmin{\Rmat}^{1/2}} $, $\CLmat{\star}=\Amat{\star}+\Bmat{\star}\Optgain{\truth}$, $E_t=e^{\CLmat{\star}^{\top}t} \RiccSol{\truth} e^{\CLmat{\star}t}$, and
	\begin{eqnarray*}
		\regterm{T} = \itointeg{0}{T}{ \norm{ \Rmat^{1/2} \left( \Gainmat{t} - \Optgain{\truth} \right) \state{t} }{}^2 }{t} 
		- 2 {\itointeg{0}{T}{ \left( \state{t}^{\top} E_{T-t} \Bmat{\star} \left( \Gainmat{t}-\Optgain{\truth} \right) \state{t} \right) }{t}}.
	\end{eqnarray*}
\end{thrm}
To establish Theorem~\ref{GeneralRegretThm}, we utilize the theory of continuous-time martingales and (in Lemma~\ref{SelfNormalizedLem}) develop novel results on comparative ratios of stochastic integrals. More importantly, we construct the new framework of \emph{policy differentiation} for finding sharp regret bounds. Broadly speaking, policy differentiation precisely evaluates the regret in terms of infinitesimal sub-optimalities and integrates these infinitesimal deviations to obtain $\regterm{T}$, in which the integrand $\norm{ \Rmat^{1/2} \left( \Gainmat{t} - \Optgain{\truth} \right) \state{t} }{}^2$ plays a role similar to the \emph{derivative} of the regret. This framework can be used for analogous regret analyses in other continuous-time reinforcement learning problems.

The boundedness and piecewise continuity conditions in Theorem~\ref{GeneralRegretThm} are somewhat natural because the optimal policy in \eqref{OptimalPolicy} is a time-invariant feedback, and so one gains nothing by violating these conditions. Furthermore, since by Theorem~\ref{OptimalityProof} we have $\mosteig{\CLmat{\star}}<0$, the matrix $E_t$ exponentially decays as $t$ grows. Thus, the second integral in the definition of $\regterm{T}$ is dominated by the first one. Moreover, note that the above-mentioned matrix appears in $\regterm{T}$ in the form of $E_{T-t}$, i.e., with a time inversion. This reflects the fact that sub-optimal control feedbacks $\action{t}=\Gainmat{t}\state{t}$ have descending effects on the regret $\regret{T}{\policy}$ as we move from $T$ backward in time that $t$ descends. 

Putting the discussions in the above two paragraphs all together, we conclude as follows. Theorem~\ref{GeneralRegretThm} shows that the sub-optimality $\policy$ incurs at time $t$, scales as the \textbf{square} of the deviation $\Gainmat{t} - \Optgain{\truth}$. On top of that, the constant $\regconstant$ reflects effects of different parameters and indicates, for example, that  $\regret{T}{\policy}-\regterm{T}$ scales linearly with $\noisedim$.

The results of Theorem~\ref{GeneralRegretThm} are insightful along different directions. First, the exact equality $\E{\regret{T}{\policy}} =  \E{\regterm{T}}$ can be used for establishing minimax \emph{lower-bounds} for regret by finding the fastest rates 
$\Gainmat{t}-\Optgain{\truth}$ can shrink. Further, since $\regterm{T}^{1/2} \log \regterm{T}=\order{\regterm{T}}$, both the average-case criteria $\E{\regret{T}{\policy}}$ as well as the worst-case regret $\regret{T}{\policy}$ are captured by~$\regterm{T}$. In other words, Theorem~\ref{GeneralRegretThm} indicates that the fluctuations of $\regret{T}{\policy}$ around its expectation $\E{\regret{T}{\policy}}$ are in magnitude smaller than the expected value. Thus, studying $\regterm{T}$ is \emph{sufficient and necessary} for regret analysis and there is a tight and reciprocal relationship between $\regret{T}{\policy},\regterm{T}$, for \textbf{all} policies.  

Moreover, as time goes by, a reinforcement learning policy becomes \emph{progressively} more capable of narrowing down the sub-optimality gap by estimating the unknown dynamics $\truth$. Indeed, as soon as having sufficiently long trajectories to learn $\truth$ accurate enough to satisfy $\norm{ \Rmat^{1/2}\left( \Gainmat{t} - \Optgain{\truth} \right) \state{t}}{2}<1$, the regret grows much slower since $\regterm{T}$ integrates the \emph{squares} of these deviations. For example, if the estimation accuracy satisfies the ideal square-root rate $\Learnerror{}{\estpara{t}}=\order{t^{-1/2}}$, then the regret is a logarithmic function of time; $\regterm{T}=\order{\log T}$. However, due to the trade-off between the exploration and exploitation that will be elaborated shortly, this is not the case and to obtain the above error rate, $\action{t}$ needs to persistently deviate from $\Optgain{\truth}\state{t}$, which causes a \emph{linearly} growing regret (see Proposition~\ref{FullEstimationProp} in the appendices). 

Theorem~\ref{GeneralRegretThm} provides both a general result for analyzing reinforcement learning policies, as well as a useful insight on how to design them to minimize the regret. We utilize this insight to design Algorithm~\ref{algo1} and to establish Theorem~\ref{BoundsThm} in the next section. Indeed, we randomize the parameter estimates so that $\action{t}$ appropriately deviates from $\optimalpolicy$, leading to $\Learnerror{}{\estpara{t}}=\orderlog{t^{-1/4}}$. So, we obtain the efficient regret bound $\regret{T}{\policy}=\order{\regterm{T}}=\orderlog{T^{1/2}}$.
\section{Randomized-Estimates Policy}  \label{AlgoSection}
In this section, we discuss a fast and tractable algorithm with an efficient performance for cost minimization subject to uncertainties about the dynamics matrices $\truth$. First, we explain the fundamental exploration-exploitation dilemma. Then, we investigate a procedure for estimating the unknown dynamics using the data of state-action trajectory. Based on that, an online reinforcement learning policy that employs randomizations of the parameter estimates for balancing exploration versus exploitation is presented in Algorithm~\ref{algo1}. Next, a regret bound is established in Theorem~\ref{BoundsThm} indicating that Algorithm~\ref{algo1} efficiently minimizes the cost function so that the regret scales as the square-root of the time. We also specify the rates of identifying the dynamics matrices.

According to Theorem~\ref{GeneralRegretThm}, the policy needs to ensure that $\action{t} \approx \Optgain{\truth} \state{t}$ in order to incur a small regret. Furthermore, since $\truth$ are unknown, the policy needs to estimate them based on the data $\left\{ \state{s},\action{s} \right\}_{0 \leq s \leq  t}$. However, if $\action{s} \approx \Optgain{\truth} \state{s}$, then the $\action{s}$ coordinates of the data point $\state{s},\action{s}$ become (almost) uninformative as they are (approximately) linear transformations of the $\state{s}$ coordinates. This defeats the purpose and renders accurate estimation of $\truth$ infeasible. Note that accurate approximations of $\truth$ are needed for taking near-optimal control actions. This, known as the exploration-exploitation dilemma, is the main obstacle in online reinforcement learning and shows that a low-regret policy needs to carefully \emph{diversify} the actions $\left\{ \action{s} \right\}_{0 \leq s \leq  t}$ by \emph{deviating} from $\left\{\Optgain{\truth} \state{s}\right\}_{0 \leq s \leq  t}$.

\subsection{Design of the algorithm and intuitions}
Now, we discuss the learning procedure in Algorithm~\ref{algo1}, based on extensions of the least-squares estimates. Suppose that instead of the full data $\state{s},\action{s}$ for real values of ${s \geq 0}$, one has access to samples at a discrete set of time points; $\state{k\epsilon},\action{k\epsilon}$ for $k=0,1, \cdots$. Then, \eqref{dynamics} for a small $\epsilon$ gives the approximate data generation mechanism
$\state{(k+1)\epsilon}-\state{k\epsilon} = \left( \Amat{\star} \state{k\epsilon} + \Bmat{\star} \action{k\epsilon} \right) \epsilon + \BMcoeff{}\left(\BM{(k+1)\epsilon} - \BM{k\epsilon}\right)$.
So, an approach is to 
estimate $\truth$ by minimizing 
$\sum_{k} \norm{\state{(k+1)\epsilon}-\state{k\epsilon} - \left( \estA{} \state{k\epsilon} + \estB{} \action{k\epsilon} \right) \epsilon}{2}^2$ over $\estpara{}$.
Letting $\statetwo{s}=\left[\state{s}^{\top},\action{s}^{\top}\right]^{\top}$
 and $\epsilon \to 0$, we get the continuous-time estimate based on the full trajectory $\left\{ \state{s},\action{s} \right\}_{0 \leq s \leq  t}$. The result is shown in \eqref{RandomLSE1} below and will be used by Algorithm~\ref{algo1}. 

To ensure that the system evolves stably, Algorithm~\ref{algo1} projects the estimates on the following stabilization oracle $\paraspace{0}$ in lights of Theorem~\ref{StabThm}. In the sequel, we explain how one can learn $\paraspace{0}$ fast. 
\begin{deffn} \label{OracleDef}
	For a fixed $\delta_0>0$, let $\paraspace{0}$ be a set containing matrices $\estpara{}$ that satisfy~\eqref{StabNeighborhood} for $\delta=\delta_0$. 
\end{deffn}
Intuitively, the system is stabilized by having access to $\paraspace{0}$, despite uncertainties about the true dynamics matrices $\truth$. Note that the condition in~\eqref{StabNeighborhood} is verifiable since $\stabradii,\RiccUpperBound$ depend on the \emph{known} parameter estimates $\estpara{}$. Availability of a stabilization set is a common assumption in the literature of online reinforcement learning policies for linear systems \citep{mandl1988consistency,mandl1989consistency,caines1992continuous,abeille2018improved,ouyang2019posterior,faradonbeh2020adaptive,ziemann2020phase,ziemann2020regret}. For example, if an initial stabilizing feedback $\Gainmat{0}$ is available, we can devote a period to only explore by applying sub-optimal control actions, and use the resulting observations to learn $\paraspace{0}$. Such procedures, that ignore the main objective of regret minimization for a short time period, are shown to be fast and effective in the sense that the probability of failing to learn to stabilize, decays exponentially with time~\citep{shirani2022thompson}. In systems that are in operation prior to running Algorithm~\ref{algo1} or in open-loop-stable systems, this condition automatically holds (in the latter case, $\Gainmat{0}=0$ is an initial stabilizer). Similarly, in systems with a \emph{reset} option that can immediately steer the system-state to small values, $\paraspace{0}$ can be learned fast~\citep{duncan1999adaptive,caines2019stochastic,basei2021logarithmic}. 

In absence of initial stabilizer and state-reset options, learning $\paraspace{0}$ will be more challenging. In Section \ref{StabilitySection} we saw that an $\epsilon_0$ neighborhood of $\truth$ is sufficient for bounding $\stabradii,\RiccUpperBound$, for $\epsilon_0$ in \eqref{StabRadiEq0}. That is, coarse-grained approximations of $\truth$ suffice for stabilization. So, $\paraspace{0}$ can be learned from short state trajectories derived by applying randomized control actions~\citep{caines2019stochastic,faradonbeh2018bfinite,faradonbeh2019randomized,lale2020explore,chen2021black,gramlich2022fast}. Importantly, fast and reliable learning-based stabilization can be ensured via Bayesian methods~\citep{faradonbeh2021bayesian}. These methods retain a Gaussian posterior about the unknown true dynamics $\truth$, and design stabilizing feedbacks as if samples from the posterior coincide with the truth. Importantly, even if failed at their first attempts, these Bayesian methods can be utilized again with no need to repeat the state observation and learning procedure. Technically, if a failure is detected (e.g., if the state magnitude $\norm{\state{t}}{2}$ keeps growing), one can can resample from the posterior until successful stabilization~\citep{faradonbeh2021bayesian}. Note that effectiveness of the above methods does not depend on availability of any prior distribution, although an informative prior can help to improve learning from shorter state trajectories.

Finally, we will show (in Theorem~\ref{BoundsThm}) that the algorithm learns $\truth$ with the rate $t^{-1/4}$. So, the projection on $\paraspace{0}$ will be automatically performed  after the time $t_0 = \orderlog{\epsilon_0^{-4}}$. 

The algorithm proceeds as follows. For some fixed $\episodetime{}>1$, the exponentially growing sequence $\left\{ \episodetime{n} \right\}_{n=0}^\infty$ contains the time instants at which the algorithm updates the parameter estimates. In fact, Algorithm~\ref{algo1} applies control actions $\action{t}=\Optgain{\estpara{n}}\state{t}$ during the time period $\episodetime{n} \leq t < \episodetime{n+1}$, where $\estpara{n}$ are estimates of $\truth$ , based on the trajectory up to time $\episodetime{n}$. 

Further, to ensure that the policy commits sufficiently to explore the environment, a random matrix $\randommatrix{n}$ is added to the parameter estimates at time $t=\episodetime{n}$. Then, Algorithm~\ref{algo1} projects the resulting $\statedim \times \left( \statedim + \controldim \right)$ matrix onto $\paraspace{0}$. Formally, let $\project{\paraspace{0}}{\cdot}$ denote projection on $\paraspace{0}$; i.e., it gives the closest matrix in $\paraspace{0}$ according to the distance induced by the Frobenius norm. Then, define
\begin{equation} \label{RandomLSE1}
\left[ \estpara{n} \right]= \project{\paraspace{0}}{\left[ \itointeg{0}{\episodetime{n}}{ \statetwo{s} }{\state{s}^{\top}} \right]^{\top} \left[ \itointeg{0}{\episodetime{n}}{\statetwo{s} \statetwo{s}^{\top}}{s} \right]^\dagger + \randommatrix{n}},
\end{equation}
where $\statetwo{s}=\left[\state{s}^{\top},\action{s}^{\top}\right]^{\top}$
and the $\statedim \times \left(\statedim+\controldim\right)$ matrices $\left\{\randommatrix{n}\right\}_{n=0}^{\infty}$ are independent of everything else and of each others. Further, the random matrix $\randommatrix{n}$ that is used at time $t=\episodetime{n}$ has independent Gaussian entries of mean zero and standard deviation $\sigma_n = \sigma_0 \left(\episodetime{-n}n\right)^{1/4} = \sigma_0 \left( t^{-1} \log_{\episodetime{}} t \right)^{1/4}$, for some fixed $\sigma_0$. This decay rate of $\sigma_n$ is delicately adjusted for two purposes. On one hand, $\randommatrix{n}$ is sufficiently large for randomizing the parameter estimates to ensure that effective exploration occurs and the current data is diverse enough so that we obtain accurate estimates in the future. On the other hand, $\randommatrix{n}$ is sufficiently small to let the current estimates remain accurate and prevent significant deviations. Otherwise, large values of $\randommatrix{n}$ deteriorate the current efficient exploitation. More precisely, this value of $\sigma_n$ is selected by minimizing its role in the regret that consists of summations of some terms of the form $\episodetime{n} \left( \sigma_n^2 + \sigma_n^{-2} n \right)$.

\begin{algorithm}
	\caption{\bf: Randomized-Estimates Policy} \label{algo1}
%
%
	\begin{algorithmic}
	\State Select $\episodetime{}>1$ and $\estpara{0} \in \paraspace{0}$ arbitrarily
	\State {For $0 \leq t < 1=\episodetime{0}$, apply $\action{t}=\Optgain{\estpara{0}} \state{t}$} 
	\For{$n=0,1,2, \cdots$}
	\State Obtain parameter estimates $\estpara{n}$ by~\eqref{RandomLSE1}
		\While {$\episodetime{n} \leq t < \episodetime{n+1}$}
		\State {Take control action $\action{t}=\Optgain{\estpara{n}} \state{t}$} 
			\EndWhile
		\EndFor
	\end{algorithmic}
\end{algorithm}

\subsection{Analysis of the algorithm and  performance guarantees}
The memory that Algorithm~\ref{algo1} occupies is remarkably small since  it can update the parameter estimates 
by only storing the values of the two integrals in~\eqref{RandomLSE1} in an online fashion. Furthermore, calculations can be done quite fast, making update of the parameter estimates at time $\episodetime{n}$ immediately effective. Note that the computational complexity of numerically obtaining a sufficiently accurate $\Optgain{\estpara{n}}$ is at worst $\order{\left(\statedim+\controldim\right)^{3} n}$, for the matrix operations in the definition of $\optimalpolicy$ in \eqref{OptimalPolicy}, and for the fact that the (integration or differentiation) procedures described after Theorem~\ref{OptimalityProof} converge exponentially fast, while the next theorem indicates that an accuracy of $\order{\episodetime{-n/4}}$ is sufficient (and necessary).

The rationale for freezing the parameter estimates for exponentially growing time intervals $\episodetime{n} \leq t < \episodetime{n+1}$ is that Algorithm~\ref{algo1} can defer the learning step until collecting enough observations $\statetwo{s}$ so that a new update of parameter estimates is effectively more accurate than the previous one. 

The following result provides performance guarantees for the randomized-estimates policy. 
\begin{thrm} [Analysis of Algorithm~\ref{algo1}] \label{BoundsThm}
	Let the policy $\policy$ and the estimates $\estpara{n}$ be those in Algorithm~\ref{algo1}. Assume that $n$ satisfies $\episodetime{n} \leq T < \episodetime{n+1}$. Then, using Definition \ref{NotationDef} and \eqref{RegretDefEq}, we have 
	\begin{eqnarray*}
		\Learnerror{}{\estpara{n}}^2 = \order{ \learnconstant T^{-1/2} \log T }, ~~~~~~~~~~~~~
		\regret{T}{\policy} = \order{ \algoconstant T^{1/2} \log T}, 
	\end{eqnarray*}
	where 
	{\small \begin{eqnarray*}
		\learnconstant = \left(\statedim+\controldim\right) \left( \frac{\statedim}{\log \episodetime{}} + \frac{\noisedim \Mnorm{\BMcoeff{}}{2}^2 }{\eigmin{\BMcoeff{}\BMcoeff{}^{\top}}}\right), ~~~~~~~~~~~~~
		\algoconstant = \frac{(\episodetime{}-1) \Mnorm{\BMcoeff{}}{2}^2 \Mnorm{\RiccSol{\truth}}{2}^6 \Mnorm{\Rmat}{2} }{ \eigmin{\Qmat}^2 \eigmin{\Rmat}^4 } ~~\learnconstant .
	\end{eqnarray*}}
\end{thrm}

Theorem~\ref{BoundsThm} indicates efficiency of Algorithm~\ref{algo1}: At time $T$, the sub-optimality gap is as small as $\order{\algoconstant T^{-1/2} \log T }$. It also provides $\algoconstant, \learnconstant$ that reflect the dependence of estimation error and regret on different parameters in the problem. So, the regret scales linearly with the number of unknown parameters in $\truth$, while the estimation error dwindles linearly with the dimension. 

To establish Theorem~\ref{BoundsThm}, we study the learning step in~\eqref{RandomLSE1} and determine the rates Algorithm~\ref{algo1} estimates $\truth$. For that purpose, we prove concentration bounds for the empirical covariance matrices of the state vectors and also show anti-concentration of the Gram matrix $ \itointeg{0}{t}{\statetwo{s} \statetwo{s}^{\top}}{s}$ of the signal $\statetwo{s} = \left[\state{s}^\top,\action{s}^\top \right]^\top$, as $t$ grows. Then, we establish bounds on the comparative ratios of stochastic integrals and use that for controlling the estimation error. Furthermore, we show that the optimal feedback matrices have a Lipschitz property with respect to the dynamics matrices and leverage that for finding the deviation rates from the optimal feedback. Finally, we utilize policy differentiation and Theorem~\ref{GeneralRegretThm} for getting the regret bounds in Theorem~\ref{BoundsThm}.

Note that for obtaining descending estimation errors and sub-linear regret bounds we need $\eigmin{\BMcoeff{}\BMcoeff{}^{\top}}>0$. This is a standard requirement in estimation and control of stochastic linear systems and expresses that all coordinates of the state vectors are randomized by the Brownian motion $\left\{\BM{t}\right\}_{t\geq 0}$ in a relatively short time~\citep{levanony2001persistent,subrahmanyam2019identification,caines2019stochastic}. So, all state variables have significant roles in the dynamics. From a modeling point of view, $\eigmin{\BMcoeff{}\BMcoeff{}^{\top}}>0$ indicates that the stochastic differential equation in~\eqref{dynamics} is irreducible in the sense that a smaller subset of state variables is \emph{insufficient} for capturing the stochastic dynamical behavior of the environment. 

To close this section, observe that the estimation error of Algorithm~\ref{algo1} shrinks as $T^{-1/4}$. So, it does not decay with the ideal square-root rate because the main priority of Algorithm~\ref{algo1} is to minimize its regret by exploring minimally. However, if the randomization matrices $\randommatrix{n}$ are persistent and their standard deviations do not diminish as $n$ grows, then we obtain the square-root consistency. This is formalized in Proposition~\ref{FullEstimationProp} in Appendix~\ref{appA}. Of course, the compromise is that $\regret{T}{\policy}$ grows \emph{linearly} with $T$ if $\randommatrix{n}$ does not dwindle as $n$ grows.

\section{Numerical Illustrations of Estimation Error and Regret} \label{NumSection}
Now, we provide numerical analyses for showcasing the performance of Algorithm~\ref{algo1} for estimating the unknown dynamics matrices and learning the optimal policy. For this purpose, we assume that the true continuous-time system matrices are lateral-directional state-space matrices of X-29A airplane at $4000$ ft altitude~\citep{bosworth1992linearized}. The system is of dimension $\statedim=4$, is controlled by two dimensional commands; $\controldim=2$, and the transition and input matrices in \eqref{dynamics} are 
{\small \begin{equation*}
	\Amat{\star} =\begin{bmatrix}
	-0.1850 & 0.1475 & -0.9825 &  0.1120 \\
	-0.3467 & -1.710 & 0.9029 &   -0.5843 \times 10^{-6}\\
	1.174 &  -0.0825   & -0.1826 &  -0.4428 \times 10^{-7}\\
	0.0 &    1.0 &    0.1429 &   0.0
	\end{bmatrix},
	\Bmat{\star} =\begin{bmatrix}
	-0.4470 \times 10^{-3} &   0.4020 \times 10^{-3}\\
	0.3715        &   0.0549 \\
	0.0265        &   -0.0135 \\
	0.0           &   0.0 
	\end{bmatrix}.
	\end{equation*}}
Note that the dimensions of the control action and the state vector, as well as open-loop \emph{instability} of $\Amat{\star}$ and the small entries of $\Bmat{\star}$, render control of the above-mentioned airplane challenging. Further, we let the coefficient matrix of the Brownian disturbance $\BM{t}$ be $\BMcoeff{}=0.2 \times I_4$, and employ Algorithm~\ref{algo1} to learn to control a quadratic cost with the weight matrices $\Qmat=10 \times I_{\statedim}, \Rmat= I_{\controldim}$. The online reinforcement learning policy of Algorithm~\ref{algo1} is run for $500$ seconds, while the parameter estimates are updated at times $t=\episodetime{n}$, for integer values of $n$ and $\episodetime{}=1.2$. To find an initial stabilizing feedback, we run a Bayesian learning algorithm for $25$ seconds~\citep{faradonbeh2021bayesian}. 

\begin{figure}[t!] 
	\centering
	\scalebox{.22}
	{\includegraphics{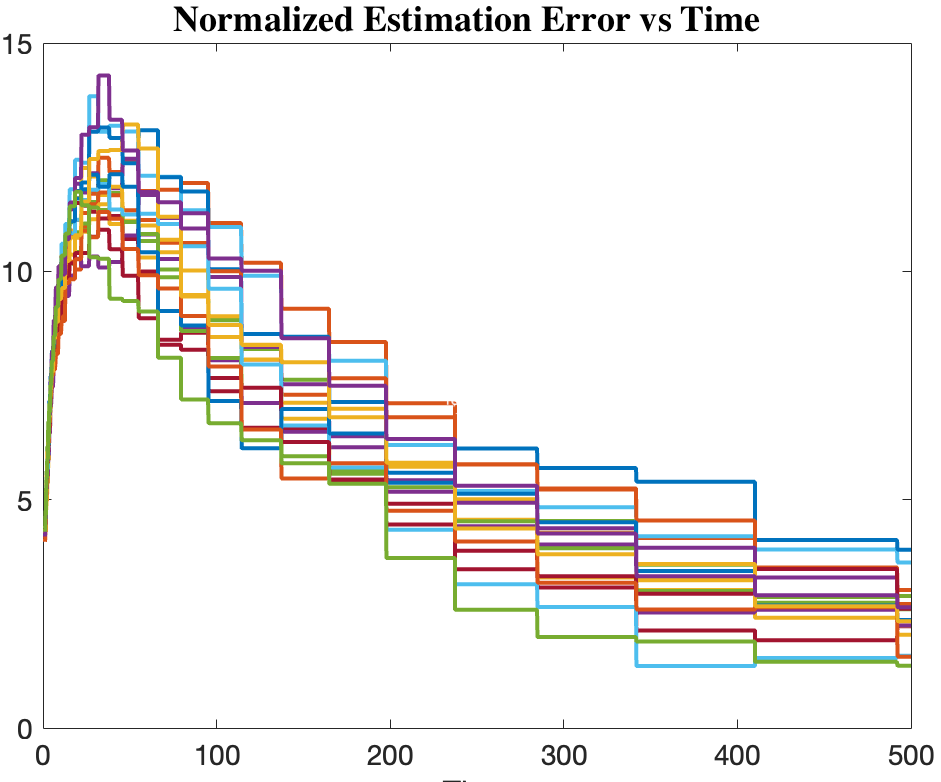} \includegraphics{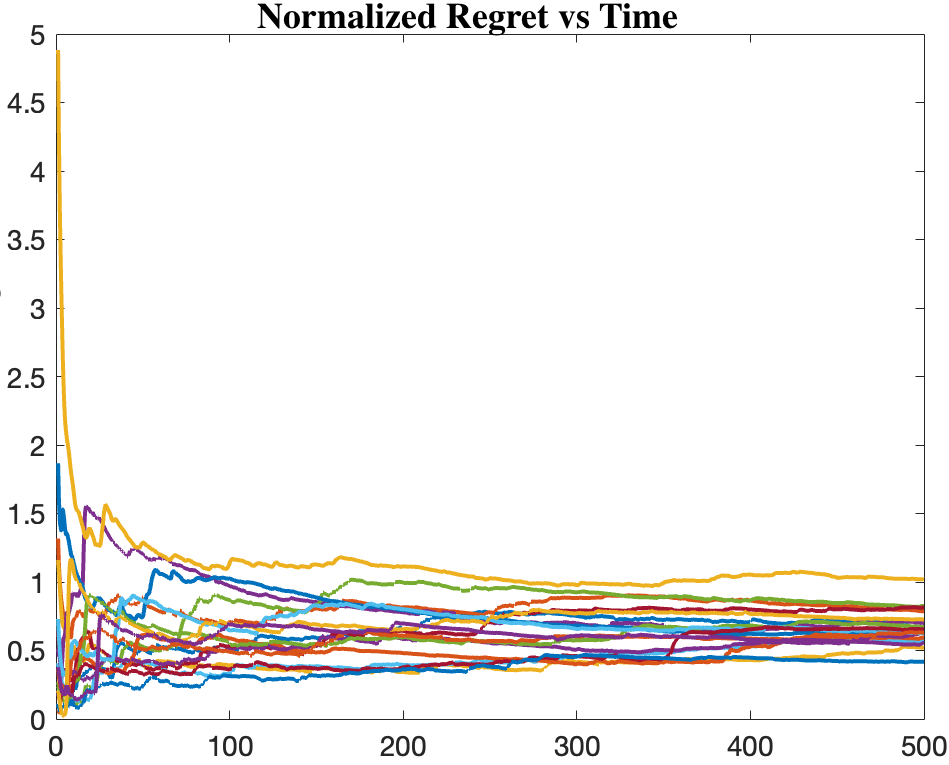}}  
	\caption{\textbf{Left:} Normalized estimation error $T^{1/2} \Learnerror{}{\estpara{n}}^2$ is plotted vs $T$, such that $\episodetime{n} \leq T < \episodetime{n+1}$, for some integer $n$. Multiple replicates of the normalized estimation error are reported in the graph, which clearly remain bounded as time grows. Therefore, the graph depicts Theorem~\ref{BoundsThm} about the rates Algorithm~\ref{algo1} learns the dynamics matrices. \newline \textbf{Right:} The graph presents curves of the normalized regret $T^{-1/2}\regret{T}{\policy}$ versus time $T$, while $\policy$ is the policy in Algorithm~\ref{algo1}. Multiple replicates of the system are simulated, all corroborating Theorem~\ref{BoundsThm} that the normalized regret remains (almost) bounded.}
	\label{ErrorFig}\label{RegretFig}
\end{figure}

The normalized rates of the estimation error are plotted in Figure~\ref{ErrorFig}, versus the continuous time $T$. To illustrate the rates in Theorem~\ref{BoundsThm}, the figure contains multiple trajectories of $T^{1/2} \Learnerror{}{\estpara{n}}^2$, while $n,T$ satisfy $\episodetime{n} \leq T < \episodetime{n+1}$. The normalized estimation errors in the left panel are (almost) bounded as time grows, corroborating Theorem~\ref{BoundsThm}. The right panel depicts normalized regret versus time: the horizontal axis is $T$ and the vertical one represents $T^{-1/2} \regret{T}{\policy}$, where $\policy$ is the online reinforcement learning policy in Algorithm~\ref{algo1}. Again, it is clear that the statement of Theorem~\ref{BoundsThm} holds, as the normalized regret remains bounded as time grows. 

\section{Concluding Remarks} \label{ConclusionSection}
We studied online reinforcement learning policies for unknown continuous-time stochastic linear systems and presented algorithms that learn to minimize quadratic costs. Three important problems are fully investigated, followed by the intuitions and implications of the presented analyses. 

First, we studied stabilization of stochastic linear systems based on inexact dynamics matrices and proved Theorem~\ref{StabThm} that specifies the coarse-grained accuracy for guaranteeing stability. Then, proposing the novel approach of policy differentiation, we established a reciprocal result in Theorem~\ref{GeneralRegretThm} for the regret that policies cause by taking sub-optimal actions. More importantly, we presented the online reinforcement learning Algorithm~\ref{algo1} and established its performance guarantees. Indeed, Theorem~\ref{BoundsThm} expresses that the estimation error rate of Algorithm~\ref{algo1} is ${\totaldim T^{-1/4} \log^{1/2}T}$ and it enjoys the efficient regret bound $\order{\totaldim^2 T^{1/2} \log T}$, where $T$ is the time and $\totaldim$ is the dimension.

As an initiating paper on design and analysis of online reinforcement learning policies for continuous-time stochastic systems, this study introduces interesting directions for future work. That includes establishing 
regret lower-bounds, investigating high-dimensional systems with structured dynamics such as low-rank or sparse matrices, and designing efficient policies under imperfect state-observations. Another interesting avenue for future studies that the authors expect the presented techniques apply to, is that of learning to control systems with non-linear dynamics or arbitrary cost functions.

\addcontentsline{toc}{section}{References}


\newpage
\onehalfspacing
\label{toc}
\tableofcontents
\newpage
\doublespacing
\appendices

\addcontentsline{toc}{section}{Appendices}
\section{Proof of Theorem~\ref{OptimalityProof} (Optimal policy)}

Fixing $\epsilon>0$, suppose that the control inputs $\action{t}$ are frozen in intervals of length $\epsilon$ and can change only at times $k\epsilon$, for $k=0,1, \cdots$. That is, for all times $t$ satisfying $k\epsilon \leq t < (k+1) \epsilon$, the action vector is fixed; $\action{t}=\action{k \epsilon}$. Next, we proceed towards finding a decision-making policy for minimizing the expected average cost. Note that due to the above-mentioned freezing during $\epsilon$-length intervals, the resulting decision-making policies can be sub-optimal, and indeed provide an upper bound for the optimal cost value. However, we will address this possible sub-optimality at the end of the proof, and with a slight abuse of notation, we still use $\optimalpolicy$ to denote the above-mentioned policy.

Next, fix an arbitrary time horizon $T$, and denote the minimum cost-to-go at time $t$ by 
\begin{eqnarray*}
	\valuefunc{t}{\state{t}} = \inf \E{\itointeg{t}{T}{\instantcost{\optimalpolicy}{s}}{s} \Bigg| \filter{t}},
\end{eqnarray*}
where the infimum is taken over non-anticipating policies that freeze the control action in $\epsilon$-length intervals, as elaborated above, and the information at time $t$ is $\filter{t}=\sigfield{\state{0:t},\action{0:t}}$; the sigma-field generated by the state and action vectors up to the time. Now, finding an optimal policy is equivalent to applying dynamic programming principle and writing Bellman optimality equations~\citep{kumar2015stochastic,chen1995linear}. So, we have 	\begin{eqnarray} \label{OptProofEq1}
	\valuefunc{k\epsilon}{\state{k \epsilon}} = \min\limits_{\action{k\epsilon}}  \E{\itointeg{k\epsilon}{(k+1)\epsilon}{\instantcost{\optimalpolicy}{\state{t},\action{k\epsilon}}}{t} + \valuefunc{(k+1)\epsilon}{\state{(k+1)\epsilon}} \Bigg| \filter{k\epsilon}},
	\end{eqnarray}
subject to the dynamics equation in~\eqref{dynamics}. 

For the sake of simplicity, suppose that $T/\epsilon$ is an integer. Solving \eqref{OptProofEq1} for $k= T/ \epsilon -1 $, we get the optimal control action $\action{k\epsilon}^{\star}=0$.  Accordingly, this gives 
\begin{eqnarray*}
	\valuefunc{(k+1)\epsilon}{\state{(k+1)\epsilon}}=\state{(k+1)\epsilon}^{\top} \Qmat \state{(k+1)\epsilon} \epsilon,
\end{eqnarray*}
for $k=T/\epsilon-2$, which, after substituting in \eqref{OptProofEq1}, becomes 
	\begin{eqnarray} \label{OptProofEq2}
	\valuefunc{k\epsilon}{\state{k \epsilon}} = \min\limits_{\action{k\epsilon}} \itointeg{k\epsilon}{(k+1)\epsilon}{ \E{\state{t}^{\top} \Qmat \state{t} \Big| \filter{k\epsilon}} }{t} + \action{k\epsilon}^{\top} \Rmat \action{k\epsilon} \epsilon + \E{ \state{(k+1)\epsilon}^{\top} \Qmat \state{(k+1)\epsilon} \Bigg| \filter{k\epsilon}} \epsilon,
	\end{eqnarray}	

where we applied Fubini's Theorem to derive 
	\begin{eqnarray} \label{SuppAuxEq2}
	\E{\itointeg{k\epsilon}{(k+1)\epsilon}{\instantcost{\optimalpolicy}{\state{t},\action{k\epsilon}^{}}}{t}  \Bigg| \filter{k\epsilon}} = \itointeg{k\epsilon}{(k+1)\epsilon}{ \E{\state{t}^{\top} \Qmat \state{t} \Big| \filter{k\epsilon}} }{t} + \action{k\epsilon}^{\top} \Rmat \action{k\epsilon} \epsilon.
	\end{eqnarray}

However, solving the dynamics~\eqref{dynamics} for $k\epsilon \leq t \leq (k+1)\epsilon$, we obtain
\begin{eqnarray*}
	\state{t}=e^{\Amat{\star}(t-k\epsilon)} \state{k\epsilon} + \itointeg{k\epsilon}{t}{e^{ \Amat{\star} (t-s)} \BMcoeff{} }{\BM{s}} + \itointeg{k\epsilon}{t}{e^{ \Amat{\star} (t-s)} }{s} \Bmat{\star}\action{k\epsilon} ,
\end{eqnarray*}
which together with Ito's Lemma, ${\diff \BM{s} \diff \BM{s}^{\top}} = I_{\noisedim} \diff s$ \citep{oksendal2013stochastic}, yields to 	\begin{eqnarray*}
	&&\E{\state{t}^{\top} \Qmat \state{t} \Big| \filter{k\epsilon}} = \itointeg{k\epsilon}{t}{ \tr{ e^{\Amat{\star}^{\top} (t-s)} \Qmat e^{\Amat{\star} (t-s)} \BMcoeff{} \BMcoeff{arg1}^{\top} } }{s} 
	\notag \\
	&+& \left( e^{\Amat{\star}(t-k\epsilon)} \state{k\epsilon} + \itointeg{k\epsilon}{t}{e^{ \Amat{\star} (t-s)} }{s} \Bmat{\star}\action{k\epsilon} \right)^{\top} \Qmat \left( e^{\Amat{\star}(t-k\epsilon)} \state{k\epsilon} + \itointeg{k\epsilon}{t}{e^{ \Amat{\star} (t-s)} }{s} \Bmat{\star}\action{k\epsilon} \right). \label{SuppAuxEq3}
	\end{eqnarray*}

Plugging these results in the dynamic programming equation in~\eqref{OptProofEq2}, the expression in front of the minimum becomes the following quadratic function of $\action{k\epsilon}$:
\begin{eqnarray*}
	&& \state{k\epsilon}^{\top} \auxQ \state{k\epsilon} + 2 \state{k\epsilon}^{\top} \auxG \action{k\epsilon} + \action{k\epsilon}^{\top} \auxR \action{k\epsilon} \\
	&+& \left( \auxA \state{k\epsilon} + \auxB \action{k\epsilon} \right)^{\top} P_{k+1} \left( \auxA \state{k\epsilon} + \auxB \action{k\epsilon} \right) + \tr{\auxP_{k+1} \BMcoeff{}\BMcoeff{}^{\top}},
\end{eqnarray*}
where $P_{k+1}=\Qmat \epsilon$, and 
\begin{eqnarray*}
	\auxA &=& e^{\Amat{\star}\epsilon},\\
	\auxB &=& \itointeg{0}{\epsilon}{e^{\Amat{\star}s} }{s} \Bmat{\star},\\
	\auxQ &=& \itointeg{0}{\epsilon}{ e^{\Amat{\star}^{\top}t } \Qmat e^{\Amat{\star}t }  }{t},\\
	\auxG &=& \itointeg{0}{\epsilon}{ e^{\Amat{\star}^{\top}t} \Qmat \left( \itointeg{0}{t}{ e^{\Amat{\star}(t-s) }  }{s} \right) \Bmat{\star} }{t},\\
	\auxR &=& \Rmat \epsilon + \itointeg{0}{\epsilon}{  \Bmat{\star}^{\top} \left( \itointeg{0}{t}{ e^{\Amat{\star}^{\top}(t-s) } \Qmat e^{\Amat{\star}(t-s) }  }{s} \right)  \Bmat{\star} }{t}, \\
	\auxP_{k+1} &=& P_{k+1} \itointeg{0}{\epsilon}{ e^{\Amat{\star}^{\top} s} e^{ \Amat{\star} s } }{s} + \itointeg{0}{\epsilon}{ \left( \itointeg{0}{t}{ e^{\Amat{\star}^{\top} s} \Qmat e^{ \Amat{\star} s } }{s} \right) }{t}.
\end{eqnarray*} 
Note that in the last equation above, we used Ito Isometry~\citep{baldi2017stochastic} to find $\auxP_{k+1}$. Now, performing the minimization the optimal control action is
\begin{eqnarray*}
	\action{k\epsilon}^{\star} = - \left( \auxB^{\top} P_{k+1} \auxB + \auxR \right)^{-1} \left( \auxB^{\top} P_{k+1} \auxA + \auxG^{\top} \right) \state{k\epsilon},
\end{eqnarray*}
and \eqref{OptProofEq2} leads to 
	\begin{eqnarray} \label{OptProofEq3}
	\valuefunc{k\epsilon}{\state{k\epsilon}} = \state{k\epsilon}^{\top} P_{k} \state{k\epsilon}+ \tr{ \BMcoeff{}\BMcoeff{}^{\top} \left[ P_{k+1} \itointeg{0}{\epsilon}{ e^{\Amat{\star}^{\top} s} e^{ \Amat{\star} s } }{s} + \itointeg{0}{\epsilon}{ \left( \itointeg{0}{t}{ e^{\Amat{\star}^{\top} s} \Qmat e^{ \Amat{\star} s } }{s} \right) }{t} \right] },
	\end{eqnarray}

where $P_{k}$ is calculated according to the discrete time Riccati equation 	
\begin{eqnarray} \label{OptProofEq4}
	P_{k} = \auxQ + \auxA^{\top} P_{k+1} \auxA - \left( \auxG + \auxA^{\top} P_{k+1} \auxB \right) \left( \auxB^{\top} P_{k+1} \auxB + \auxR \right)^{-1} \left( \auxB^{\top} P_{k+1} \auxA + \auxG^{\top} \right).
	\end{eqnarray}

It is shown that if there is some matrix $\Gainmat{}$ such that $\eigmax{\auxA+\auxB \Gainmat{}}<1$, then as $k \to -\infty$, the matrix $P_{k}$ in the above discrete time Riccati equation converges to a uniquely existing matrix $P$ that solves the algebraic Riccati equation 	\begin{eqnarray} \label{OptProofEq5}
	P = \auxQ + \auxA^{\top} P \auxA - \left( \auxG + \auxA^{\top} P \auxB \right) \left( \auxB^{\top} P \auxB + \auxR \right)^{-1} \left( \auxB^{\top} P \auxA + \auxG^{\top} \right),
	\end{eqnarray}	

regardless of the terminal matrix for the largest value $k+1$, which here corresponds to $P_{T/\epsilon}$~\citep{chan1984convergence,de1986riccati,faradonbeh2018bfinite}. 

Next, we show that if $\epsilon$ is sufficiently small, then the matrix $\Gainmat{}$ mentioned above exists. To that end, write
\begin{eqnarray*}
	\auxA &=& e^{\Amat{\star}\epsilon} =  \sum\limits_{n=0}^\infty \frac{\Amat{\star}^{n} \epsilon^{n} }{n!} = I_{\statedim} + \epsilon M(\epsilon) \Amat{\star}, \\
	\auxB &=& \sum\limits_{n=0}^\infty  \itointeg{0}{\epsilon}{\frac{\Amat{\star}^{n} s^{n} }{n!} }{s} \Bmat{\star} =  \sum\limits_{n=0}^\infty  \frac{\Amat{\star}^{n} \epsilon^{n+1} }{(n+1)!} \Bmat{\star} = \epsilon M(\epsilon) \Bmat{\star},
\end{eqnarray*}
where
\begin{eqnarray*}
	M(\epsilon) = \sum\limits_{n=1}^\infty \frac{\Amat{\star}^{n-1} \epsilon^{n-1} }{n!} = I_{\statedim}+ \epsilon \sum\limits_{n=2}^\infty \frac{\Amat{\star}^{n-1} \epsilon^{n-2} }{n!}.
\end{eqnarray*}
Then, letting $\Gainmat{}$ be as in Assumption~\ref{StabAssump}, if $\epsilon$ is small enough, it holds that 
\begin{eqnarray} \label{OptProofEq6}
\mosteig{M(\epsilon) \left( \Amat{\star}+ \Bmat{\star} \Gainmat{} \right)} < 0.
\end{eqnarray}
That is because the eigenvalues of the matrix $M(\epsilon) \left( \Amat{\star}+ \Bmat{\star} \Gainmat{} \right)$ are continuous functions of $\epsilon$, and for $\epsilon=0$ we have $\mosteig{M(0) \left( \Amat{\star}+ \Bmat{\star} \Gainmat{} \right)} = \mosteig{ \left( \Amat{\star}+ \Bmat{\star} \Gainmat{} \right)} < 0$, according to Assumption~\ref{StabAssump}. Hence, $\auxA+ \auxB \Gainmat{} = I_{\statedim} + M(\epsilon)\left(\Amat{\star}+ \Bmat{\star} \Gainmat{}\right) \epsilon$ implies that eigenvalues of $\auxA+ \auxB \Gainmat{} $ are exactly one plus the eigenvalues of $M(\epsilon)\left(\Amat{\star}+ \Bmat{\star} \Gainmat{}\right) \epsilon $. So, it holds that
	\begin{eqnarray} \label{OptProofEq7}
	\eigmax{\auxA+ \auxB \Gainmat{}}^2 \leq 1 + 2 \epsilon \mosteig{M(\epsilon)\left(\Amat{\star}+ \Bmat{\star} \Gainmat{}\right)} + \eigmax{M(\epsilon)\left(\Amat{\star}+ \Bmat{\star} \Gainmat{}\right)}^2 \epsilon^2.
	\end{eqnarray}

Now, putting \eqref{OptProofEq6} and \eqref{OptProofEq7} together, if $\epsilon$ is small enough, then $\eigmax{\auxA+ \auxB \Gainmat{}} <1$. Henceforth, suppose that $\epsilon$ is sufficiently small so that the latter inequality holds true.

As long as $\epsilon>0$ is small enough as described above, letting the time horizon $T$ tend to infinity, the $\epsilon$-length frozen optimal policy for minimizing the expected average cost is
\begin{eqnarray}
\action{k\epsilon}^{\star} = - \left( \auxB^{\top} P \auxB + \auxR \right)^{-1} \left( \auxB^{\top} P \auxA + \auxG^{\top} \right) \state{k\epsilon}, 
\end{eqnarray}
where $P$ is the unique solution of~\eqref{OptProofEq5}. On the other hand, for a fixed time horizon $T$, as $\epsilon$ shrinks the discrete-time Riccati equation in~\eqref{OptProofEq4} becomes a continuous-time Riccati equation as follows. First, we have
\begin{eqnarray*}
	\lim\limits_{\epsilon \to 0} \frac{\auxA - I_{\statedim}}{\epsilon} &=& \Amat{\star},\:\:\:\:\:\:\: \\
	\lim\limits_{\epsilon \to 0} \frac{\auxB}{\epsilon} &=& \Bmat{\star},\:\:\:\:\:\:\:\\
	\lim\limits_{\epsilon \to 0} \frac{\auxQ}{\epsilon} &=& \Qmat,\:\:\:\:\:\:\:\\
	\lim\limits_{\epsilon \to 0} \frac{\auxG}{\epsilon} &=& 0,\:\:\:\:\:\:\:\\
	\lim\limits_{\epsilon \to 0} \frac{\auxR}{\epsilon} &=& \Rmat.		
\end{eqnarray*}
Using these limits, letting $\epsilon \to 0$ in~\eqref{OptProofEq4} leads to 
	\begin{eqnarray*}
	\lim\limits_{\epsilon \to 0} \frac{P_{k}-P_{k+1}}{\epsilon} &=& \lim\limits_{\epsilon \to 0} \frac{\auxQ}{\epsilon} + \lim\limits_{\epsilon \to 0} \frac{\auxA^{\top} P_{k+1} \auxA - P_{k+1}}{\epsilon} \notag \\
	&-& \lim\limits_{\epsilon \to 0} \left( \frac{\auxG + \auxA^{\top} P_{k+1} \auxB}{\epsilon} \right) \left( \frac{\auxB^{\top} P_{k+1} \auxB + \auxR}{\epsilon} \right)^{-1} \left( \frac{\auxB^{\top} P_{k+1} \auxA + \auxG^{\top}}{\epsilon} \right) \\
	&=& \Qmat + \Amat{\star}^{\top} P_{k+1} + P_{k+1} \Amat{\star} - P_{k+1} \Bmat{\star} \Rmat^{-1} \Bmat{\star}^{\top} P_{k+1}. \label{SuppAuxEq4} 
	\end{eqnarray*}

That is, the backward differential equation
\begin{eqnarray} \label{OptProofEq8}
-\frac{\diff P_t }{\diff t} = \MatOpAve{\truth}{P},
\end{eqnarray}
with the terminal condition $P_T=0$. Thus, as $\epsilon \to 0$, the optimal policy becomes
\begin{eqnarray*}
	\action{t}^{\star} = -\Rmat^{-1} \Bmat{\star}^{\top} P_t \state{t},
\end{eqnarray*}
where $P_t$ is the solution of~\eqref{OptProofEq8}. Similarly, letting $\epsilon \to 0$ in~\eqref{OptProofEq5}, we get the optimal policy $\action{t}^{\star} = \Optgain{\truth} \state{t}$ for minimizing the infinite horizon expected average cost, where 
\begin{eqnarray*}
	\Optgain{\truth}=-\Rmat^{-1} \Bmat{\star}^{\top} \RiccSol{\truth},
\end{eqnarray*}
and $\RiccSol{\truth}$ solves $\MatOpAve{\truth}{P}=0$. Equivalently, letting $P_{t,T}$ be the solution of~\eqref{OptProofEq8} when the time horizon is $T$, it holds that $\lim\limits_{T \to \infty} P_{0,T} = \RiccSol{\truth}$, where $\RiccSol{\truth}$ solves $\MatOpAve{\truth}{P}=0$. Note that all these relationships rely on the convergence of discrete time Riccati equation~\eqref{OptProofEq4} to the algebraic Riccati equation~\eqref{OptProofEq5}, as $T \to \infty$.

Next, subtracting $\valuefunc{(k+1)\epsilon}{\state{k\epsilon}}$ from both sides of~\eqref{OptProofEq1}, dividing by $\epsilon$, and letting $\epsilon \to 0$, Ito Isomery implies that
\begin{eqnarray*}
	&& -\frac{\partial \valuefunc{t}{\state{t}} }{\partial t} \diff t = \min\limits_{\action{t}} \instantcost{\optimalpolicy}{\state{t},\action{t}} \diff t
	+ \E{ \diff \state{t}^{\top} \frac{\partial \valuefunc{t}{\state{t}}}{\partial \state{t}} + \frac{1}{2} \diff \state{t}^{\top} \frac{\partial^2 \valuefunc{t}{\state{t}}}{\partial \state{t} \partial \state{t}^{\top}} \diff \state{t} \Bigg| \filter{t}},
\end{eqnarray*} 
where we used the limits of the matrices $\auxQ,\auxG,\auxR$ as $\epsilon \to 0$ to find the expression on the right-hand-side of the above equality. Note that the above partial derivatives exist according to~\eqref{OptProofEq3} together with Dominated Convergence Theorem. Hence, substituting for $\diff \state{t}$ from the dynamics~\eqref{dynamics}, and leveraging Ito's Lemma, we obtain the Hamilton-Jacobi-Bellman \citep{yong1999stochastic} equation
\begin{eqnarray} \label{HJB}
-\frac{\partial \valuefunc{t}{\state{t}} }{\partial t} = \min\limits_{\action{t}} \instantcost{\optimalpolicy}{\state{t},\action{t}} 
+ \frac{\partial \valuefunc{t}{\state{t}}^{\top}}{\partial \state{t}} \left( \Amat{\star} \state{t} + \Bmat{\star} \action{t} \right) + \frac{1}{2} \tr { \frac{\partial^2 \valuefunc{t}{\state{t}}}{\partial \state{t} \partial \state{t}^{\top} } \BMcoeff{} \BMcoeff{}^{\top} }. 
\end{eqnarray} 
Further, letting $\epsilon \to 0$, the expression in~\eqref{OptProofEq3} gives
\begin{eqnarray} \label{OptProofEq9}
\valuefunc{t}{\state{t}} = \state{t}^{\top} P_t \state{t} + \itointeg{t}{T}{ \tr{ \BMcoeff{} \BMcoeff{}^{\top} P_s} } {s},
\end{eqnarray}
where $P_t$ solve ~\eqref{OptProofEq8}. This can be equivalently obtained using the fact that a quadratic function of the form $\valuefunc{t}{\state{t}}=\state{t}^{\top} F_t \state{t} + \varphi_t$ solves the partial differential equation~\eqref{HJB}, as long as 
\begin{eqnarray*}
	-\frac{ \diff \varphi_t }{\diff t} - \state{t}^{\top} \frac{\diff F_t}{\diff t} \state{t} &=&  \min\limits_{\action{t}} \state{t}^{\top} \Qmat \state{t} + \action{t}^{\top} \Rmat \action{t} \\
	&+&  2\state{t}^{\top} F_t \left( \Amat{\star} \state{t} + \Bmat{\star} \action{t} \right) +  \tr { F_t \BMcoeff{} \BMcoeff{}^{\top} },
\end{eqnarray*}
which after solving for $\action{t}$ gives the optimal policy $\action{t}^{\star}= -\Rmat^{-1} \Bmat{\star}^{\top} F_t \state{t}$, as well as
\begin{eqnarray*}
	-\frac{ \diff \varphi_t }{\diff t} - \state{t}^{\top} \frac{\diff F_t}{\diff t} \state{t} &=&  \state{t}^{\top} \Qmat \state{t} + 2\state{t}^{\top} F_t \left( \Amat{\star} \state{t} \right) \\
	&-& \state{t}^{\top} F_t \Bmat{\star} \Rmat^{-1} \Bmat{\star}^{\top} F_t \state{t} +  \tr { F_t \BMcoeff{} \BMcoeff{}^{\top} }.
\end{eqnarray*}
Because the equation above needs to hold for an arbitrary $\state{t}$, it splits to 
\begin{eqnarray*}
	-\frac{\diff F_t}{\diff t} = \MatOpAve{\truth}{F_t}, \frac{ \diff \varphi_t }{\diff t} = -\tr { F_t \BMcoeff{} \BMcoeff{}^{\top} },
\end{eqnarray*}
that is, $F_t$ solves~\eqref{OptProofEq8}. Further, note that cost-to-go at time $T$ is zero because time-to-go is zero, which provides the terminal condition $\valuefunc{T}{\state{T}}=0$, implying that $\varphi_t=\itointeg{t}{T}{ \tr{ \BMcoeff{} \BMcoeff{}^{\top} F_s} } {s}$. Therefore, the solutions $F_t,\varphi_t$ of \eqref{HJB} lead to the same expression as in~\eqref{OptProofEq9}.

Finally, the expected average cost of the policy $\action{t}=\Optgain{\truth}\state{t}$ is the limit of the expected average cost of the policy $\action{t}=-\Rmat^{-1} \Bmat{\star}^{\top} P_{t,T} \state{t}$, as $T \to \infty$;
\begin{eqnarray*}
	&& \limsup\limits_{T \to \infty} \frac{1}{T} \E{\itointeg{0}{T}{ \instantcost{\optimalpolicy}{s} } {s}} \\
	&=&
	\limsup\limits_{T \to \infty} \frac{1}{T} \itointeg{0}{T}{ \tr{ \BMcoeff{} \BMcoeff{}^{\top} P_{s,T}} } {s} \\
	&=& \tr{ \BMcoeff{} \BMcoeff{}^{\top}  \lim\limits_{T \to \infty} P_{s,T}} \\
	&=& \tr{ \BMcoeff{} \BMcoeff{}^{\top} \RiccSol{\truth}}. 
\end{eqnarray*}

Moreover, suppose that $\BMcoeff{}=0$, and apply the policy $\action{t}=\Optgain{\truth} \state{t}$. Then, the state trajectory becomes $\state{t} = e^{\CLmat{\star}t} \state{0}$, where $\CLmat{\star}=\Amat{\star}+\Bmat{\star}\Optgain{\truth}$. So, by~\eqref{OptProofEq9}, we have
\begin{eqnarray*}
	&&\state{0}^{\top} \RiccSol{\truth} \state{0} \\&=& \itointeg{0}{\infty}{ \state{t}^{\top} \left(\Qmat + \Optgain{\truth}^{\top} \Rmat \Optgain{\truth} \right)  \state{t} }{t} \\
	&=& \state{0}^{\top} \itointeg{0}{\infty}{ e^{\CLmat{\star}^{\top}t} \left(\Qmat + \Optgain{\truth}^{\top} \Rmat \Optgain{\truth} \right)  e^{\CLmat{\star}t} }{t} \state{0},
\end{eqnarray*}
for an arbitrary initial state $\state{0}$. Thus, \eqref{LyapInteg} holds:
\begin{eqnarray} \label{LyapInteg2}
	\RiccSol{\truth} = \itointeg{0}{\infty}{ e^{\CLmat{\star}^{\top}t} \left(\Qmat + \Optgain{\truth}^{\top} \Rmat \Optgain{\truth} \right)  e^{\CLmat{\star}t} }{t}.
\end{eqnarray}
Since $\Qmat$ is positive definite, the above equality implies that $\mosteig{\CLmat{\star}} < 0$, as well as 
\begin{eqnarray} 
&& \CLmat{\star}^{\top}\RiccSol{\truth} + \RiccSol{\truth} \CLmat{\star} \notag \\
&+& \Qmat + \Optgain{\truth}^{\top} \Rmat \Optgain{\truth} =0. \label{LyapDifferential}
\end{eqnarray}

So far, we have shown that by restricting our search for an optimal decision-making policy to the class of policies that the control action is frozen during intervals of length $\epsilon$, and then letting $\epsilon$ decay to vanish, we obtain optimal policies given by~\eqref{OptProofEq8}. Next, we show that these policies are optimal in the larger class of all control policies satisfying the information criteria at every time. That is, for all $t$, the control action $\action{t}$ can be determined using $\filter{t}=\sigfield{\state{0:t},\action{0:t}}$. For this purpose, first note that the decision-making policy $\action{t}=\Rmat^{-1} \Bmat{\star}^{\top} \RiccSol{\truth} \state{t}$ provides an upper-bound for the optimal expected average cost. That is,
\begin{eqnarray*}
	\inf_{\policy} \avecost{\policy}{} \leq \tr{ \BMcoeff{} \BMcoeff{}^{\top} \RiccSol{\truth}}.
\end{eqnarray*}

Now, suppose that there is another policy, denoted by $\auxpolicy$, that satisfies \\$\avecost{\auxpolicy}{} \leq \tr{ \BMcoeff{} \BMcoeff{}^{\top} \RiccSol{\truth}}$. Define cost-to-go of the policy $\auxpolicy$ by
\begin{eqnarray*}
	\auxvaluefunc{t}{\state{t}} = \E{\itointeg{t}{T}{\instantcost{\auxpolicy}{s}}{s} \Bigg| \filter{t}},
\end{eqnarray*}
where $T$ is large enough to satisfy \\
$\auxvaluefunc{t}{\state{t}} \leq 2 \state{t}^{\top} \RiccSol{\truth} \state{t}+ 2 T \tr{ \BMcoeff{} \BMcoeff{}^{\top} \RiccSol{\truth}}$, for all $0 \leq t \leq 1$. Note that such $T$ exists since $\auxpolicy$ provides a smaller expected average cost than the policy $\action{t}=\Rmat^{-1} \Bmat{\star}^{\top} \RiccSol{\truth} \state{t}$, and the desired upper-bound for $\auxvaluefunc{t}{\state{t}}$ is $2\valuefunc{t}{\state{t}}$; two times the cost-to-go of the policy $\action{t}=\Rmat^{-1} \Bmat{\star}^{\top} \RiccSol{\truth} \state{t}$. Next, writing
\begin{eqnarray*}
	\auxvaluefunc{t}{\state{t}} =   \E{\itointeg{t}{t+\epsilon}{\instantcost{\auxpolicy}{\state{s},\action{s}}}{s} + \auxvaluefunc{t+\epsilon}{\state{t+\epsilon}} \Bigg| \filter{t}},
\end{eqnarray*}
subtract $\auxvaluefunc{t+\epsilon}{\state{t}}$ from both sides, and divide by $\epsilon$. Letting $\epsilon$ decay to zero, the upper-bound for $\auxvaluefunc{t}{\state{t}}$ in terms of $\valuefunc{t}{\state{t}}$ implies that according to Dominated Convergence Theorem, the following derivatives exist and it holds that
\begin{eqnarray*}
	&& -\frac{\partial \auxvaluefunc{t}{\state{t}} }{\partial t} = \instantcost{\optimalpolicy}{\state{t},\auxpolicy\left( \filter{t} \right)}
	\\ &+& \frac{\partial \auxvaluefunc{t}{\state{t}}^{\top}}{\partial \state{t}} \left( \Amat{\star} \state{t} + \Bmat{\star} \auxpolicy\left( \filter{t} \right) \right) + \frac{1}{2} \tr { \frac{\partial^2 \auxvaluefunc{t}{\state{t}}}{\partial \state{t} \partial \state{t}^{\top} } \BMcoeff{} \BMcoeff{}^{\top} }.
\end{eqnarray*}
Now, note that since $\instantcost{\optimalpolicy}{t}$ as well as $\Bmat{\star}\action{t}$ are continuous functions of $\action{t}$, the above partial differential equation for $\auxvaluefunc{t}{\state{t}}$ indicates that $\auxpolicy\left( \filter{t} \right)$ is a continuous function of $\state{t}$. This, together with the fact that $\BM{t}$ is an almost surely continuous function of time $t$, in lights of the dynamics equation in~\eqref{dynamics}, leads to continuity of state trajectory $\state{t}$; i.e., $\action{t}=\auxpolicy\left( \filter{t} \right)$ is continuous as $t$ varies. Thus, decision-making policies that freeze for $\epsilon$-length intervals provide accurate approximations of $\action{t}=\auxpolicy\left( \filter{t} \right)$ in a sense that there exists a sequence $\left\{\action{t}^{(n)}\right\}_{n=1}^\infty$ such that $\action{t}^{(n)}$ freezes during intervals of the length $1/n$, and it holds that
\begin{eqnarray*}
	\limsup\limits_{n \to \infty} \limsup\limits_{T \to \infty} \frac{1}{T} \itointeg{0}{T}{ \E{ \norm{\action{t}^{(n)} - \auxpolicy\left( \filter{t} \right)}{2} } } {t} =0.
\end{eqnarray*}
Therefore, we have
\begin{eqnarray*}
	\avecost{\auxpolicy}{} \geq \inf\limits_{\epsilon>0} \inf \avecost{\policy}{} = \tr{\RiccSol{\truth} \BMcoeff{} \BMcoeff{}^{\top}},
\end{eqnarray*}
where the inner infimum is taken over all policies that freeze during $\epsilon$-length intervals. This shows that the policy $\action{t}=-\Rmat^{-1} \Bmat{\star}^{\top} \RiccSol{\truth} \state{t}$ is an optimal one, which completes the proof.	

\newpage
\section{{Proof of Theorem~\ref{StabThm}} (Stability margin)}

First, we study eigenvalues of the sum of two matrices. Suppose that $M, \erterm{}$ are arbitrary square matrices of the same size, and let $M= P^{-1} \Lambda P$ be the Jordan decomposition. That is, $\lambda_1, \cdots, \lambda_k$ are eigenvalues of $M$, $\Lambda \in \C^{\statedim \times \statedim}$ is a block diagonal matrix with blocks ${\Lambda_1,\cdots, \Lambda_k}$, and 
\begin{eqnarray} \label{JordanBlocks} 
\Lambda_i= \begin{bmatrix}
\lambda_i & 1 & 0 & \cdots & 0 & 0 \\
0 & \lambda_i & 1 & 0 & \cdots & 0 \\
\vdots & \vdots & \vdots & \vdots & \vdots & \vdots \\
0 & 0 & \cdots & 0 & \lambda_i & 1 \\
0 & 0 & 0 & \cdots & 0 & \lambda_i
\end{bmatrix} \in \C^{\mult{i} \times \mult{i}}.
\end{eqnarray}
Further, similar to Definition~\ref{MultDef}, let $\mult{M}=\max\limits_{1 \leq i \leq k} \mult{i}$. We prove that $\mosteig{M - \erterm{}}$ is at most
\begin{equation} \label{EigPerturb}
\mosteig{M} +  \mult{M}^{1/2} \Mnorm{P\erterm{}P^{-1}}{2} \vee \left( \mult{M}^{1/2} \Mnorm{P\erterm{}P^{-1}}{2} \right)^{1/\mult{M}}.
\end{equation}

To show the above inequality, first let $\lambda$ be an eigenvalue of $M-\erterm{}$ that satisfies $\Re \left(\lambda\right) > \mosteig{M}$. So, $ M - \lambda I $ is an invertible matirx, and there exists at least one vector $v$, such that $v \neq 0$ and $\left( M - \erterm{} - \lambda I \right)P^{-1}v = 0$. Then, $\left( M - \lambda I \right) P^{-1}v = \erterm{} P^{-1}v$ implies that
\begin{equation} \label{StabPerturbProofEq1}
v 
= \left( \Lambda - \lambda I \right)^{-1} P \erterm{} P^{-1}v.
\end{equation}
Because $\Lambda=\diag{\Lambda_1, \cdots, \Lambda_k}$, the matrix $\Lambda - \lambda I$ is block diagonal as well, and we have
$\left( \Lambda - \lambda I \right)^{-1} = \diag{\left( \Lambda_1 - \lambda I_{\mult{1}} \right)^{-1}, \cdots, \left( \Lambda_k - \lambda I_{\mult{k}} \right)^{-1}}$.
Further, it is straightforward to see that $\left( \Lambda_i - \lambda I_{\mult{i}} \right)^{-1}$ is
\begin{equation*}
- \begin{bmatrix}
\left( \lambda -\lambda_i \right)^{-1} & \left( \lambda -\lambda_i \right)^{-2}  & \cdots & \left( \lambda -\lambda_i \right)^{-\mult{i}} \\
0 & \left( \lambda -\lambda_i \right)^{-1} & \cdots & \left( \lambda -\lambda_i \right)^{-\mult{i} +1} \\
\vdots & \vdots & \vdots & \vdots \\
0 & \cdots & 0 & \left( \lambda -\lambda_i \right)^{-1}
\end{bmatrix}.
\end{equation*}
Therefore, we have 
$$\Mnorm{\left( \Lambda_i - \lambda I_{\mult{i}} \right)^{-1}}{2} \leq \mult{i}^{1/2}  \left( \left| \lambda - \lambda_i \right| \wedge \left| \lambda - \lambda_i \right|^{\mult{i}}  \right)^{-1}.$$
Using this bound for the operator norms of blocks of the block-diagonal matrix $\left( \Lambda - \lambda I \right)^{-1}$, since $\mult{i} \leq \mult{M}$ and $\Re\left(\lambda\right)>\mosteig{M}$, the equation in~\eqref{StabPerturbProofEq1} implies
\begin{eqnarray*}
	1 &\leq& \Mnorm{\left( \Lambda - \lambda I \right)^{-1} P \erterm{}P^{-1}}{2}
	\leq \Mnorm{\left( \Lambda - \lambda I \right)^{-1}}{2} \Mnorm{P\erterm{}P^{-1}}{2} \\
	&\leq& 
	\mult{M}^{1/2} \Mnorm{P\erterm{}P^{-1}}{2} \left( \left( \Re (\lambda) - \mosteig{M} \right) \wedge \left( \Re (\lambda) - \mosteig{M} \right)^{\mult{M}}  \right)^{-1}. \label{StabProofEq0}
\end{eqnarray*}

So, letting $\lambda$ be an eigenvalue of $M-\erterm{}$ that satisfies $\Re (\lambda) = \mosteig{M - \erterm{}}$, we obtain~\eqref{EigPerturb}.

Now, using~\eqref{EigPerturb}, we compare $\Amat{\star}+\Bmat{\star}\Optgain{\estpara{}}$ and $\estD{}=\estA{}+\estB{}\Optgain{\estpara{}}$. Since $\Amat{\star}+\Bmat{\star}\Optgain{\estpara{}} - \estD{}$ is 
\begin{equation} \label{StabThmProofEq00} 
\erterm{\star} = \Amat{\star}-\estA{} - \left( \Bmat{\star}-\estB{} \right) \Rmat^{-1} \estB{} \RiccSol{\estpara{}},
\end{equation}
using \eqref{StabCond1}, and letting $M=\estD{}$ in \eqref{EigPerturb}, we have
\begin{eqnarray*}
	&& \mosteig{\Amat{\star}+\Bmat{\star}\Optgain{\estpara{}}} \leq
	-\stabradii 
	+ \mult{\estD{}}^{1/2} \Mnorm{P^{-1}}{2} \Mnorm{P}{2} \Mnorm{\erterm{\star}}{2} \vee \left( \mult{\estD{}}^{1/2} \Mnorm{P^{-1}}{2} \Mnorm{P}{2} \Mnorm{\erterm{\star}}{2} \right)^{1/\mult{\estD{}}}.
\end{eqnarray*}
So, in order to have $\mosteig{\Amat{\star}+\Bmat{\star}\Optgain{\estpara{}}} < -\delta$, it suffices to show that
\begin{equation} \label{StabThmProofEq0}
\mult{\estD{}}^{1/2} \Mnorm{P^{-1}}{2} \Mnorm{P}{2} \Mnorm{\erterm{\star}}{2}  < \stabradii - \delta \wedge \left( \stabradii - \delta\right)^{\mult{\estD{}}}.
\end{equation}
However, since $\Mnorm{\erterm{\star}}{2} \leq \Learnerror{}{\estpara{}} \left( 1 \vee \frac{\Mnorm{\estB{}}{2} \RiccUpperBound}{\eigmin{\Rmat}}  \right)$,
\eqref{StabNeighborhood} provides~\eqref{StabThmProofEq0}, which leads to the desired result.

\subsection{Proof of sufficiency of \eqref{StabRadiEq0} for stabilization bounds}
Next, we show that $\Learnerror{}{\estpara{}} \leq \epsilon_0$ is sufficient for stabilization and express uniform bounds for $\stabradii,\RiccUpperBound$ in \eqref{StabCond1}. Let $\CLmat{\star}=\Amat{\star}+\Bmat{\star}\Optgain{\truth}=P_\star^{-1} \Lambda_\star P_\star$ be the Jordan decomposition as defined in the beginning of the proof, and define the largest block size $\mult{\star}=\mult{\CLmat{\star}}$, similar to Definition~\ref{MultDef}. Further, suppose that the following is satisfied:
\begin{equation} \label{StabRadiEq1} 
	\epsilon_0 \leq \frac{1}{1 \vee \Mnorm{\Optgain{\truth}}{2} } \left( \frac{ \left( -\mosteig{\CLmat{\star}} \right) \wedge \left( -\mosteig{\CLmat{\star}}\right)^{\mult{\star}} }{\mult{\star}^{1/2} \Mnorm{P_\star^{-1}}{2} \Mnorm{P_\star}{2} }
	\wedge \left[4 \itointeg{0}{\infty}{ \Mnorm{e^{\CLmat{\star}t}}{2}^2 }{t} \right]^{-1} \right) . 
\end{equation}

The inequality in \eqref{StabRadiEq1} implies that if we write  $\CLmat{1}=\estA{}+\estB{}\Optgain{\truth} = \Amat{\star}+ \Bmat{\star} \Optgain{\truth} + \erterm{1}=\CLmat{\star}+\erterm{1}$,
then, the matrix $\erterm{1}=\estA{}-\Amat{\star}+\left( \estB{}-\Bmat{\star} \right)\Optgain{\truth}$ satisfies 
\begin{equation*}
\Mnorm{\erterm{1}}{2} < \frac{\left( -\mosteig{\CLmat{\star}} \right) \wedge \left( -\mosteig{\CLmat{\star}} \right)^{\mult{\star}}}{\mult{\star}^{1/2} \Mnorm{P_\star}{2} \Mnorm{P_\star^{-1}}{2} }.
\end{equation*}
So, taking $M=\CLmat{\star}$, the bound in \eqref{EigPerturb} implies that $\mosteig{\CLmat{1}}<0$. Hence, we can employ Lemma~\ref{LyapLemma} to study consequences of applying the linear feedback $\Optgain{\truth}$ to a system of dynamics matrices $\estpara{}$, and get 
$$\RiccSol{\estpara{}} \leq M= \RiccSol{\estpara{}} + \itointeg{0}{\infty}{ e^{\CLmat{1}^{\top} t} F e^{\CLmat{1}t} }{t},$$
where
$$F= \left[ \Optgain{\truth} - \Optgain{\estpara{}} \right]^{\top} \Rmat \left[ \Optgain{\truth} - \Optgain{\estpara{}} \right].$$
Above, we used the fact that the initial state $\state{0}=x$ in Lemma~\ref{LyapLemma} is arbitrary, and so, the involved matrices are themselves equal. Further, similar to Lemma~\ref{LyapLemma}, it is straightforward to see that 
$$M = \itointeg{0}{\infty}{ e^{\CLmat{1}^{\top} t} \left[ \Qmat+ \Optgain{\truth}^{\top} \Rmat \Optgain{\truth} \right] e^{\CLmat{1}t} }{t}.$$
This leads to
\begin{eqnarray*}
	&& \Qmat+ \Optgain{\truth}^{\top} \Rmat \Optgain{\truth} = -\CLmat{1}^{\top} M - M \CLmat{1} \\
	&=& -\CLmat{\star}^{\top} M - M \CLmat{\star} - \erterm{1}^{\top} M - M \erterm{1}.
\end{eqnarray*} 
Because $\mosteig{\CLmat{\star}} <0$, the latter equation and~\eqref{LyapInteg} provide 

\begin{eqnarray*}
	M &=& \itointeg{0}{\infty}{ e^{\CLmat{\star}^{\top} t} \left[ \Qmat+ \Optgain{\truth}^{\top} \Rmat \Optgain{\truth} + \erterm{1}^{\top} M + M \erterm{1} \right] e^{\CLmat{\star}t} }{t} \\
	&=& \itointeg{0}{\infty}{ e^{\CLmat{\star}^{\top} t} \left[ \Qmat+ \Optgain{\truth}^{\top} \Rmat \Optgain{\truth} \right] e^{\CLmat{\star}t} }{t} + \itointeg{0}{\infty}{ e^{\CLmat{\star}^{\top} t} \left[ \erterm{1}^{\top} M + M \erterm{1} \right] e^{\CLmat{\star}t} }{t} \\
	&=& \RiccSol{\truth} + \itointeg{0}{\infty}{ e^{\CLmat{\star}^{\top} t} \left[ \erterm{1}^{\top} M + M \erterm{1} \right] e^{\CLmat{\star}t} }{t}. \label{StabThmProofEq1} 
\end{eqnarray*}

Therefore, it holds that $\Mnorm{M}{2} \leq \Mnorm{\RiccSol{\truth}}{2} + 2 \Mnorm{\erterm{1}}{2} \Mnorm{M}{2} \itointeg{0}{\infty}{ \Mnorm{e^{\CLmat{\star}t}}{2}^2 }{t}$,
which, according to~\eqref{StabRadiEq1} and $\RiccSol{\estpara{}} \leq M$, yields to
\begin{equation} \label{StabThmProofEq2}
\Mnorm{\RiccSol{\estpara{}}}{2} \leq \Mnorm{M}{2} \leq 2 \Mnorm{\RiccSol{\truth}}{2}. 
\end{equation}

To proceed, suppose that $v \in \C^{\statedim}$ satisfies $\norm{v}{2}=1$ and $\estD{}v = \lambda v$. Now, \eqref{LyapInteg} implies that
\begin{eqnarray*}
	v^*\RiccSol{\estpara{}}v 
	&=& \itointeg{0}{\infty}{ v^* e^{\estD{}^{\top} t} \left[ \Qmat + \Optgain{\estpara{}}^{\top} \Rmat \Optgain{\estpara{}} \right] e^{\estD{} t}v }{t} \\
	&=& \itointeg{0}{\infty}{ \norm{\left[ \Qmat + \Optgain{\estpara{}}^{\top} \Rmat \Optgain{\estpara{}} \right]^{\frac{1}{2}} e^{\lambda t}v}{2}^2 }{t} ,
\end{eqnarray*}
where $v^* $ is the transposed complex conjugate of $v$. Thus, maximizing the left-hand-side above while taking minimum on the right-hand-side, it holds that
\begin{equation} \label{StabThmProofEq3}
\Mnorm{\RiccSol{\estpara{}}}{2} \geq \eigmin{\Qmat} \itointeg{0}{\infty}{ e^{2 \Re \left( \lambda \right) t} }{t} \geq \frac{\eigmin{\Qmat} }{2 \Re \left( -\lambda \right) }.
\end{equation}
Putting \eqref{StabThmProofEq2} and \eqref{StabThmProofEq3} together, we obtain $\mosteig{\estD{}} \leq -\eigmin{\Qmat} \left(4\Mnorm{\RiccSol{\truth}}{2}\right)^{-1}$. This and~\eqref{StabThmProofEq2} imply that $\Learnerror{}{\estpara{}} \leq \epsilon_0$ is sufficient for~\eqref{StabCond1}, with $\stabradii=\eigmin{\Qmat} 4^{-1} \Mnorm{\RiccSol{\truth}}{2}^{-1}$, $\RiccUpperBound=2 \Mnorm{\RiccSol{\truth}}{2}$. ~\hfill~$\blacksquare$

\newpage
\section{{Proof of Theorem~\ref{GeneralRegretThm}} (Regret analysis)}
Let $M=\Qmat+ \Optgain{\truth}^{\top} \Rmat \Optgain{\truth}$. Recall that $\policy$ applies $\action{t} = \Gainmat{t} \state{t}$ at time $t$. Now, for a given $T$, suppose that $\epsilon>0$ is a fixed small real, and let $N=\lceil T / \epsilon\rceil$. Then, define the sequence of policies $\left\{ \policy_i \right\}_{i=0}^N$:
\begin{equation*}
\policy_i = \begin{cases}
\action{t} = \Gainmat{t} \state{t} & t < i\epsilon\\
\action{t} = \Optgain{\truth} \state{t} & t \geq i\epsilon
\end{cases}.
\end{equation*} 
Note that as long as one concerns about times $t \leq T$, it holds that $\optimalpolicy = \policy_0, \policy_N=\policy$. Clearly, since $\regret{T}{\policy_0}=0$, we have $\regret{T}{\policy} = \sum\limits_{i=0}^{N-1} \left( \regret{T}{\policy_{i+1}} - \regret{T}{\policy_i} \right)$. Thus, Lemma~\ref{LocalRegretLem} gives $\regret{T}{\policy} = \sum\limits_{i=0}^{N-1} \left(\state{i\epsilon}^{\top} F_{i\epsilon} \state{i\epsilon} + 2\state{i\epsilon}^{\top} g_{i\epsilon} + \ssconstant_{i\epsilon}\right)$,
where the matrix $F_{i\epsilon}$, the vector $g_{i\epsilon}$, and the scalar $\ssconstant_{i\epsilon}$ are defined in \eqref{LongEq1}, \eqref{LongEq2}, and \eqref{LongEq3}, respectively. Now, letting $\epsilon\to 0$, since $\Gainmat{t}$ is piecewise continuous, we have
\begin{equation} \label{GenRegThmProofEq0}
\regret{T}{\policy} = \itointeg{0}{T}{ \left(\state{t}^{\top} \widetilde{F}_{t} \state{t} + 2\state{t}^{\top} \widetilde{g}_{t} + \widetilde{\ssconstant}_{t}\right) }{t},
\end{equation}
where $\widetilde{F}_{t} = \lim\limits_{\epsilon\to 0 , i\epsilon\to t} {\epsilon}^{-1} F_{i\epsilon}$, $\widetilde{g}_{t} = \lim\limits_{\epsilon\to 0 , i\epsilon\to t} {\epsilon}^{-1} g_{i\epsilon}$, and $\widetilde{\ssconstant}_{t} = \lim\limits_{\epsilon\to 0 , i\epsilon\to t} {\epsilon}^{-1} \ssconstant_{i\epsilon}$. Note that the above limits exist, since $F_{i\epsilon},g_{i\epsilon},\ssconstant_{i\epsilon}$ are continuous. To calculate $\widetilde{F}_t,\widetilde{g}_t, \widetilde{\ssconstant}_t$, using Lemma~\ref{LocalRegretLem} and the piecewise continuity of $\Gainmat{t}$, we obtain $\widetilde{\ssconstant}_{t} = 0$, 
\begin{eqnarray*}
	\widetilde{F}_{t} &=& S_t + 2 H_t^{\top} \itointeg{t}{T}{ e^{\CLmat{\star}^{\top}(s-t)} M e^{\CLmat{\star}(s-t)} }{s} , \\
	\widetilde{g}_{t} &=& \itointeg{t}{T}{ \left(  H_t^{\top} e^{\CLmat{\star}^{\top}(s-t)} M \itointeg{t}{s}{ e^{\CLmat{\star}(s-u)}\BMcoeff{} }{\BM{u}} \right) }{s} ,
\end{eqnarray*}
where $S_t = \Gainmat{t}^{\top} \Rmat \Gainmat{t} - \Optgain{\truth}^{\top} \Rmat \Optgain{\truth}$, and 
\begin{eqnarray*}
	H_t &=& \lim\limits_{\epsilon\to 0} \frac{e^{\left(\Amat{\star}+\Bmat{\star}\Gainmat{t}\right)\epsilon} - e^{\CLmat{\star}\epsilon} }{\epsilon}  = \Bmat{\star} \left(\Gainmat{t}-\Optgain{\truth}\right).
\end{eqnarray*}
Now, by \eqref{LyapInteg} and $\itointeg{T}{\infty}{ e^{\CLmat{\star}^{\top}(s-t)} M e^{\CLmat{\star}(s-t)} }{s} =E_{T-t}$, the expression for $\widetilde{F}_t$ becomes
\begin{equation} \label{GenRegThmProofEq11}
S_t + H_t^{\top} \RiccSol{\truth} + \RiccSol{\truth} H_t - H_t^{\top} E_{T-t} - E_{T-t} H_t.
\end{equation}
So, after doing some algebra (see \eqref{LyapAuxEq}), we get 
\begin{eqnarray} 
S_t &+& H_t^{\top} \RiccSol{\truth} + \RiccSol{\truth} H_t \notag \\
&=& \left(\Gainmat{t}-\Optgain{\truth}\right)^{\top} \Rmat \left(\Gainmat{t}-\Optgain{\truth}\right). \label{GenRegThmProofEq1}
\end{eqnarray}
Since $\BM{u}$ has independent increments and in $\widetilde{g}_t$ we have $u\geq t$, Fubini's Theorem gives
$$\E{\state{t}^{\top}\widetilde{g}_t} = \E{\E{\state{t}^{\top}\widetilde{g}_t \Big| \sigfield{\BM{0:t}}}} = \E{\state{t}^{\top} \E{\widetilde{g}_t \Big| \sigfield{\BM{0:t}}}} = 0.$$
 
Hence, \eqref{GenRegThmProofEq0}, \eqref{GenRegThmProofEq11}, \eqref{GenRegThmProofEq1}, and Fubini's Theorem imply that $\E{\regret{T}{\policy}}=\E{\regterm{T}}$.

To proceed towards establishing the second result, apply Stochastic Fubini Theorem~\citep{oksendal2013stochastic,baldi2017stochastic} to get
\begin{eqnarray*}
	\itointeg{0}{T}{\state{t}^{\top}\widetilde{g}_t}{t} &=& \itointeg{0}{T}{\itointeg{t}{T}{   \itointeg{t}{s}{ \left(\state{t}^{\top} H_t^{\top} e^{\CLmat{\star}^{\top}(s-t)} M e^{\CLmat{\star}(s-u)}\BMcoeff{} \right) }{\BM{u}} }{s} }{t} \notag \\
	&=& \itointeg{0}{T}{\itointeg{0}{u}{ \itointeg{u}{T}{ \left(\state{t}^{\top} H_t^{\top} e^{\CLmat{\star}^{\top}(s-t)} M e^{\CLmat{\star}(s-u)}\BMcoeff{} \right) }{s} }{t} }{\BM{u}}
	= \itointeg{0}{T}{ \statetwo{u}^{\top} }{\BM{u}}, \label{GenRegThmProofEq2}
\end{eqnarray*}

where, using the expression for $H_t$, the vector $\statetwo{u}$ can be written as  

$$\statetwo{u}^{\top} = \itointeg{0}{u}{ \itointeg{u}{T}{ \left(\state{t}^{\top} H_t^{\top} e^{\CLmat{\star}^{\top}(s-t)} M e^{\CLmat{\star}(s-u)}\BMcoeff{} \right) }{s} }{t} 
= \itointeg{0}{u}{ \left(\state{t}^{\top} \left( \Gainmat{t}-\Optgain{\truth} \right)^{\top} P_{t,u}^{\top} \right) }{t},$$

for $P_{t,u}^{\top} = \itointeg{u}{T}{ \Bmat{\star}^{\top} e^{\CLmat{\star}^{\top}(s-t)} M e^{\CLmat{\star}(s-u)}\BMcoeff{} }{s}$. Now, letting $\empiricalcovmat{T}=\itointeg{0}{T}{\norm{\statetwo{u}}{2}^2}{u}$, for $\empiricalcovmat{T}<1$, Ito Isometry~\citep{baldi2017stochastic}, and for $\empiricalcovmat{T} \geq 1$, Lemma~\ref{SelfNormalizedLem}, imply that 
\begin{equation} \label{GenRegThmProofEq3}
\itointeg{0}{T}{ \statetwo{u}^{\top} }{\BM{u}} = \order{ \noisedim \empiricalcovmat{T}^{1/2} \log^{1/2} \empiricalcovmat{T}}.
\end{equation}
However, by using the triangle inequality and Fubini's Theorem, we obtain
\begin{eqnarray*}
	\empiricalcovmat{T} &\leq& \itointeg{0}{T}{ \itointeg{0}{u}{ \norm{P_{t,u}  \left(\Gainmat{t}-\Optgain{\truth}\right) \state{t}}{2}^2 }{t} }{u} \\
	&=& 
	\itointeg{0}{T}{ \left( \state{t}^{\top} \left(\Gainmat{t}-\Optgain{\truth}\right)^{\top} \left[\itointeg{t}{T}{ P_{t,u}^{\top} P_{t,u} }{u}\right] \left(\Gainmat{t}-\Optgain{\truth}\right) \state{t} \right) }{t} \notag \\
	&\leq& \itointeg{0}{T}{ \eigmax{\itointeg{t}{T}{ \Rmat^{-1/2} P_{t,u}^{\top} P_{t,u} \Rmat^{-1/2} }{u}} \norm{\Rmat^{1/2}\left(\Gainmat{t}-\Optgain{\truth}\right) \state{t}}{2}^2 }{t}. \label{AuxEqTable1}
\end{eqnarray*}

The second part of the integrand above appears in $\regterm{T}$. So, we proceed by finding an upper-bound for the first part. For this purpose, we use the triangle inequality and~\eqref{LyapInteg} to get the equation
\begin{eqnarray*}
\eigmax{\itointeg{t}{T}{ P_{t,u}^{\top} P_{t,u} }{u}} &\leq& \itointeg{t}{T}{ \Mnorm{\Bmat{\star}^{\top} e^{\CLmat{\star}^{\top}(u-t)}}{2}^2 \Mnorm{ \itointeg{u}{T}{ e^{\CLmat{\star}^{\top}(s-u)} M e^{\CLmat{\star}(s-u)} }{s} }{2}^2 \Mnorm{\BMcoeff{}}{2}^2 }{u}  \\
&\leq& \Mnorm{\Bmat{\star}}{2}^2 \Mnorm{\RiccSol{\truth}}{2}^2 \Mnorm{\BMcoeff{}}{2}^2 \itointeg{0}{\infty}{ \Mnorm{ e^{\CLmat{\star}^{\top}u}}{2}^2  }{u}. ~~~~~ \label{AuxEqTable2}
\end{eqnarray*}

Therefore, by using~\eqref{StabThmProofEq3}, we get
\begin{equation*}
\empiricalcovmat{T} \leq \frac{\Mnorm{\Bmat{\star}}{2}^2 \Mnorm{\RiccSol{\truth}}{2}^3 \Mnorm{\BMcoeff{}}{2}^2}{\eigmin{\Qmat} \eigmin{\Rmat}} \regterm{T}, 
\end{equation*}
since $E_t$ decays exponentially with $t$. So, \eqref{GenRegThmProofEq3} gives the desired result. ~\hfill~$\blacksquare$

\newpage
\section{{Proof of Theorem~\ref{BoundsThm}} (Analysis of Algorithm~\ref{algo1})}
In order to establish Theorem~\ref{BoundsThm}, we study the estimation procedure in~\eqref{RandomLSE1} and specify the accuracy at which the algorithm is able to estimate $\truth$. To that end, Lemma~\ref{EmpCovLemma} and Lemma~\ref{OptManifoldLemma} are utilized to study the Gram matrix $\empiricalcovmat{n} = \itointeg{0}{\episodetime{n}}{\statetwo{s} \statetwo{s}^{\top}}{s}$ in~\eqref{RandomLSE1}, while Lemma~\ref{SelfNormalizedLem} is used for bounding the estimation error. Then, by leveraging Lemma~\ref{LipschitzLemma}, we find the rates of deviating from the optimal policy in~\eqref{OptimalPolicy}. Finally, the resulting regret of Algorithm~\ref{algo1} is investigated in lights of Theorem~\ref{GeneralRegretThm}.

By using~\eqref{dynamics} to substitute for $\diff \state{t}$, $\left[ \itointeg{0}{\episodetime{n}}{ \statetwo{s} }{\state{s}^{\top}} \right]^{\top} \empiricalcovmat{n}^\dagger$ is
\begin{equation*}
\left[ \itointeg{0}{\episodetime{n}}{ \statetwo{s}  \statetwo{s}^{\top} \left[\truth\right]^{\top}}{s} + \itointeg{0}{\episodetime{n}}{ \statetwo{s} }{\BM{s}^{\top}}  \BMcoeff{}^{\top} \right]^{\top} \empiricalcovmat{n}^\dagger.
\end{equation*}
In~\eqref{BoundThmProofEq2}, we show that $\empiricalcovmat{n}$ is non-singular. So, we have
\begin{equation} \label{EstErrorEq}
\left[ \itointeg{0}{\episodetime{n}}{ \statetwo{s} }{\state{s}^{\top}} \right]^{\top} \empiricalcovmat{n}^{-1} = \left[\truth\right] + \left[ \empiricalcovmat{n}^{-1} \itointeg{0}{\episodetime{n}}{ \statetwo{s} }{\BM{s}^{\top}}  \BMcoeff{}^{\top} \right]^{\top}.
\end{equation}
Because $\left[\truth\right] \in \paraspace{0}$, \eqref{RandomLSE1} and \eqref{EstErrorEq} lead to
\begin{equation*}
\Learnerror{}{\estpara{n}} \leq \Mnorm{\empiricalcovmat{n}^{-1} \itointeg{0}{\episodetime{n}}{ \statetwo{s} }{\BM{s}^{\top}}  \BMcoeff{}^{\top}}{2} + \Mnorm{\randommatrix{n}}{2}.
\end{equation*}
Since entries of $\randommatrix{n}$ are $\normaldist{0}{\episodetime{-n/2} n^{1/2}}$, we have 
$$\PP{\Mnorm{\randommatrix{n}}{2} \geq \statedim^{1/2} \left(\statedim+\controldim\right)^{1/2} \episodetime{-n/4}n^{1/2}} = \order{e^{-n^{1/2}}}.$$ 
This, by Borel-Cantelli Lemma, leads to $$\Mnorm{\randommatrix{n}}{2}=\order{\statedim^{1/2}\left(\statedim+\controldim\right)^{1/2}\episodetime{-n/4}n^{1/2}}.$$ 

Thus, letting $\totaldim=\statedim+\controldim$, according to Lemma~\ref{SelfNormalizedLem}, $\Learnerror{}{\estpara{n}}$ is at most
\begin{equation}\label{BoundThmProofEq1}
\totaldim^{1/2} \order{ \noisedim^{1/2} \Mnorm{\BMcoeff{}}{2} \left(\frac{\log \eigmax{\empiricalcovmat{n}}}{\eigmin{\empiricalcovmat{n}}}\right)^{1/2} + \statedim^{1/2} \episodetime{-n/4} n^{1/2}}.
\end{equation}
Now, Lemma~\ref{EmpCovLemma} provides $\order{\log \eigmax{\empiricalcovmat{n}}} = n \log \episodetime{}$. Further, we will shortly show that 
\begin{equation} \label{BoundThmProofEq2}
\liminf\limits_{n \to \infty} \episodetime{-n/2} \eigmin{\empiricalcovmat{n}} \geq \eigmin{\BMcoeff{}\BMcoeff{}^{\top}}. 
\end{equation}
Thus, \eqref{BoundThmProofEq1} and \eqref{BoundThmProofEq2} yield to the upper-bound
\begin{equation*}
\Learnerror{}{\estpara{n}} = \order{ \totaldim^{1/2} \left( \statedim^{1/2} + \frac{\noisedim^{1/2} \Mnorm{\BMcoeff{}}{2} \log^{1/2} \episodetime{}}{\eigmin{\BMcoeff{}\BMcoeff{}^{\top}}^{1/2}} \right) \episodetime{-n/4} n^{1/2} }.
\end{equation*}
This gives the first result in Theorem~\ref{BoundsThm}. To prove the other statement, let $\ssconstant_\star$ be as defined in Lemma~\ref{LipschitzLemma}. So, Lemma~\ref{LipschitzLemma} implies that 
$$\Mnorm{\Optgain{\estpara{n}} - \Optgain{\truth}}{2}^2 =\order{ \left( \statedim+\controldim \right) \ssconstant_\star^2 \left( \statedim + \frac{\noisedim \Mnorm{\BMcoeff{}}{2}^2 \log \episodetime{}}{\eigmin{\BMcoeff{}\BMcoeff{}^{\top}}} \right) \episodetime{-n/2} n }.$$

Now, since during the time period $\episodetime{n-1} \leq t < \episodetime{n}$ the feedback matrix is frozen to $\Optgain{\estpara{n-1}}$, according to Lemma~\ref{EmpCovLemma}, we have 
$$\itointeg{0}{\episodetime{n}}{ \norm{ \Rmat^{1/2} \left( \Gainmat{t} - \Optgain{\parameter{\star}} \right) \state{t} }{}^2 }{t} = \order{\sum\limits_{k=1}^{n} \ssconstant_{\Gainmat{}} \episodetime{k-1} \episodetime{-(k-1)/2} k } ,$$

where
\begin{equation*}
\ssconstant_{\Gainmat{}} = \totaldim \ssconstant_\star^2 \left( \statedim + \frac{\noisedim \Mnorm{\BMcoeff{}}{2}^2 \log \episodetime{}}{\eigmin{\BMcoeff{}\BMcoeff{}^{\top}}} \right)  \left(\episodetime{}-1\right) \Mnorm{\Rmat}{2} \Mnorm{\BMcoeff{}}{2}^2.
\end{equation*}
Moreover, since by Theorem~\ref{OptimalityProof} we have $\mosteig{\CLmat{\star}}<0$, the matrix $E_t$ in Theorem~\ref{GeneralRegretThm} decays exponentially with $t$. So, it holds that 
$${\itointeg{0}{T}{ \left( \state{t}^{\top} E_{T-t} \Bmat{\star} \left( \Gainmat{t}-\Optgain{\truth} \right) \state{t} \right) }{t}} = \order{\log^2 T}.
$$ 

Therefore, according to Theorem~\ref{GeneralRegretThm}, we have the following:
\begin{equation*}
\regret{T}{\policy} = \order{\sum\limits_{k=1}^{\lceil (\log T)/ (\log \episodetime{}) \rceil} \episodetime{(k-1)/2}k} = \order{ \frac{\ssconstant_{\Gainmat{}}}{\log \episodetime{}} T^{1/2} \log T }.
\end{equation*}
This, according to $\ssconstant_\star$ in Lemma~\ref{LipschitzLemma}, completes the proof.

To prove \eqref{BoundThmProofEq2}, let $\CLmat{k-1}=\Amat{\star}+\Bmat{\star} \Optgain{\estpara{k-1}}$. Then, by Lemma~\ref{EmpCovLemma}, we have
\begin{equation} \label{BoundThmProofEq3}
\liminf\limits_{k \to \infty} \episodetime{-k} \eigmin{ \itointeg{\episodetime{k-1}}{\episodetime{k}}{ \state{t}\state{t}^{\top} }{t}} 
\geq \eta_k \eigmin{\BMcoeff{}\BMcoeff{}^{\top}},
\end{equation}
where $\eta_k=\left( 1 - {\episodetime{}}^{-1} \right) \left(\itointeg{0}{1}{ \Mnorm{e^{-\CLmat{k-1}s}}{2}^2 }{s}\right)^{-1}$. Hence, \eqref{BoundThmProofEq3} implies that to establish \eqref{BoundThmProofEq2}, it suffices to show that the following inequality holds for some $0 \leq \ell <n-1$:
\begin{equation} \label{BoundThmProofEq4}
\liminf\limits_{n \to \infty} \eigmin{\sum\limits_{k=\ell}^{n-1} \episodetime{ k-n/2 } \begin{bmatrix}
	I_{\statedim} \\ \Optgain{\estpara{k}}
	\end{bmatrix} \begin{bmatrix}
	I_{\statedim} \\ \Optgain{\estpara{k}}
	\end{bmatrix}^{\top}   } \geq \max\limits_{\ell \leq k \leq n-1} \frac{1}{\eta_k}.
\end{equation}

For an arbitrary fixed $\epsilon>0$, consider the event that the above-mentioned smallest eigenvalue is less than $\epsilon$, and let $\manifold{M}_n(\epsilon)$ be the set of matrices $\left[\estpara{k}\right]_{k=\ell}^{n-1}$ for which this event occurs:
\begin{equation*}
\manifold{M}_n(\epsilon) = \left\{ \left[\estpara{\ell}, \cdots, \estpara{n-1}\right] : \eigmin{P_{\ell,n} P_{\ell,n}^{\top}} \leq \epsilon \right\},
\end{equation*} 
where the ${(\statedim+\controldim) \times \statedim (n-\ell)}$ matrix $P_{\ell,n}$ is
\begin{equation*}
\left[ \episodetime{\frac{\ell}{2}-\frac{n}{4}} \begin{bmatrix}
I_{\statedim} \\
\Optgain{\estpara{\ell}}
\end{bmatrix} , \cdots, \episodetime{\frac{n-1}{2}-\frac{n}{4}} \begin{bmatrix}
I_{\statedim} \\
\Optgain{\estpara{n-1}}
\end{bmatrix}\right].
\end{equation*}

Now, note that the set of all matrices 
\begin{equation*}
F_n = \begin{bmatrix}
\episodetime{\ell/2-n/4} I_{\statedim}  & \cdots & \episodetime{(n-1)/2-n/4} I_{\statedim} \\
\episodetime{\ell/2-n/4} \Gainmat{\ell} &  \cdots & \episodetime{(n-1)/2-n/4} \Gainmat{n-1}
\end{bmatrix},
\end{equation*}
that there exists $v \in \R^{\statedim+\controldim}$ satisfying $\norm{v}{2}=1$ and $F_n^{\top}v=0$, is of dimension $\statedim+\controldim-1+(n-\ell)(\controldim-1)$. To show that, on one hand, the set of unit $\statedim+\controldim$ dimensional vectors is (a sphere) of dimension $\statedim+\controldim-1$. On the other hand, by writing $v=\left[v_1^{\top},v_2^{\top}\right]^{\top}$, for $v_1 \in \R^{\statedim}$ and $v_2 \in \R^{\controldim}$, clearly, $F_n^{\top}v=0$ is equivalent to  $\Gainmat{k}^{\top}v_2=-v_1$, for all $k=\ell,\cdots, n-1$. The latter enforces every column of $\Gainmat{k}$ to be in a certain hyperplane in $\R^{\controldim}$.

Thus, according to Lemma~\ref{OptManifoldLemma}, the dimension of $\manifold{M}_n(0)$ is at most $\statedim + (\controldim-1) (n-\ell+1)+ (n-\ell) \statedim^2$. Further, if $\ell$ is sufficiently large so that $\episodetime{-\ell+n/2} \epsilon <1$, then for every $\left[\estpara{k}\right]_{k=\ell}^{n-1} \in \manifold{M}_n(\epsilon)$, there exists some $\left[\auxA_k,\auxB_k\right]_{k=\ell}^{n-1} \in \manifold{M}_n(0)$, such that for all $k=\ell,\cdots, n-1$, it holds that
\begin{equation*}
\Mnorm{\left[\estpara{k}\right] - \left[\auxA_k,\auxB_k\right]}{2} = \order{\episodetime{-k/2 +n/4} \epsilon^{1/2}}.
\end{equation*}
The random matrices $\left\{\randommatrix{k}\right\}_{k=0}^{n-1}$ are independent, and entries of $\randommatrix{k}$ are independent identically distributed $\normaldist{0}{\episodetime{-k/2}k^{1/2}}$ random variables. Hence, we have
\begin{equation*}
\PP{\manifold{M}_n(\epsilon)} = \left[ \order{\episodetime{\ell/4} \ell^{-1/4} \episodetime{-\ell/2 +n/4} \epsilon^{1/2}} \wedge 1 \right]^{m},
\end{equation*}
where $m=\left(\statedim \controldim - \controldim+1\right) (n-\ell) - \statedim-\controldim+1$. To see that, note that $\manifold{M}_n(0)$ is a $\statedim + (\controldim-1) (n-\ell+1)+ (n-\ell) \statedim^2$ dimensional object in a $\statedim (\statedim+\controldim)(n-\ell)$ dimensional space. So, the exponent is at least $m$.
Letting $\ell = n-5$, we have $m\geq 5$. Further, as $n$ grows, $ \order{\ell^{-1/4} \gamma^{(n-\ell)/4} \epsilon^{1/2} } < 1 $ holds for $\epsilon=\max\limits_{\ell \leq k \leq n-1} {\eta_k}^{-1}$. So, we have $\sum\limits_{n=5}^{\infty} \PP{\manifold{M}_n(\epsilon)} = \sum\limits_{n=5}^{\infty} \order{n^{-1/4}}^5 <  \infty$,
which by Borel-Cantelli Lemma implies \eqref{BoundThmProofEq4}. ~\hfill~$\blacksquare$


\newpage
\section{Estimation Rates under Persistent Randomization} \label{appA}
\begin{propo} \label{FullEstimationProp}
	Assume that in Algorithm~\ref{algo1} the variance of entries of $\randommatrix{n}$ is $\variance_n^2$, where 
	$$\liminf\limits_{n \to \infty} \variance_n > 0.$$
	Then, letting $\learnconstant$ be as in Theorem~\ref{BoundsThm}, $\statetwo{s} = \left[\state{s}^{\top}, \action{t}^{\top}\right]^{\top}$, and $\empiricalcovmat{n} = \itointeg{0}{\episodetime{n}}{\statetwo{s} \statetwo{s}^{\top}}{s}$, we have 
	\begin{equation*}
	\Mnorm{\left( \itointeg{0}{\episodetime{n}}{ \statetwo{s} }{\state{s}^{\top}} \right)^{\top} \empiricalcovmat{n}^\dagger - \left[ \truth \right]}{2}^2 = \order{\learnconstant \episodetime{-n} n^2}.
	\end{equation*}
\end{propo}

\begin{proof} 
	By~\eqref{EstErrorEq}, it suffices to study $\erterm{}=\empiricalcovmat{n}^{\dagger} \itointeg{0}{\episodetime{n}}{ \statetwo{s} }{\BM{s}^{\top}}  \BMcoeff{}^{\top}$. In the sequel, we show that 
$$\liminf\limits_{n \to \infty} n \episodetime{-n} \eigmin{\empiricalcovmat{n}} \geq \eigmin{\BMcoeff{}\BMcoeff{}^{\top}}.$$ 
So, putting Lemma~\ref{SelfNormalizedLem} and Lemma~\ref{EmpCovLemma} together, we obtain the desired result, since they give $$\Mnorm{\erterm{}}{2}^2=\order{ \left(\statedim+\controldim\right)  \noisedim \Mnorm{\BMcoeff{}}{2}^2  \frac{\episodetime{-n} n^{2} \log \episodetime{} }{\eigmin{\BMcoeff{}\BMcoeff{}^{\top}}} }.$$ 

Thus, by~\eqref{BoundThmProofEq3}, it is enough to show that for some $0 \leq \ell <n-1$, 
\begin{equation*}
\liminf\limits_{n \to \infty} \eigmin{\sum\limits_{k=\ell}^{n-1} \episodetime{ k-n } n \begin{bmatrix}
	I_{\statedim} \\ \Optgain{\estpara{k}}
	\end{bmatrix} \begin{bmatrix}
	I_{\statedim} \\ \Optgain{\estpara{k}}
	\end{bmatrix}^{\top}   }
\end{equation*}
is at least $\epsilon=\max\limits_{\ell \leq k \leq n-1} {\eta_k}^{-1}$. Let $\manifold{M}_n(\epsilon)$ be the set of $\left[\estpara{k}\right]_{k=\ell}^{n-1}$ that the above does not hold: $\manifold{M}_n(\epsilon) = \left\{ \left[\estpara{\ell}, \cdots, \estpara{n-1}\right] : \eigmin{P_{\ell,n} P_{\ell,n}^{\top}} \leq \epsilon \right\}$, where $P_{\ell,n} $ is
\begin{equation*}
\left[ \episodetime{\frac{\ell-n}{2}} n^{\frac{1}{2}} \begin{bmatrix}
I_{\statedim} \\
\Optgain{\estpara{\ell}}
\end{bmatrix} , \cdots, \episodetime{-\frac{1}{2}} n^{\frac{1}{2}} \begin{bmatrix}
I_{\statedim} \\
\Optgain{\estpara{n-1}}
\end{bmatrix}\right].
\end{equation*}

Similar to the proof of Theorem~\ref{BoundsThm}, for $\left[\estpara{k}\right]_{k=\ell}^{n-1} \in \manifold{M}_n(\epsilon)$, there is $\left[\auxA_k,\auxB_k\right]_{k=\ell}^{n-1} \in \manifold{M}_n(0)$, such that $\Mnorm{\left[\estpara{k}\right] - \left[\auxA_k,\auxB_k\right]}{2}^2 = \order{\episodetime{n-k} n^{-1} \epsilon}$. Thus, $\liminf\limits_{n \to \infty} \variance_n > 0$, together with the dimension of $\manifold{M}_n(0)$ that we calculated in the proof of Theorem~\ref{BoundsThm}, leads to
\begin{equation*}
\PP{\manifold{M}_n(\epsilon)} = \left[ \order{ \episodetime{(n-\ell)/2} n^{-1/2} \epsilon^{1/2}} \wedge 1 \right]^{m},
\end{equation*}
for $m=\left(\statedim \controldim - \controldim+1\right) (n-\ell) - \statedim-\controldim+1$. Finally, $\ell = n-4$ gives $\sum\limits_{n=4}^{\infty} \PP{\manifold{M}_n(\epsilon)}  <  \infty$. Therefore, Borel-Cantelli Lemma implies the desired result. 
\end{proof}

\newpage
\section{Auxiliary Lemmas}
In this section, we state the auxiliary lemmas used for establishing the main  results and provide their proofs, each subsection corresponding to one lemma. 

First, in Lemma~\ref{LocalRegretLem} in Subsection~\ref{appB}, we provide expressions for the difference between the regrets of two policies. Study of self-normalized stochastic integrals is the content of Lemma~\ref{SelfNormalizedLem}, while Lemma~\ref{LipschitzLemma} on Lipschitz continuity of the optimal feedback with respect to the dynamics matrices is established in Subsection~\ref{appD}. 

Next, in Lemma~\ref{LyapLemma}, we consider the total cumulative cost for the case of applying a sub-optimal time-invariant linear feedback policy to a deterministic system. Then, Lemma~\ref{EmpCovLemma} focuses on explicit calculation of the empirical covariance matrix of the state vectors. Finally, in Lemma~\ref{OptManifoldLemma} in Subsection~\ref{appF6} we specify the set of dynamics matrices that possess the same optimal linear feedback matrix.


\subsection{\bf Difference in regrets of two policies} \label{appB}
\begin{lemm} \label{LocalRegretLem}
	For fixed $0 \leq t_1 \leq t_2 \leq T$, define the policies $\policy_1,\policy_2$ according to
	\begin{equation*}
	\policy_i = \begin{cases}
	\action{t}=\Gainmat{} \state{t} & t < t_i \\
	\action{t}=\Optgain{\truth} \state{t} & t \geq t_i
	\end{cases}.
	\end{equation*}
	Further, let $\CLmat{\star}=\Amat{\star}+\Bmat{\star}\Optgain{\truth}$, $\CLmat{}=\Amat{\star}+\Bmat{\star}\Gainmat{}$, $M_\star = \Qmat+ \Optgain{\truth}^{\top} \Rmat \Optgain{\truth}$, $M = \Qmat + \Gainmat{}\Rmat \Gainmat{}$, $\erterm{t} = e^{\CLmat{}(t-t_1)} - e^{\CLmat{\star}(t-t_1)}$,
	$Z_t = \itointeg{t_1}{t}{ \left[ e^{\CLmat{}(t-s)} - e^{\CLmat{\star}(t-s)} \right] \BMcoeff{}}{\BM{s}}$, and
	$S = M-M_\star = \Gainmat{}^{\top} \Rmat \Gainmat{} - \Optgain{\truth}^{\top} \Rmat \Optgain{\truth}$. 
	
	Then, we have $\regret{T}{\policy_2} - \regret{T}{\policy_1} = \state{t_1}^{\top} F_{t_1} \state{t_1} + 2\state{t_1}^{\top} g_{t_1} + \ssconstant_{t_1}$, where $F_{t_1},g_{t_1}$, and $\ssconstant_{t_1}$ are 
	\begin{eqnarray}   
			F_{t_1} &=& \itointeg{t_1}{t_2}{ \left(e^{\CLmat{\star}^{\top}(t-t_1)} S e^{\CLmat{\star}(t-t_1)} + 2 \erterm{t}^{\top} M e^{\CLmat{\star}(t-t_1)} + \erterm{t}^{\top} M \erterm{t}\right) }{t} \notag \\
			&+& \itointeg{t_2}{T}{ \left(2 \erterm{t_2}^{\top} e^{\CLmat{\star}^{\top}(t-t_2)} M_{\star} e^{\CLmat{\star}(t-t_1)} + \erterm{t_2}^{\top} e^{\CLmat{\star}^{\top}(t-t_2)} M_\star e^{\CLmat{\star}(t-t_2)} \erterm{t_2}\right) }{t}, \label{LongEq1}
		\end{eqnarray}
		\begin{eqnarray}
			g_{t_1} &=& \itointeg{t_1}{t_2}{ \left( S \itointeg{t_1}{t}{ e^{\CLmat{\star}(t-s)} \BMcoeff{} }{\BM{s}} +  \erterm{t}^{\top} M \itointeg{t_1}{t}{ e^{\CLmat{\star}(t-s)} \BMcoeff{} }{\BM{s}} + e^{\CLmat{\star}^{\top}(t-t_1)} M Z_t  \right) }{t} \notag \\
			&+& \itointeg{t_1}{t_2}{ \erterm{t}^{\top} M Z_t }{t} + \itointeg{t_2}{T}{ \erterm{t_2}^{\top} e^{\CLmat{\star}^{\top}(t-t_2)} M_\star \left( e^{\CLmat{\star}(t-t_2)} Z_{t_2} + \itointeg{t_1}{t}{ e^{\CLmat{\star}(t-s)} \BMcoeff{} }{\BM{s}} \right) }{t} \notag \\
			&+& \itointeg{t_2}{T}{ \left(  e^{\CLmat{\star}^{\top}(t-t_1)} M_\star e^{\CLmat{\star}(t-t_2)} Z_{t_2} \right) }{t} ,\label{LongEq2} 
		\end{eqnarray}
		\begin{eqnarray}
			\ssconstant_{t_1} &=& \itointeg{t_1}{t_2}{ \left(\norm{S^{1/2} \itointeg{t_1}{t}{ e^{\CLmat{\star}(t-s)} \BMcoeff{} }{\BM{s}} }{2}^2 + 2 Z_t^{\top} M \itointeg{t_1}{t}{ e^{\CLmat{\star}(t-s)} \BMcoeff{} }{\BM{s}} + Z_t^{\top} M Z_t \right) }{t} \notag \\
			&+& \itointeg{t_2}{T}{ \left(2 Z_{t_2}^{\top} e^{\CLmat{\star}^{\top}(t-t_2)} M_\star \itointeg{t_1}{t}{ e^{\CLmat{\star}(t-s)} \BMcoeff{} }{\BM{s}} + Z_{t_2}^{\top} e^{\CLmat{\star}^{\top}(t-t_2)} M_\star e^{\CLmat{\star}(t-t_2)} Z_{t_2} \right) }{t} . \label{LongEq3}
		\end{eqnarray} 
\end{lemm}

\begin{proof}
	Letting $\state{t}^{\policy_i}$ be the state of the system under the policy $\policy_i$, clearly, for $t \leq t_1$, it holds that $\state{t}^{\policy_1}=\state{t}^{\policy_2}$. So, we use $\state{t_1}$ for both states at time $t_1$. Moreover, for $t_1 \leq t \leq t_2$, we have
	\begin{eqnarray*}
		\state{t}^{\policy_1} &=& e^{\CLmat{\star}(t-t_1)} \state{t_1} + \itointeg{t_1}{t}{e^{\CLmat{\star}(t-s)} \BMcoeff{}}{\BM{s}}, \\
		\state{t}^{\policy_2} &=& e^{\CLmat{}(t-t_1)} \state{t_1} + \itointeg{t_1}{t}{e^{\CLmat{}(t-s)} \BMcoeff{}}{\BM{s}},
	\end{eqnarray*}
	where $\statetwo{t} = \state{t}^{\policy_2} - \state{t}^{\policy_1}$. So, by denoting the instantaneous cost of policy $\policy_i$ at time $t$ by $\instantcost{\policy_i}{t}$, we get 
	$\statetwo{t} = \erterm{t} \state{t_1} + Z_t$,
	as well as
	\begin{eqnarray}
		\itointeg{t_1}{t_2}{\left( \instantcost{\policy_2}{t} - \instantcost{\policy_1}{t} \right)}{t}
		&=& \itointeg{t_1}{t_2}{ \left[ \left( \state{t}^{\policy_1}+\statetwo{t} \right)^{\top} M \left( \state{t}^{\policy_1}+\statetwo{t} \right) - {\state{t}^{\policy_1}}^{\top} M_\star \state{t}^{\policy_1} \right] }{t} \notag \\
		&=& \itointeg{t_1}{t_2}{ \left[ {\state{t}^{\policy_1}}^{\top} S \state{t}^{\policy_1} + 2 \statetwo{t}^{\top} M \state{t}^{\policy_1} + \statetwo{t}^{\top} M \statetwo{t} \right] }{t}. \label{LocalRegretLemProofEq1} ~~~~~~~~
		\end{eqnarray}
	
	On the other hand, for $t \geq t_2$, we have
	\begin{eqnarray}
		\itointeg{t_2}{T}{\left( \instantcost{\policy_2}{t} - \instantcost{\policy_1}{t} \right)}{t} 
		&=& \itointeg{t_2}{T}{ \left[ \left( \state{t}^{\policy_1}+\statetwo{t} \right)^{\top} M_\star \left( \state{t}^{\policy_1}+\statetwo{t} \right) - {\state{t}^{\policy_1}}^{\top} M_\star \state{t}^{\policy_1} \right] }{t} \notag \\
		&=& \itointeg{t_2}{T}{ \left[ 2 \statetwo{t}^{\top} M_\star \state{t}^{\policy_1} + \statetwo{t}^{\top} M_\star \statetwo{t} \right] }{t}. \label{LocalRegretLemProofEq2}
	\end{eqnarray}
	and  
	\begin{eqnarray}
	\state{t}^{\policy_i} &=& e^{\CLmat{\star}(t-t_2)} \state{t_2}^{\policy_i} + \itointeg{t_2}{t}{e^{\CLmat{\star}(t-s)} \BMcoeff{}}{\BM{s}}, \notag \\
	\statetwo{t} &=& e^{\CLmat{\star}(t-t_2)} \left[ \state{t_2}^{\policy_2} - \state{t_2}^{\policy_1} \right] = e^{\CLmat{\star}(t-t_2)} \left[ \erterm{t_2} \state{t_1} + Z_{t_2} \right]. \notag
	\end{eqnarray}
	Thus, putting \eqref{LocalRegretLemProofEq1} and \eqref{LocalRegretLemProofEq2} together, we obtain the desired results.
\end{proof}

\subsection{\bf Upper-bounding comparative ratios of stochastic integrals} \label{appC}
\begin{lemm} \label{SelfNormalizedLem}
	Suppose that $\statetwo{t} \in \R^{m}$ is a vector-valued stochastic process such that $\statetwo{t}$ is $\filter{t}$-measurable for the natural filtration $\filter{t} = \sigfield{ \left\{ \BM{s} \right\}_{0 \leq s \leq t}}$. Then, letting $\empiricalcovmat{t}=\itointeg{0}{t}{ \statetwo{s} \statetwo{s}^{\top} }{s}$, we have
	$$\Mnorm{ \left(I+ \empiricalcovmat{t} \right)^{-1/2} \itointeg{0}{t}{ \statetwo{s} }{\BM{s}^{\top}} }{2}^2 = \order{ m \noisedim \log \eigmax{\empiricalcovmat{t}} }.$$
\end{lemm}

\begin{proof}
	First, fix $t>0$, and for an arbitrary $\epsilon>0$, let $n=\lfloor t/\epsilon \rfloor$. So, for $k=0,1,\cdots, n$, consider the sequence of matrices $M_k = \epsilon^{-1} I + \sum\limits_{i=0}^k \statetwo{i\epsilon} \statetwo{i\epsilon}^{\top}$.
	Then, for $k=1, \cdots,n$, consider the sequence of scalars $\ssconstant_k$ defined according to $\ssconstant_k = \statetwo{k\epsilon}^{\top} M_{k-1}^{-1} \statetwo{k\epsilon}$.   
	Using the formula for determinants of the products of matrices, we have
	$$\det M_k = \det \left[ M_{k-1} \left( I + M_{k-1}^{-1}\statetwo{k\epsilon} \statetwo{k\epsilon}^{\top} \right) \right] = \det \left(M_{k-1}\right) \det \left( I + M_{k-1}^{-1}\statetwo{k\epsilon} \statetwo{k\epsilon}^{\top} \right).$$
	 
	Since all eigenvalues of $I + M_{k-1}^{-1}\statetwo{k\epsilon} \statetwo{k\epsilon}^{\top}$ are unit, except one of them which is $1+\ssconstant_k$, we have $\left({1+\ssconstant_k}\right) {\det M_{k-1}}= {\det M_{k}}$.
	On the other hand, matrix inversion formula gives
	\begin{equation*}
	M_k^{-1} = M_{k-1}^{-1} - \frac{1}{1+\statetwo{k\epsilon}^{\top} M_{k-1}^{-1} \statetwo{k\epsilon}} M_{k-1}^{-1} \statetwo{k\epsilon} \statetwo{k\epsilon}^{\top} M_{k-1}^{-1},
	\end{equation*}
	which leads to
	$$\statetwo{k\epsilon}^{\top} M_{k}^{-1} \statetwo{k\epsilon} = \statetwo{k\epsilon}^{\top} \left( M_{k-1} + \statetwo{k\epsilon} \statetwo{k\epsilon}^{\top} \right)^{-1} \statetwo{k\epsilon} = \ssconstant_k - \frac{\ssconstant_k^2}{1+\ssconstant_k} = 1- \frac{1}{1+\ssconstant_k} = 1-\frac{\det M_{k-1}}{\det M_{k}}.$$
	
	Further, by using the inequality $1-\ssconstant \leq -\log \ssconstant$ for $\ssconstant>0$, the latter equality gives
	\begin{equation} \label{SelfNormLemProofEq1}
	\statetwo{k\epsilon}^{\top} M_{k}^{-1} \statetwo{k\epsilon} \leq \log {\det M_{k}} - \log {\det M_{k-1}}.
	\end{equation}
	
	Now, let $F_k=\sum\limits_{i=0}^{k} \statetwo{i\epsilon} \left( \BM{(i+1)\epsilon} - \BM{i\epsilon} \right)^{\top}$. Using the facts that $\statetwo{k\epsilon},F_{k-1}$, and $M_k$ all are $\filter{k\epsilon}$-measurable, the Brownian motion $\BM{t}$ has independent increments, and its covariance matrix is a multiple of identity, properties of conditional expectations give
	\begin{eqnarray*}
		&& \E{F_{k}^{\top} M_k^{-1} F_{k} } 
		= \E{\E{F_{k}^{\top} M_k^{-1} F_{k} \Big| \filter{k\epsilon} }} \notag \\
		&=& \E{\E{ \left(F_{k-1} + \statetwo{k\epsilon} \left( \BM{(k+1)\epsilon} - \BM{k\epsilon} \right)^{\top} \right)^{\top} M_k^{-1} \left(F_{k-1} + \statetwo{k\epsilon} \left( \BM{(k+1)\epsilon} - \BM{k\epsilon} \right)^{\top} \right) \Big| \filter{k\epsilon} }} \notag \\
		&=& \E{ F_{k-1}^{\top} M_k^{-1} F_{k-1} + \E{ \left( \BM{(k+1)\epsilon} - \BM{k\epsilon} \right) \statetwo{k\epsilon}^{\top} M_k^{-1} \statetwo{k\epsilon} \left( \BM{(k+1)\epsilon} - \BM{k\epsilon} \right)^{\top} \Big| \filter{k\epsilon} } }  \notag \\
		&=& \E{ F_{k-1}^{\top} M_k^{-1} F_{k-1} + \left(\statetwo{k\epsilon}^{\top} M_k^{-1} \statetwo{k\epsilon}\right) \epsilon I  }. ~~~~~\label{AuxTableEq3}
	\end{eqnarray*}
	
	So, using \eqref{SelfNormLemProofEq1} together with the fact that (as positive semidefinite matrices) the order $M_{k-1} \leq M_k$ holds, we get the telescopic relationships
	$$\eigmax{\E{F_{k}^{\top} M_k^{-1} F_{k}}} - \eigmax{\E{F_{k-1}^{\top} M_{k-1}^{-1} F_{k-1}}} \leq \epsilon\left(\log \frac{\det \left(\epsilon M_{k}\right)}{\det \left(\epsilon M_{k-1}\right)}\right).$$
	
	Since $F_{k}^{\top} \left(M_k\right)^{-1} F_{k}$ is positive semidefinite, its trace is larger than its largest eigenvalue. Hence, adding up for $k=0,1,\cdots, n$, by interchanging trace and expectation, we obtain 
	$$\E{\eigmax{F_{n}^{\top} \left(M_n\right)^{-1} F_{n}}} \leq \E{\tr{F_{n}^{\top} \left( M_n\right)^{-1} F_{n}}} \leq \noisedim \eigmax{\E{F_{n}^{\top} \left( M_n\right)^{-1} F_{n}}},$$
	
	which, by $\epsilon M_0 \geq I$, leads to 
	$$\E{\eigmax{F_{n}^{\top} \left(\epsilon M_n\right)^{-1} F_{n}}} \leq m \noisedim \log \eigmax{\epsilon M_{n}}.$$
	
	Thus, according to Doob's Martingale Convergence Theorem~\citep{oksendal2013stochastic,baldi2017stochastic}, we have
	\begin{equation*}
	\Mnorm{\left(\epsilon M_n\right)^{-1/2} F_{n}}{2}^2 = \order{ m \noisedim \log \eigmax{\epsilon M_{n}}}.
	\end{equation*}
	Finally, letting $\epsilon \to 0$, we obtain the desired result, because $\epsilon M_n, F_n$ are $\epsilon$-approximations of the corresponding integrals. 
\end{proof}

\subsection{\bf Lipschitz continuity of optimal feedback} \label{appD}
\begin{lemm} \label{LipschitzLemma}
	Using the Jordan decomposition $\CLmat{\star}=\Amat{\star}+\Bmat{\star}\Optgain{\truth}=P_\star^{-1} \Lambda_\star P_\star$, define $\mult{\star}=\mult{\CLmat{\star}}$, similar to Definition~\ref{MultDef}, and suppose that $\Learnerror{}{\estpara{}} \leq \kappa_\star$, for 
	\begin{equation*}
		\kappa_\star = \frac{1}{ 1 \vee \Mnorm{\Optgain{\truth}}{2} } \left( \frac{ \left( -\mosteig{\CLmat{\star}} \right) \wedge \left( -\mosteig{\CLmat{\star}}\right)^{\mult{\star}} }{ \mult{\star}^{1/2} \Mnorm{P_\star^{-1}}{2} \Mnorm{P_\star}{2} } \wedge \left[ 4 \itointeg{0}{\infty}{ \Mnorm{e^{\CLmat{\star}t}}{2}^2 }{t} \right]^{-1} \right). \label{LipschitzLemmaEq1}
	\end{equation*}
	Then, letting
	\begin{equation*}
		\ssconstant_\star = \frac{2\Mnorm{\RiccSol{\truth}}{2}}{\eigmin{\Rmat}} \left[1 + \frac{4 \Mnorm{\Bmat{\star}}{2} }{\eigmin{\Qmat}} \Mnorm{\RiccSol{\truth}}{2} \left( 1 \vee \frac{ 2 \left( \Mnorm{\Bmat{\star}}{2} + \kappa_\star \right) \Mnorm{\RiccSol{\truth}}{2} }{\eigmin{\Rmat}} \right) \right], \label{LipschitzLemmaEq2}
	\end{equation*}
	 we have 
	\begin{equation*}
	\Mnorm{\Optgain{\estpara{}} - \Optgain{\truth} }{} \leq \ssconstant_\star \Learnerror{}{\estpara{}}.
	\end{equation*}
	In general, without the condition $\Learnerror{}{\estpara{}} \leq \kappa_\star$, the constant $\ssconstant_\star$ is replaced with 
	\begin{eqnarray*} 
	\ssconstant &=& \frac{\Mnorm{\RiccSol{\estpara{}}}{2}}{\eigmin{\Rmat}} + \frac{2 \Mnorm{\Bmat{\star}}{2} \Mnorm{ \RiccSol{\estpara{0}} }{2}^2 }{\eigmin{\Qmat}\eigmin{\Rmat}} \left( 1 \vee \frac{ \left( \Mnorm{\Bmat{\star}}{2} + \Learnerror{}{\estpara{}} \right) \Mnorm{\RiccSol{\estpara{0}}}{2} }{\eigmin{\Rmat}} \right) , \label{LipschitzLemmaEq3}
	\end{eqnarray*} 
	for some convex combination $\left[ \estpara{0} \right] = \eta \left[ \estpara{} \right] + (1-\eta) \left[\truth\right]$, and $0 \leq \eta \leq 1$.
	
\end{lemm}

\begin{proof}
	Fix the matrices $\estpara{}$, and consider the matrix-valued curve
	$$\curve=\left\{ (1-\eta) \left[ \truth \right] + \eta \left[ \estpara{} \right] \right\}_{0 \leq \eta \leq 1}.$$
	
	For an arbitrary $\estpara{0} \in \curve$, we find the derivative of the matrix $\RiccSol{\estpara{0}}$ at $\estpara{0}$, assuming that the matrices $\estpara{0}$ vary along $\curve$. For this purpose, letting $\erterm{A}=\estA{}-\Amat{\star}$, $\erterm{B}=\estB{}-\Bmat{\star}$, we first calculate $\RiccSol{\estpara{1}}$ for $\estA{1}=\estA{0}+ \eta \erterm{A}, \estB{1}=\estB{0}+ \eta \erterm{B}$, and then let $\eta \to 0$. First, letting $P=\RiccSol{\estpara{1}}-\RiccSol{\estpara{0}}$, we get
	\begin{eqnarray*}
		&& \RiccSol{\estpara{0}} \estB{1} \Rmat^{-1} \estB{1}^{\top} \RiccSol{\estpara{0}} \\
		&=& \eta \RiccSol{\estpara{0}} \erterm{B} \Rmat^{-1} \estB{1}^{\top} \RiccSol{\estpara{0}} 
		+ \RiccSol{\estpara{0}} \estB{0} \Rmat^{-1} \estB{1}^{\top} \RiccSol{\estpara{0}} \notag \\
		&=& \eta^2 \RiccSol{\estpara{0}} \erterm{B} \Rmat^{-1} \erterm{B}^{\top} \RiccSol{\estpara{0}} 
		+ \eta \RiccSol{\estpara{0}} \erterm{B} \Rmat^{-1} \estB{0}^{\top} \RiccSol{\estpara{0}} \\
		&+&
		\eta \RiccSol{\estpara{0}} \estB{0} \Rmat^{-1} \erterm{B}^{\top} \RiccSol{\estpara{0}} \notag \\
		&+&
		\RiccSol{\estpara{0}} \estB{0} \Rmat^{-1} \estB{0}^{\top} \RiccSol{\estpara{0}}. \label{LipschitzLemmaEq4}
	\end{eqnarray*} 
	
	The above expression, because of
	 \begin{eqnarray*}
	 	&& \RiccSol{\estpara{1}} \estB{1} \Rmat^{-1} \estB{1}^{\top} \RiccSol{\estpara{1}} \\
	 	&=& \RiccSol{\estpara{1}} \estB{1} \Rmat^{-1} \estB{1}^{\top} P + \RiccSol{\estpara{1}} \estB{1} \Rmat^{-1} \estB{1}^{\top} \RiccSol{\estpara{0}} \notag \\
	 	&=& P \estB{1} \Rmat^{-1} \estB{1}^{\top} P + \RiccSol{\estpara{0}} \estB{1} \Rmat^{-1} \estB{1}^{\top} P \\
	 	&+& P \estB{1} \Rmat^{-1} \estB{1}^{\top} \RiccSol{\estpara{0}}	
	 	+ \RiccSol{\estpara{0}} \estB{1} \Rmat^{-1} \estB{1}^{\top} \RiccSol{\estpara{0}}, \label{LipschitzLemmaEq5}
	 \end{eqnarray*}
	 
	 implies that the followings hold true:
	 \begin{eqnarray}
	 	&& \RiccSol{\estpara{1}} \estB{1} \Rmat^{-1} \estB{1}^{\top} \RiccSol{\estpara{1}} -  \RiccSol{\estpara{0}} \estB{0} \Rmat^{-1} \estB{0}^{\top} \RiccSol{\estpara{0}} \notag \\
	 	&=& P \estB{1} \Rmat^{-1} \estB{1}^{\top} P + \RiccSol{\estpara{0}} \estB{1} \Rmat^{-1} \estB{1}^{\top} P 
	 	+ P \estB{1} \Rmat^{-1} \estB{1}^{\top} \RiccSol{\estpara{0}} \notag\\
	 	&+&	\eta^2 \RiccSol{\estpara{0}} \erterm{B} \Rmat^{-1} \erterm{B}^{\top} \RiccSol{\estpara{0}} 
	 	+ \eta \RiccSol{\estpara{0}} \erterm{B} \Rmat^{-1} \estB{0}^{\top} \RiccSol{\estpara{0}} \notag \\
	 	&+& 
	 	\eta \RiccSol{\estpara{0}} \estB{0} \Rmat^{-1} \erterm{B}^{\top} \RiccSol{\estpara{0}} . \label{LipschitzLemmaEq6}
	 \end{eqnarray}
	 
	 By plugging~\eqref{LipschitzLemmaEq6} and 
	 \begin{eqnarray*}
	 	&& \estA{1}^{\top} \RiccSol{\estpara{1}} + \RiccSol{\estpara{1}} \estA{1} 
	 	= \estA{1}^{\top} \RiccSol{\estpara{0}} + \estA{1}^{\top} P + \RiccSol{\estpara{0}} \estA{1} + P \estA{1} \notag \\
	 	&=& \estA{0}^{\top} \RiccSol{\estpara{0}} + \eta \erterm{A}^{\top} \RiccSol{\estpara{0}} + \estA{1}^{\top} P 
	 	+ \RiccSol{\estpara{0}} \estA{0} + \eta \RiccSol{\estpara{0}} \erterm{A} + P \estA{1}, \label{LipschitzLemmaEq66}
	 \end{eqnarray*}
	 
	 in $\MatOpAve{\estpara{i}}{\RiccSol{\estpara{i}}}=0$ for $i=0,1$, we obtain
		\begin{eqnarray*}
			0 &=& \left[ \estA{1}^{\top} - \RiccSol{\estpara{0}} \estB{1} \Rmat^{-1} \estB{1}^{\top} \right] P + P \left[ \estA{1} - \estB{1} \Rmat^{-1} \estB{1}^{\top} \RiccSol{\estpara{0}} \right] - P \estB{1} \Rmat^{-1} \estB{1}^{\top} P \\
			&+& \eta \erterm{A}^{\top} \RiccSol{\estpara{0}} + \eta \RiccSol{\estpara{0}} \erterm{A}  - \eta^2 \RiccSol{\estpara{0}} \erterm{B} \Rmat^{-1} \erterm{B}^{\top} \RiccSol{\estpara{0}} \\
			&-& \eta \RiccSol{\estpara{0}} \erterm{B} \Rmat^{-1} \estB{0}^{\top} \RiccSol{\estpara{0}} - 
			\eta \RiccSol{\estpara{0}} \estB{0} \Rmat^{-1} \erterm{B}^{\top} \RiccSol{\estpara{0}},
		\end{eqnarray*}
		or equivalently,
	\begin{equation} \label{LipschitzLemProofEq1}
	0 = \auxA^{\top} P + P \auxA - P \estB{1} \Rmat^{-1} \estB{1}^{\top} P + \auxQ,
	\end{equation}	
	for $\auxA = \estA{1} - \estB{1} \Rmat^{-1} \estB{1}^{\top} \RiccSol{\estpara{0}}$, and 
	\begin{eqnarray*}
		\auxQ &=& \eta \erterm{A}^{\top} \RiccSol{\estpara{0}} + \eta \RiccSol{\estpara{0}} \erterm{A}  -	\eta^2 \RiccSol{\estpara{0}} \erterm{B} \Rmat^{-1} \erterm{B}^{\top} \RiccSol{\estpara{0}} \\
		&=& \eta \RiccSol{\estpara{0}} \Big[ \erterm{A} + \erterm{B} \Optgain{\estpara{0}} \Big] +
		\eta \Big[ \Optgain{\estpara{0}}^{\top} \erterm{B}^{\top} + \erterm{A}^{\top} \Big] \RiccSol{\estpara{0}} \\
		&-& \eta^2 \RiccSol{\estpara{0}} \erterm{B} \Rmat^{-1} \erterm{B}^{\top} \RiccSol{\estpara{0}}. \label{LipschitzLemmaEq7}
	\end{eqnarray*}
	
	Suppose that $\eta$ is sufficiently small so that $\mosteig{\auxA}<0$. Note that it is possible thanks to stabilizability of $\estpara{0}$, Theorem~\ref{OptimalityProof}, and $\lim\limits_{\eta \to 0} \auxA=\estA{0}+\estB{0}\Optgain{\estpara{0}}=\estD{0}$. So, since $P\estB{1}\Rmat^{-1}\estB{1}^{\top}P$ is a positive semidefinite matrix,  \eqref{LipschitzLemProofEq1} implies that
	\begin{eqnarray*}
		P &=& \itointeg{0}{\infty}{ e^{\auxA^{\top}t} \left( -P\estB{1}\Rmat^{-1}\estB{1}^{\top}P+\auxQ \right) e^{\auxA t} }{t} \\
		&\leq& \itointeg{0}{\infty}{ e^{\auxA^{\top}t} \auxQ e^{\auxA t} }{t} \leq \left( \Mnorm{\auxQ}{2} \itointeg{0}{\infty}{ \Mnorm{e^{\auxA t}}{2}^2 }{t} \right) I_{\statedim}.
	\end{eqnarray*}
	This, because of $\lim\limits_{\eta \to 0} \auxQ=0$, leads to $\lim\limits_{\eta \to 0} P=0$. Thus, letting $M=\erterm{A} + \erterm{B} \Optgain{\estpara{0}}$ and $\erterm{\RiccSol{\estpara{0}}}=\lim\limits_{\eta \to 0} \eta^{-1} P$, \eqref{LipschitzLemProofEq1} gives the following for $\erterm{\RiccSol{\estpara{0}}}$:
	\begin{equation} \label{LipschitzLemProofEq2}
	\itointeg{0}{\infty}{ e^{\estD{0}^{\top}t} \left( \RiccSol{\estpara{0}} M + M^{\top}  \RiccSol{\estpara{0}} \right) e^{\estD{0} t} }{t}.
	\end{equation}
	By
	$$\RiccSol{\estpara{}}-\RiccSol{\truth} = \itointeg{0}{1}{\erterm{(1-\eta) \left[ \truth \right] + \eta \left[ \estpara{} \right]}}{\eta},$$
	
	\eqref{StabThmProofEq3}, \eqref{LipschitzLemProofEq2}, and the Cauchy-Schwarz inequality provide
	\begin{eqnarray*}
		&& \Mnorm{\RiccSol{\estpara{}}-\RiccSol{\truth}}{2} \\
		&\leq& \Learnerror{}{\estpara{}} \sup\limits_{\left[\estpara{0}\right] \in \curve} 2 \Mnorm{\RiccSol{\estpara{0}}}{2} \left( 1 \vee \Mnorm{\Optgain{\estpara{0}}}{2} \right) \itointeg{0}{\infty}{ \Mnorm{e^{\estD{0}t}}{2}^2 }{t} \notag \\
		&\leq& \Learnerror{}{\estpara{}} \frac{2}{\eigmin{\Qmat}} \sup\limits_{\left[\estpara{0}\right] \in \curve} \Mnorm{\RiccSol{\estpara{0}}}{2}^2 \left( 1 \vee \Mnorm{\Optgain{\estpara{0}}}{2} \right) . \label{LipschitzLemmaEq77}
	\end{eqnarray*}
	
	Next, note that $\Learnerror{}{\estpara{}} \leq \kappa_\star$, together with~\eqref{StabRadiEq1} and \eqref{StabThmProofEq2}, implies that
	\begin{eqnarray*} 
	&& \Mnorm{\RiccSol{\estpara{}}-\RiccSol{\truth}}{2} \leq \Learnerror{}{\estpara{}} \frac{8\Mnorm{\RiccSol{\truth}}{2}^2 }{\eigmin{\Qmat}} \left( 1 \vee \frac{ 2\left( \Mnorm{\Bmat{\star}}{2} + \kappa_\star \right) \Mnorm{\RiccSol{\truth}}{2} }{\eigmin{\Rmat}} \right) . \label{LipschitzLemmaEq8} 
	\end{eqnarray*} 
	
	Therefore, using \eqref{StabThmProofEq2}, and putting the above inequality together with 
	\begin{eqnarray*} 
	&& \Mnorm{\Optgain{\estpara{}}- \Optgain{\truth}}{2} \\
	&=& \Mnorm{ \Rmat^{-1} \left[ \left( \Bmat{\star} - \estB{1} \right) \RiccSol{\estpara{}} + \Bmat{\star} \left( \RiccSol{\truth} - \RiccSol{\estpara{}} \right) \right] }{2} \notag \\
	&\leq& \Mnorm{\Rmat^{-1}}{2} \left[ \Mnorm{\Bmat{\star} - \estB{1}}{2} \Mnorm{\RiccSol{\estpara{}}}{2} + \Mnorm{\Bmat{\star}}{1} \Mnorm{\RiccSol{\truth} - \RiccSol{\estpara{}}}{2} \right] , \label{LipschitzLemmaEq9}
	\end{eqnarray*} 
	we get the first desired result. To establish the second result, it suffices to let $\estpara{0}$ be the one for which the above supremum over $\curve$ is achieved. 
\end{proof}

\subsection{\bf Effects of sub-optimal linear feedback policies}
\begin{lemm} \label{LyapLemma}
	Consider a noiseless linear dynamical system with the stabilizable dynamics matrices $\estpara{}$. That is, $\diff \state{t} = \left( \estA{} \state{t} + \estB{} \action{t} \right) \diff t$, starting from $\state{0}=x$. Then, if we apply the linear feedback 
	$$\policy: ~~~\action{t}=\Gainmat{}\state{t},$$
	as long as $\mosteig{\estA{}+\estB{}\Gainmat{}}<0$, it holds that
	\begin{eqnarray*}
		&& \itointeg{0}{\infty}{\instantcost{\policy}{t} }{t} = x^{\top} \RiccSol{\estpara{}} x 
		+ \itointeg{0}{\infty}{ \norm{\Rmat^{1/2} \left( \Gainmat{} - \Optgain{\estpara{}} \right) e^{\left( \estA{}+\estB{}\Gainmat{} \right)t} x }{}^2  }{t}  .
	\end{eqnarray*}
\end{lemm}
\begin{proof}
	Denote $\estD{1}=\estA{}+\estB{}\Optgain{\estpara{}}$ and $\CLmat{2}=\estA{}+\estB{}\Gainmat{}$. So, the dynamics equation $\diff \state{t} = \left( \estA{} \state{t} + \estB{} \Gainmat{}\state{t} \right) \diff t$ implies that $\state{t}=e^{\CLmat{2}t} x$, which leads to
	\begin{eqnarray*}
		&&\itointeg{0}{\infty}{\instantcost{\policy}{t} }{t} = \itointeg{0}{\infty}{ \state{t}^{\top} \left( \Qmat + \Gainmat{}^{\top} \Rmat \Gainmat{} \right) \state{t} }{t} \\
		&=& \itointeg{0}{\infty}{ x^{\top} e^{\CLmat{2}^{\top}t} \left( \Qmat + \Gainmat{}^{\top} \Rmat \Gainmat{} \right) e^{\CLmat{2}t} x }{t} = x^{\top} P x,
	\end{eqnarray*}
	where 
	\begin{eqnarray*}
		P &=& \itointeg{0}{\infty}{ e^{\CLmat{2}^{\top}t} \left( \Qmat + \Gainmat{}^{\top} \Rmat \Gainmat{} \right) e^{\CLmat{2}t} }{t} \\
		&=& \itointeg{0}{\epsilon}{ e^{\CLmat{2}^{\top}t} \left( \Qmat + \Gainmat{}^{\top} \Rmat \Gainmat{} \right) e^{\CLmat{2}t} }{t} \\
		&+& e^{\CLmat{2}^{\top}\epsilon} \left( \itointeg{0}{\infty}{ e^{\CLmat{2}^{\top}t} \left( \Qmat + \Gainmat{}^{\top} \Rmat \Gainmat{} \right) e^{\CLmat{2}t} }{t} \right) e^{\CLmat{2}\epsilon} \\
		&=& \itointeg{0}{\epsilon}{ e^{\CLmat{2}^{\top}t} \left( \Qmat + \Gainmat{}^{\top} \Rmat \Gainmat{} \right) e^{\CLmat{2}t} }{t} + e^{\CLmat{2}^{\top}\epsilon} P e^{\CLmat{2}\epsilon},
	\end{eqnarray*}
	which yields to
	\begin{eqnarray*}
		\Qmat + \Gainmat{}^{\top} \Rmat \Gainmat{}
		&=& \lim\limits_{\epsilon \to 0} \frac{1}{\epsilon} \itointeg{0}{\epsilon}{ e^{\CLmat{2}^{\top}t} \left( \Qmat + \Gainmat{}^{\top} \Rmat \Gainmat{} \right) e^{\CLmat{2}t} }{t} \\
		&=& \lim\limits_{\epsilon \to 0} \frac{1}{\epsilon} \left[ P - e^{\CLmat{2}^{\top}\epsilon} P + e^{\CLmat{2}^{\top}\epsilon} P - e^{\CLmat{2}^{\top}\epsilon} P e^{\CLmat{2}\epsilon} \right] \\
		&=& - \CLmat{2}^{\top} P - P \CLmat{2}.
	\end{eqnarray*}
	Similar to~\eqref{LyapDifferential}, it holds that $\estD{1}^{\top} \RiccSol{\estpara{}}+ \RiccSol{\estpara{}} \estD{1} + \Qmat + \Optgain{\estpara{}}^{\top} \Rmat \Optgain{\estpara{}}$. So, subtracting the latter two equalities, we get
	\begin{eqnarray} 
	&&\left( \CLmat{2} - \estD{1} \right)^{\top} \RiccSol{\estpara{}} + \RiccSol{\estpara{}} \left( \CLmat{2} - \estD{1} \right) \label{LyapLemmaProofEq1} \\
	&+& \CLmat{2}^{\top} \left( P- \RiccSol{\estpara{}} \right) + \left( P - \RiccSol{\estpara{}} \right) \CLmat{2} + S=0, \notag
	\end{eqnarray}
	where 
	\begin{equation*}
	S = \Gainmat{}^{\top} \Rmat \Gainmat{} - \Optgain{\estpara{}}^{\top} \Rmat \Optgain{\estpara{}}.
	\end{equation*}
	Because $\mosteig{\CLmat{2}} <0$, solving~\eqref{LyapLemmaProofEq1} for $P - \RiccSol{\estpara{}}$, and using the fact $\CLmat{2} - \estD{1} = \estB{} \left[ \Gainmat{} - \Optgain{\estpara{}} \right]$, we have
	\begin{equation*}
	P - \RiccSol{\estpara{}} = \itointeg{0}{\infty}{ e^{\CLmat{2}^{\top}t} F e^{\CLmat{2}t} }{t},
	\end{equation*}
	where
	\begin{equation*}
	F=S+ \left[ \Gainmat{} - \Optgain{\estpara{}} \right]^{\top} \estB{}^{\top} \RiccSol{\estpara{}} + \RiccSol{\estpara{}} \estB{} \left[ \Gainmat{} - \Optgain{\estpara{}} \right].
	\end{equation*}
	Then, using $\estB{}^{\top} \RiccSol{\estpara{}} = - \Rmat \Optgain{\estpara{}}$, after doing some algebra we obtain
	\begin{eqnarray} 
	S &+& \left[ \Gainmat{} - \Optgain{\estpara{}} \right]^{\top} \estB{}^{\top} \RiccSol{\estpara{}} + \RiccSol{\estpara{}} \estB{} \left[ \Gainmat{} - \Optgain{\estpara{}} \right] \notag\\
	&=& \left[ \Gainmat{} - \Optgain{\estpara{}} \right]^{\top} \Rmat \left[ \Gainmat{} - \Optgain{\estpara{}} \right]. \label{LyapAuxEq}
	\end{eqnarray}
	Thus, $P - \RiccSol{\estpara{}}$ is
	\begin{equation*}
	\itointeg{0}{\infty}{ e^{\CLmat{2}^{\top}t} \left[ \Gainmat{} - \Optgain{\estpara{}} \right]^{\top} \Rmat \left[ \Gainmat{} - \Optgain{\estpara{}} \right] e^{\CLmat{2}t} }{t} ,
	\end{equation*}
	which implies the desired result.
\end{proof}

\subsection{\bf Convergence of empirical covariance matrix of the state vectors}
\begin{lemm} \label{EmpCovLemma}
	Suppose that for $t \geq \episodetime{}$, the linear feedback $\Gainmat{}$ is applied to the system~\eqref{dynamics} such that $\mosteig{\CLmat{}}<0$, where $\CLmat{}=\Amat{\star}+\Bmat{\star}\Gainmat{}$. Then, we have
	\begin{eqnarray*}
		\lim\limits_{T \to \infty} \frac{1}{T} \itointeg{\episodetime{}}{\episodetime{}+T}{\state{t}\state{t}^{\top}}{t} = \itointeg{0}{\infty}{ e^{\CLmat{}s} \BMcoeff{}\BMcoeff{}^{\top} e^{\CLmat{}^{\top}s} }{s}.
	\end{eqnarray*}
\end{lemm}
\begin{proof}
	First, denote 
	\begin{eqnarray*}
		\empiricalcovmat{T} = \frac{1}{T} \itointeg{\episodetime{}}{\episodetime{}+T}{\state{t}\state{t}^{\top}}{t}.
	\end{eqnarray*}
	Then, define the matrix $\statetwo{t}=\state{t}\state{t}^{\top}$, and apply Ito's Formula~\citep{baldi2017stochastic} to find $\diff \statetwo{t}$:
	\begin{eqnarray*}
		\diff \statetwo{t} = \diff \state{t}\state{t}^{\top} + \state{t} \diff \state{t}^{\top} + \diff \state{t} \diff \state{t}^{\top}.
	\end{eqnarray*}
	Plugging in for $\diff \state{t}$ from~\eqref{dynamics}, we obtain
	\begin{eqnarray*}
		\diff \statetwo{t} &=& \left( \CLmat{} \state{t} \diff t + \BMcoeff{} \diff \BM{t} \right) \state{t}^{\top} \\
		&+& \state{t} \left( \CLmat{} \state{t} \diff t + \BMcoeff{} \diff \BM{t} \right)^{\top} + \BMcoeff{} \BMcoeff{}^{\top} \diff t,
	\end{eqnarray*}
	where we used the facts $\diff t \diff t =0$, $\diff\BM{t} \diff t=0$, and Ito Isometry $\diff \BM{t} \diff \BM{t}^{\top}=\diff t I_{\noisedim}$ \citep{baldi2017stochastic}. Thus, we have
	\begin{eqnarray*}
		\statetwo{\episodetime{}+T} - \statetwo{\episodetime{}} = \itointeg{\episodetime{}}{\episodetime{}+T}{ \diff \statetwo{t} }{t} 
		&=& \itointeg{\episodetime{}}{\episodetime{}+T}{ \left(\CLmat{} \state{t} \state{t}^{\top} + \state{t} \state{t}^{\top} \CLmat{}^{\top} + \BMcoeff{} \BMcoeff{}^{\top}\right) } {t} + T M_{\episodetime{},T} ,
	\end{eqnarray*}
	where 
	\begin{eqnarray*}
		M_{\episodetime{},T} = \frac{1}{T} \itointeg{\episodetime{}}{\episodetime{}+T}{ \state{t} }{\BM{t}^{\top}}  \BMcoeff{}^{\top} + \frac{1}{T} \left( \itointeg{\episodetime{}}{\episodetime{}+T}{ \state{t} }{\BM{t}^{\top}}  \BMcoeff{}^{\top} \right)^{\top}.
	\end{eqnarray*}
	This can equivalently be written as 
	\begin{eqnarray*}
		\frac{1}{T} \left(\state{\episodetime{}+T} \state{\episodetime{}+T}^{\top} - \state{\episodetime{}} \state{\episodetime{}}^{\top}\right)
		&=& \CLmat{} \empiricalcovmat{T} + \empiricalcovmat{T} \CLmat{}^{\top} + \BMcoeff{}\BMcoeff{}^{\top} + M_{\episodetime{},T}.
	\end{eqnarray*}
	Since $\mosteig{\CLmat{}}<0$, the latter equality implies that $\empiricalcovmat{T} $ is
	{\small \begin{eqnarray*}
			\itointeg{0}{\infty}{ e^{\CLmat{}s} \left( \BMcoeff{}\BMcoeff{}^{\top} + M_{\episodetime{},T} +  \frac{1}{T}\state{\episodetime{}} \state{\episodetime{}}^{\top}  -  \frac{1}{T} \state{\episodetime{}+T} \state{\episodetime{}+T}^{\top}  \right)e^{\CLmat{}^{\top}s} }{s}.
	\end{eqnarray*}}
	Now, according to the following statements, the above leads to the desired result, because the terms corresponding to $M_{\episodetime{},T}, \state{\episodetime{}}, \state{\episodetime{}+T}$ vanish as $T$ grows. 	
	\begin{enumerate}
		\item 
		Clearly, it holds that $\lim\limits_{T \to \infty} T^{-1/2} \norm{\state{\episodetime{}}}{2} =0$.
		\item 
		Since $\mosteig{\CLmat{}}<0$, the expression 
		\begin{eqnarray*}
			\state{\episodetime{}+T} = e^{\CLmat{}T} \state{\episodetime{}} + \itointeg{\episodetime{}}{\episodetime{}+T}{e^{\CLmat{}(\episodetime{}+T-s)} \BMcoeff{}}{\BM{s}}
		\end{eqnarray*}
		implies that $\lim\limits_{T \to \infty} T^{-1/2} \norm{\state{\episodetime{}+T}}{2} =0$.
		\item 
		Putting $\mosteig{\CLmat{}}<0$ together with Doob's Martingale Convergence Theorem~\citep{oksendal2013stochastic,baldi2017stochastic}, we get
		$\lim\limits_{T \to \infty} M_{\episodetime{},T} =0$.
	\end{enumerate}
\end{proof}

\subsection{\bf Manifolds of dynamical systems with equal optimal feedback matrices} \label{appF6}
\begin{lemm} \label{OptManifoldLemma}
	Consider the set of dynamics matrices $\estpara{}$ that share optimal feedback with $\estpara{0}$:
	\begin{eqnarray*}
		\manifold{M}_0 = \left\{ \left[ \estpara{} \right] \in \R^{\statedim \times \left(\statedim + \controldim\right)} : \Optgain{\estpara{}}=\Optgain{\estpara{0}} \right\}.
	\end{eqnarray*} 
	Then, $\manifold{M}_0$ is a manifold of dimension $\statedim^2$.
\end{lemm}
\begin{proof}
	Suppose that for the matrix $\left[ \estpara{} \right] = \left[ \estpara{0} \right] + \epsilon \left[M,N\right]$, it holds that $\Optgain{\estpara{}}=\Optgain{\estpara{0}}$. We find the derivative of $\Optgain{\estpara{0}}$ along the direction $\left[M,N\right]$. First, using the expressions in \eqref{LyapDifferential} for $\estpara{}$ and for $\estpara{0}$, we get
	\begin{eqnarray*}
		&&\left(\CLmat{0}+\epsilon M+\epsilon N\Optgain{\estpara{0}}\right)^{\top} \RiccSol{\estpara{}} \\
		&+& \RiccSol{\estpara{}} \left(\CLmat{0}+\epsilon M+\epsilon N\Optgain{\estpara{0}}\right) \\
		&=& - \Qmat - \Optgain{\estpara{0}}^{\top} \Rmat \Optgain{\estpara{0}} \\
		&=& \CLmat{0}^{\top} \RiccSol{\estpara{0}} + \RiccSol{\estpara{0}} \CLmat{0} ,
	\end{eqnarray*}
	where $\CLmat{0}=\estA{0}+\estB{0}\Optgain{\estpara{0}}$. Simplifying the above expressions and letting $\epsilon \to 0$, for the matrix
	\begin{eqnarray*}
		\erterm{}=\lim\limits_{\epsilon \to 0} \epsilon^{-1} \left( \RiccSol{\estpara{}} - \RiccSol{\estpara{0}} \right),
	\end{eqnarray*}
	we have
	\begin{eqnarray*}
		&& \CLmat{0}^{\top} \erterm{} + \erterm{} \CLmat{0} + \left(M+ N\Optgain{\estpara{0}}\right)^{\top} \RiccSol{\estpara{0}} \\
		&+& \RiccSol{\estpara{0}} \left(M+ N\Optgain{\estpara{0}}\right) =0.
	\end{eqnarray*}
	Thus, since according to Theorem~\ref{OptimalityProof}, $\mosteig{\CLmat{0}}<0$, it yields to
	\begin{eqnarray*}
		\erterm{} = \itointeg{0}{\infty}{ e^{\CLmat{0}^{\top}t} F e^{\CLmat{0}t} }{t},
	\end{eqnarray*}
	where
	\begin{eqnarray*}
		F &=& \left(M+ N\Optgain{\estpara{0}}\right)^{\top} \RiccSol{\estpara{0}} \\
		&+& \RiccSol{\estpara{0}} \left(M+ N\Optgain{\estpara{0}}\right).
	\end{eqnarray*}
	On the other hand, $\Optgain{\estpara{}}=-\Rmat^{-1} \estB{}^{\top} \RiccSol{\estpara{}}$ gives 
	\begin{eqnarray*}
		0 &=& \lim\limits_{\epsilon \to 0} \frac{1}{\epsilon} \left(\estB{}^{\top} \RiccSol{\estpara{}} - \estB{0}^{\top} \RiccSol{\estpara{0}}\right) \\
		&=& \lim\limits_{\epsilon \to 0} \frac{1}{\epsilon} \Bigg[ \left(\estB{}^{\top} \right. - \left. \estB{0}^{\top}\right) \RiccSol{\estpara{}}  \\
		&-& \estB{0}^{\top} \left(\RiccSol{\estpara{0}} - \RiccSol{\estpara{}}\right)\Bigg] \\
		&=& N^{\top} \RiccSol{\estpara{0}} + \estB{0}^{\top} \erterm{}.
	\end{eqnarray*}
	So, $\manifold{M}_0$ is a manifold, and its tangent space consists of matrices $M,N$ satisfying the above equation. To find the dimension, select a $\statedim \times \statedim$ matrix $P$ arbitrarily, and let $N$ be 
	\begin{eqnarray} \label{SuppAuxEq1}
		N = -\RiccSol{\estpara{0}}^{-1} \itointeg{0}{\infty}{ e^{\CLmat{0}^{\top}t} \left[ P^{\top} \RiccSol{\estpara{0}} + \RiccSol{\estpara{0}} P  \right] e^{\CLmat{0}t} \estB{0} }{t} =0.
	\end{eqnarray}
	
	Note that since $\eigmin{\Qmat}>0$, the inverse $\RiccSol{\estpara{0}}^{-1}$ exists. Then, solve for $M$ according to $M+N \Optgain{\estpara{0}}=P$. Therefore, the matrices $M,N$ satisfy in $N^{\top} \RiccSol{\estpara{0}} + \estB{0}^{\top} \erterm{}=0$, and so correspond to a member of $\manifold{M}_0$. Conversely, every matrices $M,N$ in the tangent space of $\manifold{M}_0$ provide a $\statedim \times \statedim$ matrix $P=M+N \Optgain{\estpara{0}}$ such that $N^{\top} \RiccSol{\estpara{0}} + \estB{0}^{\top} \erterm{}=0$. Thus, $\manifold{M}_0$ is of dimension $\statedim^2$, which is the desired result.
\end{proof}


\begin{thebibliography}{10}	
	\bibitem{gillespie2007stochastic}
	D.~T. Gillespie, ``Stochastic simulation of chemical kinetics,'' \emph{Annu.
		Rev. Phys. Chem.}, vol.~58, pp. 35--55, 2007.
	
	\bibitem{schmidli2007stochastic}
	H.~Schmidli, \emph{Stochastic control in insurance}.\hskip 1em plus 0.5em minus
	0.4em\relax Springer Science \& Business Media, 2007.
	
	\bibitem{pham2009continuous}
	H.~Pham, \emph{Continuous-time stochastic control and optimization with
		financial applications}.\hskip 1em plus 0.5em minus 0.4em\relax Springer
	Science \& Business Media, 2009, vol.~61.
	
	\bibitem{lawrence2010learning}
	N.~D. Lawrence, M.~Girolami, M.~Rattray, and G.~Sanguinetti, \emph{Learning and
		inference in computational systems biology}.\hskip 1em plus 0.5em minus
	0.4em\relax MIT press, 2010.
	
	\bibitem{abbasi2011regret}
	Y.~Abbasi-Yadkori and C.~Szepesv{\'a}ri, ``Regret bounds for the adaptive
	control of linear quadratic systems,'' in \emph{Proceedings of the 24th
		Annual Conference on Learning Theory}.\hskip 1em plus 0.5em minus 0.4em\relax
	JMLR Workshop and Conference Proceedings, 2011, pp. 1--26.
	
	\bibitem{abeille2018improved}
	M.~Abeille and A.~Lazaric, ``Improved regret bounds for thompson sampling in
	linear quadratic control problems,'' in \emph{International Conference on
		Machine Learning}.\hskip 1em plus 0.5em minus 0.4em\relax PMLR, 2018, pp.
	1--9.
	
	\bibitem{ouyang2019posterior}
	Y.~Ouyang, M.~Gagrani, and R.~Jain, ``Posterior sampling-based reinforcement
	learning for control of unknown linear systems,'' \emph{IEEE Transactions on
		Automatic Control}, vol.~65, no.~8, pp. 3600--3607, 2019.
	
	\bibitem{faradonbeh2019applications}
	M.~K.~S. Faradonbeh, A.~Tewari, and G.~Michailidis, ``On applications of
	bootstrap in continuous space reinforcement learning,'' in \emph{2019 IEEE
		58th Conference on Decision and Control (CDC)}.\hskip 1em plus 0.5em minus
	0.4em\relax IEEE, 2019, pp. 1977--1984.
	
	\bibitem{faradonbeh2020optimism}
	------, ``Optimism-based adaptive regulation of linear-quadratic systems,''
	\emph{IEEE Transactions on Automatic Control}, vol.~66, no.~4, pp.
	1802--1808, 2020.
	
	\bibitem{dean2020sample}
	S.~Dean, H.~Mania, N.~Matni, B.~Recht, and S.~Tu, ``On the sample complexity of
	the linear quadratic regulator,'' \emph{Foundations of Computational
		Mathematics}, vol.~20, no.~4, pp. 633--679, 2020.
	
	\bibitem{faradonbeh2020adaptive}
	M.~K.~S. Faradonbeh, A.~Tewari, and G.~Michailidis, ``On adaptive
	linear--quadratic regulators,'' \emph{Automatica}, vol. 117, p. 108982, 2020.
	
	\bibitem{faradonbeh2020input}
	------, ``Input perturbations for adaptive control and learning,''
	\emph{Automatica}, vol. 117, p. 108950, 2020.
	
	\bibitem{cassel2020logarithmic}
	A.~Cassel, A.~Cohen, and T.~Koren, ``Logarithmic regret for learning linear
	quadratic regulators efficiently,'' in \emph{International Conference on
		Machine Learning}.\hskip 1em plus 0.5em minus 0.4em\relax PMLR, 2020, pp.
	1328--1337.
	
	\bibitem{ziemann2020phase}
	I.~Ziemann and H.~Sandberg, ``On a phase transition of regret in linear
	quadratic control: The memoryless case,'' \emph{IEEE Control Systems
		Letters}, vol.~5, no.~2, pp. 695--700, 2020.
	
	\bibitem{ziemann2020regret}
	------, ``Regret lower bounds for unbiased adaptive control of linear quadratic
	regulators,'' \emph{IEEE Control Systems Letters}, vol.~4, no.~3, pp.
	785--790, 2020.
	
	\bibitem{asghari2020regret}
	S.~M. Asghari, Y.~Ouyang, and A.~Nayyar, ``Regret bounds for decentralized
	learning in cooperative multi-agent dynamical systems,'' in \emph{Conference
		on Uncertainty in Artificial Intelligence}.\hskip 1em plus 0.5em minus
	0.4em\relax PMLR, 2020, pp. 121--130.
	
	\bibitem{lale2020logarithmic}
	S.~Lale, K.~Azizzadenesheli, B.~Hassibi, and A.~Anandkumar, ``Logarithmic
	regret bound in partially observable linear dynamical systems,'' \emph{arXiv
		preprint arXiv:2003.11227}, 2020.
	
	\bibitem{mandl1988consistency}
	P.~Mandl, T.~E. Duncan, and B.~Pasik-Duncan, ``On the consistency of a least
	squares identification procedure,'' \emph{Kybernetika}, vol.~24, no.~5, pp.
	340--346, 1988.
	
	\bibitem{mandl1989consistency}
	P.~Mandl, ``Consistency of estimators in controlled systems,'' in
	\emph{Stochastic Differential Systems}.\hskip 1em plus 0.5em minus
	0.4em\relax Springer, 1989, pp. 227--234.
	
	\bibitem{duncan1990adaptive}
	T.~E. Duncan and B.~Pasik-Duncan, ``Adaptive control of continuous-time linear
	stochastic systems,'' \emph{Mathematics of Control, signals and systems},
	vol.~3, no.~1, pp. 45--60, 1990.
	
	\bibitem{duncan1992least}
	T.~E. Duncan, P.~Mandl, and B.~Pasik-Duncan, ``On least squares estimation in
	continuous time linear stochastic systems,'' \emph{Kybernetika}, vol.~28,
	no.~3, pp. 169--180, 1992.
	
	\bibitem{duncan1999adaptive}
	T.~E. Duncan, L.~Guo, and B.~Pasik-Duncan, ``Adaptive continuous-time linear
	quadratic gaussian control,'' \emph{IEEE Transactions on automatic control},
	vol.~44, no.~9, pp. 1653--1662, 1999.
	
	\bibitem{caines2019stochastic}
	P.~E. Caines and D.~Levanony, ``Stochastic $\varepsilon$-optimal linear
	quadratic adaptation: An alternating controls policy,'' \emph{SIAM Journal on
		Control and Optimization}, vol.~57, no.~2, pp. 1094--1126, 2019.
	
	\bibitem{doya2000reinforcement}
	K.~Doya, ``Reinforcement learning in continuous time and space,'' \emph{Neural
		computation}, vol.~12, no.~1, pp. 219--245, 2000.
	
	\bibitem{rizvi2018output}
	S.~A.~A. Rizvi and Z.~Lin, ``Output feedback reinforcement learning control for
	the continuous-time linear quadratic regulator problem,'' in \emph{2018
		Annual American Control Conference (ACC)}.\hskip 1em plus 0.5em minus
	0.4em\relax IEEE, 2018, pp. 3417--3422.
	
	\bibitem{wang2020reinforcement}
	H.~Wang, T.~Zariphopoulou, and X.~Y. Zhou, ``Reinforcement learning in
	continuous time and space: A stochastic control approach.'' \emph{J. Mach.
		Learn. Res.}, vol.~21, pp. 198--1, 2020.
	
	\bibitem{basei2021logarithmic}
	M.~Basei, X.~Guo, A.~Hu, and Y.~Zhang, ``Logarithmic regret for episodic
	continuous-time linear-quadratic reinforcement learning over a finite-time
	horizon,'' \emph{Available at SSRN 3848428}, 2021.
	
	\bibitem{shirani2022thompson}
	M.~K. Shirani~Faradonbeh, M.~S. Shirani~Faradonbeh, and M.~Bayati, ``Thompson
	sampling efficiently learns to control diffusion processes,'' \emph{Advances
		in Neural Information Processing Systems}, vol.~35, pp. 3871--3884, 2022.
	
	\bibitem{yong1999stochastic}
	J.~Yong and X.~Y. Zhou, \emph{Stochastic controls: Hamiltonian systems and HJB
		equations}.\hskip 1em plus 0.5em minus 0.4em\relax Springer Science \&
	Business Media, 1999, vol.~43.
	
	\bibitem{oksendal2013stochastic}
	B.~Oksendal, \emph{Stochastic differential equations: an introduction with
		applications}.\hskip 1em plus 0.5em minus 0.4em\relax Springer Science \&
	Business Media, 2013.
	
	\bibitem{baldi2017stochastic}
	P.~Baldi, \emph{Stochastic Calculus: An Introduction Through Theory and
		Exercises}.\hskip 1em plus 0.5em minus 0.4em\relax Springer, 2017.
	
	\bibitem{chen1995linear}
	G.~Chen, G.~Chen, and S.-H. Hsu, \emph{Linear stochastic control
		systems}.\hskip 1em plus 0.5em minus 0.4em\relax CRC press, 1995, vol.~3.
	
	\bibitem{bauer1960norms}
	F.~L. Bauer and C.~T. Fike, ``Norms and exclusion theorems,'' \emph{Numerische
		Mathematik}, vol.~2, no.~1, pp. 137--141, 1960.
	
	\bibitem{caines1992continuous}
	P.~Caines, ``Continuous time stochastic adaptive control: non-explosion,
	$\varepsilon$-consistency and stability,'' \emph{Systems \& control letters},
	vol.~19, no.~3, pp. 169--176, 1992.
	
	\bibitem{faradonbeh2018bfinite}
	M.~K.~S. Faradonbeh, A.~Tewari, and G.~Michailidis, ``Finite-time adaptive
	stabilization of linear systems,'' \emph{IEEE Transactions on Automatic
		Control}, vol.~64, no.~8, pp. 3498--3505, 2018.
	
	\bibitem{faradonbeh2019randomized}
	------, ``Randomized algorithms for data-driven stabilization of stochastic
	linear systems,'' in \emph{2019 IEEE Data Science Workshop (DSW)}.\hskip 1em
	plus 0.5em minus 0.4em\relax IEEE, 2019, pp. 170--174.
	
	\bibitem{lale2020explore}
	S.~Lale, K.~Azizzadenesheli, B.~Hassibi, and A.~Anandkumar, ``Explore more and
	improve regret in linear quadratic regulators,'' \emph{arXiv preprint
		arXiv:2007.12291}, 2020.
	
	\bibitem{chen2021black}
	X.~Chen and E.~Hazan, ``Black-box control for linear dynamical systems,'' in
	\emph{Conference on Learning Theory}.\hskip 1em plus 0.5em minus 0.4em\relax
	PMLR, 2021, pp. 1114--1143.
	
	\bibitem{gramlich2022fast}
	D.~Gramlich and C.~Ebenbauer, ``Fast identification and stabilization of
	unknown linear systems,'' \emph{arXiv preprint arXiv:2208.10392}, 2022.
	
	\bibitem{faradonbeh2021bayesian}
	M.~K.~S. Faradonbeh and M.~S.~S. Faradonbeh, ``Bayesian algorithms learn to
	stabilize unknown continuous-time systems,'' \emph{IFAC-PapersOnLine},
	vol.~55, no.~12, pp. 377--382, 2022.
	
	\bibitem{levanony2001persistent}
	D.~Levanony and P.~E. Caines, ``On persistent excitation for linear systems
	with stochastic coefficients,'' \emph{SIAM journal on control and
		optimization}, vol.~40, no.~3, pp. 882--897, 2001.
	
	\bibitem{subrahmanyam2019identification}
	A.~Subrahmanyam and G.~P. Rao, \emph{Identification of Continuous-time Systems:
		Linear and Robust Parameter Estimation}.\hskip 1em plus 0.5em minus
	0.4em\relax CRC Press, 2019.
	
	\bibitem{bosworth1992linearized}
	J.~T. Bosworth, \emph{Linearized aerodynamic and control law models of the
		X-29A airplane and comparison with flight data}.\hskip 1em plus 0.5em minus
	0.4em\relax National Aeronautics and Space Administration, Office of
	Management~…, 1992, vol. 4356.
	
	\bibitem{kumar2015stochastic}
	P.~R. Kumar and P.~Varaiya, \emph{Stochastic systems: Estimation,
		identification, and adaptive control}.\hskip 1em plus 0.5em minus 0.4em\relax
	SIAM, 2015.
	
	\bibitem{chan1984convergence}
	S.~Chan, G.~Goodwin, and K.~Sin, ``Convergence properties of the riccati
	difference equation in optimal filtering of nonstabilizable systems,''
	\emph{IEEE Transactions on Automatic Control}, vol.~29, no.~2, pp. 110--118,
	1984.
	
	\bibitem{de1986riccati}
	C.~De~Souza, M.~Gevers, and G.~Goodwin, ``Riccati equations in optimal
	filtering of nonstabilizable systems having singular state transition
	matrices,'' \emph{IEEE Transactions on Automatic control}, vol.~31, no.~9,
	pp. 831--838, 1986.
	
\end{thebibliography}
\end{document}